\RequirePackage[l2tabu,orthodox]{nag}
\documentclass
[11pt,letterpaper]
{article} 

\usepackage[utf8]{inputenc}
\usepackage{amsmath}
\usepackage{amsthm}
\usepackage{amsfonts}
\usepackage{amssymb}
\usepackage{pgfplots}
\usepackage{systeme}
\usepackage{pgf,tikz}
\usepackage{mathrsfs}
\usepackage{extarrows}
\usetikzlibrary{arrows}
\usepackage{tikz-3dplot}
\usepackage{pgfplots}
\pgfplotsset{compat=1.8}
\usepackage[parfill]{parskip}
\usepackage{xcolor}
\usepackage{bm}
\usepackage{turnstile}

\usepackage{hyperref}
\usepackage[capitalise,nameinlink]{cleveref}
\crefname{algocf}{Algorithm}{Algorithms}
\Crefname{algocf}{Algorithm}{Algorithms}

\usepackage[
letterpaper,
top=1.2in,
bottom=1.2in,
left=1in,
right=1in]{geometry}

\usepackage[notes=true,later=false,camera=false]{dtrt}

\usepackage[linesnumbered, ruled, vlined]{algorithm2e}

\newcommand{\matbeg}{\left(\begin{array}}
\newcommand{\matend}{\end{array}\right)}

\usepackage{eurosym}
\DeclareRobustCommand{\officialeuro}{%
  \ifmmode\expandafter\text\fi
  
{\fontencoding{U}\fontfamily{eurosym}\selectfont e}}

\usepackage{mdframed}

\providecommand*{\rank}{\mathrm{rank}}

\providecommand{\TV}{\mathrm{TV}}

\providecommand*{\sststile}[2]{\vdash^{\raise{4mu}{\mkern-8mu #2}}_{\raise{-4mu}{\mkern-8mu #1}}}

\providecommand*{\proves}[2]{\sststile{#2}{#1}}

\providecommand*{\N}{\mathbb N}
\providecommand*{\R}{\mathbb R}

\providecommand*{\Tr}{\mathop{\mathrm{Tr}}}

\providecommand*{\E}{\mathop{\mathbb{E}}}

\providecommand*{\bm}[1]{\boldsymbol{#1}}

\providecommand*{\qed}{\blacksquare}

\renewcommand{\Pr}{\mathbb P}
\providecommand*{\cA}{{\mathcal A}}
\providecommand*{\cB}{{\mathcal B}}
\providecommand*{\cC}{{\mathcal C}}

\providecommand*{\Paren}[1]{\left(#1\right)}

\providecommand*{\Abs}[1]{\left\lvert#1\right\rvert}

\providecommand*{\Set}[1]{\left\{#1\right\}}

\providecommand*{\norm}[1]{\lVert#1\rVert}
\providecommand*{\Norm}[1]{\left\lVert#1\right\rVert}

\providecommand*{\iprod}[1]{\langle#1\rangle}
\providecommand*{\Iprod}[1]{\left\langle#1\right\rangle}

\def \poly {\mathrm{poly}}
\def \polylog {\mathrm{polylog}}

\renewcommand{\epsilon}{\varepsilon}

\newtheorem{theorem}{Theorem}[section]
\newtheorem{definition}[theorem]{Definition}

\newtheorem{lemma}[theorem]{Lemma}
\newtheorem{corollary}[theorem]{Corollary}
\newtheorem{claim}[theorem]{Claim}
\newtheorem{remark}[theorem]{Remark}
\newtheorem{fact}[theorem]{Fact}
\newtheorem{proposition}[theorem]{Proposition}

\providecommand*{\iprod}[1]{\langle#1\rangle}
\providecommand*{\Iprod}[1]{\left\langle#1\right\rangle}

\usepackage{fancyhdr}

\begin{document}

\pagestyle{plain}
\setcounter{page}{1}

\title{Dimension Reduction via Sum-of-Squares and Improved Clustering Algorithms for Non-Spherical Mixtures}
\newcommand{\FormatAuthor}[3]{
\begin{tabular}{c}
#1 \\ {\small\texttt{#2}} \\ {\small #3}
\end{tabular}
}
\author{
\begin{tabular}[h!]{ccc}
  \FormatAuthor{Prashanti Anderson}{paanders@csail.mit.edu}{MIT}&
  \FormatAuthor{Mitali Bafna\thanks{Most work on this paper was completed when the author was at MIT.}}{mitalib@cs.washington.edu}{University of Washington}&
  \FormatAuthor{Rares-Darius Buhai\thanks{Most work on this paper was completed when the author was at ETH Zurich.}}{rares-darius.buhai@epfl.ch}{EPFL}
\end{tabular}\\\\
\begin{tabular}[h!]{cc}
  \FormatAuthor{Pravesh K.\ Kothari}{kothari@cs.princeton.edu}{Princeton University}&
  \FormatAuthor{David Steurer}{dsteurer@inf.ethz.ch}{ETH Zurich}
\end{tabular}
}
\maketitle
\begin{abstract}
We develop a new approach for clustering \emph{non-spherical} (i.e., arbitrary component covariances) Gaussian mixture models via a subroutine, based on the sum-of-squares method, that finds a low-dimensional separation-preserving projection of the input data. Our method gives a non-spherical analog of the classical dimension reduction, based on singular value decomposition, that, among several other applications, forms a key component of the celebrated spherical clustering algorithm of Vempala and Wang~\cite{vempalawang}.

As applications, we obtain an algorithm to (1) cluster an arbitrary total-variation separated mixture of $k$ \emph{centered} (i.e., zero-mean) Gaussians with $n\geq \poly(d) f(w_{\min}^{-1})$ samples and $\poly(n)$ time, and (2) cluster an arbitrary total-variation separated mixture of $k$ Gaussians with identical but arbitrary unknown covariance with $n \geq d^{O(\log w_{\min}^{-1})} f(w_{\min}^{-1})$ samples and $n^{O(\log w_{\min}^{-1})}$ time. Here, $w_{\min}$ is the minimum mixing weight of the input mixture, and $f$ does not depend on the dimension $d$. Our algorithms naturally extend to tolerating a dimension-independent fraction of arbitrary outliers. Before this work, the techniques in the state-of-the-art non-spherical clustering algorithms needed $d^{O(k)} f(w_{\min}^{-1})$ samples and time for clustering such mixtures. 

Our results may come as a surprise in the context of the $d^{\Omega(k)}$ statistical query and sum-of-squares lower bounds~\cite{MR3734219-Diakonikolas17, MR4849259-Diakonikolas24} for clustering non-spherical Gaussian mixtures. While these results are usually thought to rule out $d^{o(k)}$ cost algorithms for the problem, our results show that the lower bounds can in fact be circumvented for a remarkably general class of Gaussian mixtures.
\end{abstract} 

\clearpage

\maketitle
\tableofcontents
\newpage
\newcommand{\cN}{\mathcal{N}}
\section{Introduction}
In this work, we revisit the problem of clustering a high-dimensional Gaussian mixture model.
Specifically, we are interested in algorithms that take i.i.d. samples drawn from a mixture of $k$ $d$-dimensional Gaussian distributions $\sum_{i \leq k} w_i N(\mu_i, \Sigma_i)$ as input, where $k\ll d$.
The goal is to output a good clustering of the data: partitioning the points based on the component that generated them, with at most $1\%$ misclassified points.
Such a goal is statistically possible only if the component Gaussians are pairwise well-separated, i.e., have a total variation distance of at least $1-O(w_{\min})$ (or ``overlap'' $\ll w_{\min}$) between them.\footnote{See~\cite{bakshi2020outlierrobust,dhkk20} for a detailed discussion.}

The Gaussian mixture model (GMM) has a history beginning more than a century ago in the work of Pearson~\cite{pearson1894contributions}.
In the last 50 years, the model has been extensively utilized in applications in science and engineering.
Starting with~\cite{dasgupta-1999}, it has become a central model studied in algorithmic statistics.
And, recently, an exciting line of work has even found deep connections between the hardness of clustering non-spherical Gaussian mixture models and lattice problems pivotal to modern cryptography~\cite{MR4398874-Bruna21,MR4537290-Gupte22}. 
From an algorithm design perspective, the problem of learning GMMs has long served as a benchmark for provable techniques in algorithmic statistics.
Indeed, advances for the problem over the years are closely linked to the development of new broadly applicable tools in statistical learning, such as dimension reduction~\cite{vempalawang, BelkinSinha}, the random projection method~\cite{kalai-moitra-valiant,moitravaliant10}, the method of moments~\cite{kalai-moitra-valiant,moitravaliant10}, tensor methods~\cite{MR3385380-Hsu13, MR3388256-Ge15}, outlier-filters~\cite{dks17learning}, and the sum-of-squares method~\cite{kotharisteinhardt18,hopkinsli18,dhkk20,bakshi2020outlierrobust}.

In his pioneering work, Dasgupta~\cite{dasgupta-1999} began the quest for provably efficient algorithms for Gaussian mixture models by focusing on the case of clustering \emph{spherical} (i.e., when all $\Sigma_i$s are scalings of $I_d$) mixtures.
This quest~\cite{arorakannan01,am05} culminated in algorithms to cluster mixtures with a dimension-independent pairwise component separation~\cite{vempalawang}.
Their work relied on a \emph{separation-preserving projection} to a low-dimensional subspace based on singular value decomposition.
The quantitative separation requirement was substantially improved to near-optimal in recent years~\cite{dks17learning,kotharisteinhardt18,hopkinsli18,liuli2022}, and the resulting algorithms also have \emph{outlier-robust} analogs that tolerate a constant fraction of arbitrary outliers~\cite{diakonikolas2019robust,lai-rao-vempala}. 

The \emph{non-spherical} case (i.e., arbitrary $\Sigma_i$s) has proved more challenging. 
The first clustering algorithm for such mixtures formed a crucial subroutine in the celebrated works~\cite{BelkinSinha,moitravaliant10} that gave the first polynomial-time algorithms for learning an arbitrary (potentially non-clusterable) GMM.
Their work required running time and sample complexity $(d/w_{\min})^{k^{\poly(k)}}$. 
Motivated by the goal of developing outlier-robust algorithms, a recent sequence of works~\cite{bakshi2020outlierrobust,dhkk20,liumoitra22,MR4490075-Bakshi22} improved this bound to $d^{O(k)} f(w_{\min}^{-1})$ using a new approach based on connections to efficient certificates for concentration and anti-concentration properties of Gaussian distributions. 
All these algorithms require $d^{\Omega(k)}$ time and samples, which is in stark contrast to the best-known algorithm for clustering spherical mixtures~\cite{liuli2022} that runs in time $d^{O(1)}$.
In fact, the statistical query lower bound of \cite{MR3734219-Diakonikolas17} and the sum-of-squares lower bound of \cite{MR4849259-Diakonikolas24} suggest that this cost is likely necessary!

\paragraph{A sub-exponential algorithm for GMMs?} At first glance, this might seem like the final word on the problem --- we have algorithms with a cost that aligns perfectly with the lower bounds in various models.
But a deeper dive suggests fascinating twists hiding just beneath the surface.  

The hardness results above are based on the \emph{parallel pancakes} construction that starts with a $1$-dimensional mixture of $k$ Gaussians which matches the standard Gaussian distribution in its first $\Omega(k)$ moments \cite{MR3734219-Diakonikolas17}. 
The ``hard'' mixture is then obtained by \emph{planting} the $1$-D construction in a random direction (with all orthogonal directions being standard Gaussian in all components).
Importantly, the resulting hard family of mixture models has (1) components of exponentially small weight, i.e., $w_{\min} \sim 2^{-k}$, and (2) non-zero means. 

What happens if we restrict ourselves to mixture models where the minimum mixing weight, $w_{\min}$, is at least $1/\poly(k)$?
Or if all the components are centered?
These are natural and well-studied settings in their own right, capturing problems such as mixtures of linear regressions and subspace clustering (covariances with distinct range spaces)~\cite{MR4141784-Chen20,dk20smallcover}.

Known algorithmic techniques do not appear to have the power to exploit such natural structures of the mixture components.
Indeed, the clustering algorithms in~\cite{bakshi2020outlierrobust,dhkk20} need $d^{\poly(k)}$ time and samples even for equiweighted (i.e., $w_{\min} = 1/k$) mixtures with centered components.\footnote{A concrete GMM where such methods need $d^{\Omega(k)}$ samples is the \emph{mixture of hyperplanes} setting: $\frac{1}{k} \sum_{i} N(0,I-v_iv_i^{\top})$ for arbitrary but distinct unit vectors $v_1, v_2, \ldots, v_k \in \R^d$. This is a special case of the more general problem of subspace clustering where the covariances are arbitrary projection matrices for which there was no known $d^{o(k)}$ time algorithm before this work.}
Technically, this is because the total variation separation between two Gaussian distributions $N(0,\Sigma_1)$ and $N(0,\Sigma_2)$ can arise due to the presence of a single direction $v$ such that $v^{\top} \Sigma_1 v/ v^{\top} \Sigma_2 v \gg 1$.
This setting of \emph{spectral separation} in both these works is tackled via certificates of \emph{anti-concentration}, and the best-known techniques~\cite{kkk19,ry20learning,bakshi2024efficient} for such certificates necessarily incur a time and sample complexity of $d^{\Omega(k)}$.
A recent work of \cite{buhai2023beyond}, also motivated by the question above, showed a weak \textit{partial clustering} algorithm that requires $d^{O(\log w_{\min}^{-1})}$ time and samples for Gaussian mixtures with identical covariances.
Their weak partial clustering,\footnote{Usually partial clustering algorithms give a near-laminar partition of the data -- i.e., every component is almost entirely on one side of the partition and each side has at least one component. In the weak partial clustering guaranteed in~\cite{buhai2023beyond}, the partition created is non-laminar and can split all but two clusters arbitrarily.} however, is only guaranteed to separate the samples from two of the components, and is allowed to split the samples from other components arbitrarily, so they cannot iterate to obtain a full clustering.

To summarize, our current algorithmic techniques cannot exploit arguably natural additional structure in non-spherical Gaussian mixtures.
At the same time, our hardness results, based as they are on essentially a single family of hard examples,\footnote{One exception that is \emph{not} based on the parallel pancakes construction is the family of hard examples given in \cite{diakonikolas2023sq}, which shows a quasi-polynomial lower bound for statistical query algorithms that learn mixtures of bounded-covariance GMMs with $\polylog(k)$ separation with respect to the \emph{largest} eigenvalue of the component covariance matrices.} do not imply any barriers for such settings.
In this work, we develop a new approach that leads to outlier-robust algorithms for non-spherical mixtures with zero means or identical unknown covariances that substantially improve on the best-known algorithms (even in the setting without outliers).

\paragraph{Our results} In this paper, we present a natural non-spherical analog of the influential dimension reduction technique of~\cite{vempalawang}, which gave the first algorithm for clustering spherical mixtures with dimension-independent separation requirements.
Specifically, we show an efficient algorithm based on a new sum-of-squares SDP relaxation that identifies a low-dimensional projection of the input mixture, preserving the total variation separation between at least one pair of initial components.
Unlike the method in~\cite{buhai2023beyond}, identifying this low-dimensional subspace allows us to partially cluster (i.e., find a cluster-preserving partition of the data) and iterate while accumulating only a small error.

As applications, we obtain new algorithms for outlier-robust clustering of well-separated\footnote{Our separation requirement is satisfied by all mixtures in which every pair of components has total variation distance at least $1-f(w_{\min})$ for a monotone $f$. All existing results on clustering mixtures of non-spherical Gaussians require a similar separation.}
centered Gaussian mixtures and mixtures with the same (albeit unknown and non-spherical) covariances. Our results significantly improve on all prior works in the high-dimensional setting when $d \gg k, w_{\min}^{-1}$.

The first setting we consider is clustering well-separated centered Gaussian mixtures. Specifically, we prove:

\begin{theorem}[Main theorem 1, see~\Cref{thm:zero-mean-main} for full version]
\label{thm:main1-intro}
Let $\mathcal{M}$ be a mixture of $k$ centered Gaussians with minimum weight $w_{\min}$ in which every pair of components is well-separated.\footnote{Like prior works on clustering non-spherical mixtures~\cite{dhkk20,bakshi2020outlierrobust}, we need the separation between every pair of components to be $1-f(w_{\min})$ in total variation distance for a monotone $f$.}
Also let $\mathcal{M}'$ be a distribution satisfying $d_{\TV}(\mathcal{M}', \mathcal{M}) \leq \epsilon$, where $\epsilon \ll w_{\min}$.
Then there exists an algorithm that takes as input independent samples from $\mathcal{M}'$ and outputs a clustering such that in each cluster at most a $O(kw_{\min}^{-1}\max\{\epsilon, k^{-k}\})$-fraction of points are misclassified. The sample complexity and running time of the algorithm is $d^{O(1)}f(w_{\min}^{-1})$.\footnote{It suffices to take $f(x) = \exp(\exp(\exp(\exp(\tilde{O}(x^2)))))$.}
\end{theorem}
Observe that the exponent of the dimension in the cost above is a \emph{fixed constant completely independent of $k$}! This is in contrast to the best-known results~\cite{dhkk20,bakshi2020outlierrobust,MR4490075-Bakshi22}, that necessarily incur a $d^{\Omega(k)}$ running time and sample complexity even for zero-mean Gaussian mixtures. 

\paragraph{Subspace clustering}
As an immediate corollary (see~\Cref{cor:zero-mean-subspace}), we obtain an outlier-robust algorithm for Gaussian \emph{subspace clustering} (with arbitrary dimensional subspaces) that runs in time and samples $d^{O(1)} f(w_{min}^{-1})$.
There is a vast body of empirical work on subspace clustering where, roughly speaking, the goal is to find the constituent subspaces in data generated from a union of subspaces.
Our corollary applies to the Gaussian model (studied in~\cite{MR4141784-Chen20} for the special case of subspaces of co-dimension 1).
Formally, we are given a sample from a mixture of zero-mean Gaussian distributions with each covariance restricted to a distinct subspace.
The total variation distance between $\cN(0,\Sigma)$ and $\cN(0,\Sigma')$ is $1$ whenever the range spaces of $\Sigma, \Sigma'$ are different.
Thus, our theorem above is applicable. While there have been improved algorithms for the special case of subspaces with low dimension or co-dimension (e.g., mixture of hyperplanes~\cite{MR4141784-Chen20,dk20smallcover}), to the best of our knowledge, only $d^{\poly(k)}$ sample and time algorithms~\cite{MR4262511-Bakshi21,MR4490075-Bakshi22} were known in general prior to our work.

The second setting we consider is clustering well-separated mixture models where all components have identical but unknown and arbitrary covariance.
This setting was studied recently in the work of Buhai and Steurer~\cite{buhai2023beyond}, who gave a \emph{partial clustering} algorithm, with time and sample complexity $d^{O(\log w_{\min}^{-1})}$, for the problem that splits the input sample into two groups so that each side of the partition contains 99\% of the samples from at least one component.
Their algorithm, however, can split other components arbitrarily across the two groups. 
Hence, one cannot iterate their partial clustering subroutine to find a full clustering of the data.
This limitation appears inherent to their approach.
Here, we give a $d^{O(\log w_{\min}^{-1})} f(w_{\min}^{-1})$ samples and $d^{O(\log^2 w_{\min}^{-1})}f(w_{\min}^{-1})$ time algorithm to cluster such mixtures.
Our algorithm builds on a result of~\cite{buhai2023beyond} showing that a same-covariance mixture cannot match too many moments of a standard Gaussian; see Section~\ref{sec:tech} and specifically Section~\ref{sec:tech-same-cov} for a technical overview of the algorithm.

\begin{theorem}[Main theorem 2, see~\Cref{thm:same-cov-clustering} for full version]
\label{thm:main2-intro}
Let $\mathcal{M}$ be a mixture of $k$ identical-covariance Gaussians with minimum weight $w_{\min}$ in which every pair of components is well-separated.
Also let $\mathcal{M}'$ be a distribution satisfying $d_{\TV}(\mathcal{M}', \mathcal{M}) \leq \epsilon$, where $\epsilon \ll w_{\min}$.
Then there exists an algorithm that takes as input independent samples from $\mathcal{M}'$ and outputs a clustering such that in each cluster at most a $O(kw_{\min}^{-1}\max\{\epsilon, k^{-k}\})$-fraction of points are misclassified, as long as every pair of components in the input mixture is well-separated.
The sample complexity of the algorithm is $d^{O(\log w_{\min}^{-1})} f(w_{\min}^{-1})$ and running time is $d^{O(\log^2 w_{\min}^{-1})}f(w_{\min}^{-1})$.\footnote{It suffices to take $f(x) = \exp(\exp(\exp(\exp(\tilde{O}(x^2)))))$.}
\end{theorem}

\begin{remark}
Following the upload of a preprint of our paper, \cite{pmlr-v267-diakonikolas25b} showed a $d^{\Omega(\log k)}$ statistical query lower bound for the problem considered in \Cref{thm:main2-intro}. In fact, their lower bound holds in the special case when all mixing weights are equal and there are no outliers ($\epsilon = 0$). Thus, in the statistical query model, one cannot improve the query complexity to $d^{O(1)} f(w_{\min}^{-1})$, proving a separation between the centered Gaussians setting in \Cref{thm:main1-intro} and the identical-covariance Gaussians setting in \Cref{thm:main2-intro}.
\end{remark}

\paragraph{Our key idea: dimension reduction via sum-of-squares} A key component of our algorithms is a subroutine for finding a projection of the mixture to a subspace of dimension independent of the underlying dimension $d$ that still maintains the pairwise separation between at least one pair of components. Let us explain this key idea through our first result on zero-mean non-spherical mixtures of Gaussians. Let us first recall the classical and celebrated idea of Vempala and Wang that accomplished such a step in the context of spherical Gaussian mixtures.

Given a spherical mixture, say, $\sum_{i = 1}^k w_i \cN(\mu_i, I)$, the second moment of the mixture is $\sum_{i} w_i \mu_i \mu_i^{\top} + I$. If we could project the mixture onto the $k$-dimensional subspace spanned by the means $\mu_i$s, we could maintain all the pairwise $\ell_2$ separations between the component means (and thus maintain the total variation distance) while drastically reducing the dimensionality of the problem to $k \ll d$. The span of $\mu_i$s can be easily computed by noting that the second moment of the input mixture minus the identity is $\sum_i w_i \mu_i \mu_i^{\top}$. While an outlier-robust variant of this~\cite{hopkinsli18,kotharisteinhardt18,dks17learning} requires more care, this idea immediately gives a dimension reduction subroutine in the setting without any outliers. 

Now, let us attempt a natural generalization to the non-spherical case when the input mixture is $\mathcal{M}=\sum_i w_i \cN(0,\Sigma_i)$ for unknown $\Sigma_i$s. The pairwise total variation separation between the components is now governed by the spectral norm $\Norm{\Sigma_i^{-1} \Sigma_j}_2$ (or the relative Frobenius norm $\Norm{ \Sigma_i^{-1/2} \Sigma_j \Sigma_i^{-1/2} -I}_F$, see \cite{dhkk20,bakshi2020outlierrobust}, but this type of separation can be tackled by ideas in prior works). A natural idea would be to find the span of the  $\Sigma_i$s and then project the second tensor powers of the data points onto this span. However, the natural analog that relies on the fourth-moment tensor (the second moment of the second tensor powers of the data points),  tempting as it may appear, \emph{provably fails}, \emph{even for a mixture of two Gaussians}. The key difficulty is that the data only gives access to the \emph{symmetrized} moment tensor. That is, the data allows us to compute $\E_{\mathcal{M}} [x^{\otimes 4}] = \operatorname{Sym}(\sum_i  w_i \Sigma_i^{\otimes 2})$, where $\operatorname{Sym}$ averages its argument over the symmetries of a fourth-order tensor. The $(a,b,c,d)$-th entry of the resulting moments is then $\sum_i w_i (\Sigma_i(a,b) \Sigma_i (c,d) + \Sigma_i(a,c) \Sigma_i(b,d) + \Sigma_i (a,d) \Sigma_i (b,c))$. This ``mixing up" of the four modes makes the resulting fourth-moment tensor, when viewed as a matrix, a \emph{high-rank matrix} (as opposed to rank $k$). And, in fact, recovery is provably impossible, even if $\Sigma_i$s are ``smoothed" (see Lemma A.2 in \cite{MR4537271-Bafna22}). 

A more sophisticated idea is to find a small-dimensional subspace that is guaranteed to contain at least one ``direction of separation" --- a vector $v \in \R^d$ such that $v^{\top} \Sigma_i v \gg v^{\top} \Sigma_j v$ for some $i,j$ that is guaranteed to exist when the components are spectrally separated. To illustrate, consider the special case when $\Sigma_i = I - R_i$ for some low-rank projection matrix $R_i$. Then, any $v$ in the subspace associated with $R_i$ is a direction of separation. Of course, in general, while we can use techniques in prior works to ensure that $\Norm{\Sigma_i-I}_F \leq \poly(k,1/w_{min})$, we cannot actually assume that $\Sigma_i$s are low-rank perturbations as in the above ideal case. Nevertheless, this idealized case already helps us showcase the technical issues we will tackle via our new dimension reduction approach. By estimating the fourth \emph{Hermite} moment tensor (which can be obtained by combining the first four raw moment tensors), we can get our hands on $T_4 = \operatorname{Sym}(\sum_i w_i R_i^{\otimes 2})$. This now appears quite promising since the $R_i$s have low rank, and one may expect to ``read" off a $v$ by simply looking at, say, the canonical flattening of $T_4$ and taking the top eigenvector. But, unfortunately, as  Lemma A.2 from~\cite{MR4537271-Bafna22} shows, there can be ``cancellations" caused by the symmetrization operation that could, in principle, force that no $v$ is close to the top subspace of this matrix.

\paragraph{Sum-of-Squares approach to find \emph{spherical} directions} Our key idea to escape these issues is to develop a new approach to find the \emph{spherical} directions in the input Gaussian mixture. These are directions $v$ such that $v^{\top} \Sigma_i v \approx \Norm{v}_2^2$.  For every $\Sigma_i$, there is a $d-f(k)$ dimensional subspace such that all vectors in this subspace are such spherical directions. However, note that the set of spherical directions is not a linear subspace but a complicated semi-algebraic set. Nevertheless, we show that a natural sum-of-squares relaxation can recover spherical directions and, with some augmentation, also help find a direction that is \emph{not} in such a spherical set. We describe our new dimension reduction approach in detail in the technical overview that follows. 

While we focused this discussion on zero-mean mixtures, we note that similar issues arise in recovering directions of separation even in the case of unknown but identical-covariance mixtures, as in our second result. Indeed, this is a key difficulty encountered in the prior work~\cite{buhai2023beyond} that dealt with such mixtures. It turns out that our dimension reduction subroutine naturally applies and allows us to recover directions of separation even in that setting. 

We believe that this subroutine will likely find more applications in statistical estimation.
The next section explains the key ideas behind this subroutine and how it applies to clustering mixtures of Gaussians.

\section{Technical Overview}
\label{sec:tech}

To present our key technical ideas, we focus on the case of equiweighted centered mixtures of Gaussians (without outliers) as the main example and explain how our ideas apply to the identical covariance case at the end.

\paragraph{The certifiable anti-concentration bottleneck}
Before discussing our new approach, we explain the key bottleneck in prior approaches to clustering~\cite{bakshi2020outlierrobust,dhkk20}.
We note that the works~\cite{liumoitra22,MR4490075-Bakshi22} use a clustering step~\cite{bakshi2020outlierrobust,dhkk20} as a subroutine. 

To begin with, let us translate total variation distance into a more directly interpretable parameter distance bound.
Recall (see~\cite{bakshi2020outlierrobust,dhkk20}) that if two Gaussian distributions $N(\mu_1, \Sigma_1)$ and $N(\mu_2,\Sigma_2)$ are separated by a total variation distance of $1-\exp(-O(\Delta^2 \log \Delta))$ for some parameter $\Delta>0$, then they must be $\Delta$-separated in \emph{parameter distance} as defined below:

\begin{definition}[Parameter distance]
\label{def:param-distance}
We say that two distributions on $\R^d$ with means $\mu_1,\mu_2$ and positive definite covariances $\Sigma_1, \Sigma_2$ are $\Delta$-separated in parameter distance if at least one of the following conditions holds for $\Set{i,j}=\Set{1,2}$:
\begin{enumerate}
\item Mean separation: $\exists v \in \R^d$ such that $\iprod{\mu_i-\mu_j,v}^2 > \Delta^2 \cdot v^{\top} (\Sigma_i + \Sigma_j) v$, 
\item Spectral separation: $\exists v \in \R^d$ such that $v^{\top} \Sigma_i v > \Delta^2 \cdot v^{\top} \Sigma_j v$, 
\item Relative Frobenius separation: $\Norm{\Sigma_i^{-1/2} \Sigma_j \Sigma_i^{-1/2} - I}^2_F > \Delta^2 \cdot \Norm{\Sigma_i^{-1/2} \Sigma_j \Sigma_i^{-1/2}}^2$.  
\end{enumerate}
\end{definition}

In the clustering problem, we are thus given a mixture model where every pair of components is separated in at least one of the three modes of separation above for some $\Delta$ depending only on $w_{\min}$.
The key idea in~\cite{bakshi2020outlierrobust} is to write a sum-of-squares relaxation that attempts to find a cluster that satisfies some concentration (moments of degree $\leq 2$ polynomials grow at the expected sub-exponential rate) and anti-concentration (every one-dimensional marginal takes a small value only with a small probability) properties that an actual Gaussian cluster satisfies.
Their analysis shows that these properties characterize actual clusters: any subset of the data that satisfies these properties cannot appreciably intersect two distinct ground truth clusters.
To invoke the above properties, they need efficient algorithms to certify the concentration and anti-concentration properties of the input sample in \emph{all} directions (and degree-$2$ polynomials) simultaneously.

Algorithmically verifying the concentration property required to handle relative Frobenius separation needs only $d^{O(1)}$ time and samples.
Indeed, this observation is a core component in robustly learning a mixture of $k$ arbitrary Gaussians (see Theorem 4.3 in~\cite{MR4490075-Bakshi22}).

The critical bottleneck in~\cite{bakshi2020outlierrobust} is the cost of certifying the anti-concentration property, which uses moment-based certificates and necessarily incurs a cost of $d^{\Omega(k)}$ in sample complexity, as $\Omega(k)$ moments are necessary for their goal. They incur this cost even in the case of centered mixtures, and it is required for their algorithm to cluster mixtures that contain pairs of components that only satisfy condition 2 in~\Cref{def:param-distance}.
A concrete example for which the prior methods require $d^{\Omega(k)}$ samples is the total-variation separated mixture $\frac{1}{k} \sum_{i} N(0,I-v_iv_i^{\top})$ for arbitrary but distinct unit vectors $v_1, v_2, \ldots, v_k \in \R^d$. 

\subsection{Clustering Mixtures with Centered Components}

\paragraph{No parallel pancakes with centered components}
Our starting point is the observation that there is no well-separated $1$D mixture with centered components that matches the first \emph{four} moments of a standard Gaussian. 
This is in contrast to the parallel pancakes example, which consists of a $1$D mixture of $k$ Gaussians that matches the first $\Omega(k)$ moments of $N(0,1)$. 

\begin{proposition}[No moment matching with centered mixtures]
Let $\sum_{i=1}^k w_i N(0, \sigma_i^2)$ be a mixture of $k$ Gaussians that has second moment $1$ and fourth moment $3$.
Then $\sigma_i^2=1$ for all $i \in [k]$.
\end{proposition}
\begin{proof}
Each $N(0, \sigma_i^2)$ has second moment $\sigma_i^2$ and fourth moment $3\sigma_i^4$.
Then ${\sum_{i=1}^k w_i \sigma_i^2=1}$ and $\sum_{i=1}^k w_i (3\sigma_i^4)=3$, so ${\sum_{i=1}^k w_i \sigma_i^4 = 1=\Paren{\sum_{i=1}^k w_i \sigma_i^2}^2}$, which is only possible if ${\sigma_1^2=\ldots=\sigma_k^2=1}$.
\end{proof}

Of course, this fact does not necessarily imply efficient clusterability.
In fact, it does not even rule out the existence of two distinct pairs of well-separated centered $1$D mixtures with the same first $\Omega(k)$ moments.
Nevertheless, we will crucially use this observation's high-dimensional variant.%

\paragraph{Reduction to tensor subspace recovery}
Given the mixture's fourth moment tensor, we can reduce our task to the following tensor subspace recovery problem.
(For a tensor $T$ of order $m$, we define $\operatorname{Sym}(T) = \frac{1}{m!} \sum_{\pi} \pi T$ over permutations $\pi$ of $[m]$, where $(\pi T)_{i_1, \ldots, i_m} = T_{\pi(i_1), \ldots, \pi(i_m)}$.)

\begin{proposition}[Tensor subspace recovery, informal]
\label{prop:tensor}
Let ${M_4 = \frac{1}{k} \sum_{i=1}^k \operatorname{Sym}(S_i^{\otimes 2}) \in \R^{d^{\otimes 4}}}$, where $S_1, \ldots, S_k \in \R^{d \times d}$ are unknown symmetric matrices.
For each $S_i$, define $P_i^{\geq\tau}$ as the span of the eigenvectors of $S_i$ with eigenvalues larger than $\tau$ in absolute value.
Fix some $\tau > 0$, and let $\tau' = 1/f(k, \tau^{-1})$ be small enough.
Then, there exists an algorithm that takes as input $M_4$, runs in time $\poly(d) \cdot f(k, \max_i \dim(P_i^{\geq \tau'}))$, and outputs an approximate cover of $P_1^{\geq\tau} \cup \ldots \cup P_k^{\geq\tau}$.
\end{proposition}

Let us briefly explain the reduction.
We can assume without loss of generality that the input mixture $\frac{1}{k} \sum_{i=1}^k N(0, \Sigma_i)$ is in isotropic position $\frac{1}{k} \sum_{i=1}^k \Sigma_i = I_d$.
Furthermore, if the mixture satisfies $\Norm{\Sigma_i - I_d}_F \gg 1$ for any $i$, we can apply an existing algorithm of~\cite{MR4490075-Bakshi22} to partially cluster the mixture in polynomial time and iterate, so we can assume without loss of generality that $\Norm{\Sigma_i - I_d}_F \leq O_k(1)$ for all $i$.
Under this Frobenius closeness assumption, any two mixture components can only be separated spectrally (Condition 2 in~\Cref{def:param-distance}).
Our strategy is then to find a ``direction of separation" along which at least two components satisfy spectral separation --- by projecting the samples along this direction, we can partially cluster the mixture and iterate.
It turns out that all directions of such separation $v$ must satisfy $v^\top \Sigma_i v \ll 1$ for some $i$, so $v$ must be close to the eigenspace of $\Sigma_i$ with eigenvalues bounded away from one.
Hence, applying \Cref{prop:tensor} with $M_4=\frac{1}{k} \sum_{i=1}^k \operatorname{Sym}((\Sigma_i - I_d)^{\otimes 2})$ (which is the fourth Hermite moment tensor of the mixture and can be estimated from the samples) and $\tau = \Omega_k(1)$ would yield a cover of all directions of separation.

It remains to argue that applying \Cref{prop:tensor} takes time $\poly(d) \cdot O_k(1)$ in our setting.
Since we have that $\Norm{S_i}_F \leq O_k(1)$, and because the eigenvectors of $S_i$ in $P_i^{\geq \tau'}$ have eigenvalues lower bounded by $\tau' = \Omega_k(1)$ in absolute value, we must also have that $\dim(P_i^{\geq \tau'}) \leq O_k(1)$ for all $i$.
This gives the desired bound on the time complexity.

Let us focus now on the tensor problem in \Cref{prop:tensor}. For simplicity of notation, we drop the superscript from $P_1^{\geq \tau}, \ldots, P_k^{\geq \tau}$.

\paragraph{Linear algebraic algorithms?}
Our algorithm will be based on sum-of-squares relaxations.
One might wonder, however, whether simpler linear algebra techniques might suffice to solve \Cref{prop:tensor}.
Let us briefly explore this.

To begin with, consider a much simpler setting: Suppose we are given $\widetilde{M}_4 = \frac{1}{k}\sum_{i=1}^k S_i^{\otimes 2}$ without the symmetrization operation.
Then, using $\widetilde{M}_4$, we can recover the $(\leq k)$-dimensional span of the matrices $\operatorname{span}(S_1, \ldots, S_k)$.
Furthermore, because of the bound $\Norm{S_i}_F \leq O_k(1)$, we can afford in time $\poly(d) \cdot O_k(1)$ to enumerate a net of matrices in $\operatorname{span}(S_1, \ldots, S_k)$ with bounded Frobenius norm such that, for each $S_i$, we find some $\hat{S}_i$ satisfying $\Norm{S_i - \hat{S}_i}_F \leq \Omega_k(1)$.
Then the span of the large eigenvectors of $\hat{S}_i$ must be close to the span of the large eigenvectors of $S_i$, and it suffices to include in the output a cover of this span for every matrix that we enumerate.

The symmetrization operation, however, changes the problem significantly, and it is no longer clear how to obtain $\operatorname{span}(S_1, \ldots, S_k)$.
This is reminiscent of the prior works on power sum decomposition~\cite{MR3388256-Ge15, MR4232095-Garg20, MR4537271-Bafna22}, where the key idea is also to undo the effect of the symmetrization operation (see also the related approach of \cite{MR4695181-Chandra24}). However, the ideas in those works strongly rely on $\Sigma_i$s being \emph{random} (with later works extending this to a setting with smoothing with additive gaussian perturbations) and break down in the setting of this work, where we have no assumptions on $\Sigma_i$s except for the total variation separation. %

\paragraph{Identifiability via unit vectors}
To find a set of interest --- in our case, $P_1 \cup \ldots \cup P_k$ ---, a common approach in the robust statistics literature is to formulate a system of constraints over unit vectors $v \in \R^d$ that identifies the points in this set.
Then, the hope is that a sum-of-squares relaxation of this program can approximate the solution set.

With this idea in mind, we observe that for a vector $v$ we have $\langle M_4, v^{\otimes 4}\rangle = \frac{1}{k} \sum_{i=1}^k (v^\top S_i v)^2$.
Hence, $\langle M_4, v^{\otimes 4}\rangle$ is large when at least one of $|v^\top S_i v|$ is large.
Let $P_i^+$ and $P_i^-$ be the span of the eigenvectors of $S_i$ captured by $P_i$ with positive / negative eigenvalues, such that $P_i = \operatorname{span}(P_i^+ \cup P_i^-)$.
Then, for all unit vectors $v$ contained in either $P_i^+$ or $P_i^-$, we have $\langle M_4, v^{\otimes 4}\rangle \geq \tau^2/k$.
This observation suggests that, to identify $\operatorname{span}(P_1 \cup \ldots \cup P_k)$, it may suffice to identify those unit vectors $v$ for which $\langle M_4, v^{\otimes 4}\rangle$ is bounded away from zero.
Unfortunately, $\langle M_4, v^{\otimes 4}\rangle$ remains large even for $v$ that are $\Omega(1)$-far from the span of $P_1, \ldots, P_k$, and therefore it describes a set of points whose cover has size exponential in $d$.

\paragraph{Orthogonal complement}
Instead, we observe the intriguing fact that $\langle M_4, v^{\otimes 4}\rangle \approx 0$ if and only if $v^\top S_i v \approx 0$ simultaneously for all $i$.
This is essentially the opposite of what we are looking for and suggests switching up our goal: what if we search instead for a subspace \emph{orthogonal} to $\operatorname{span}(P_1 \cup \ldots \cup P_k)$?
Taking the orthogonal complement of such a subspace would give us $\operatorname{span}(P_1 \cup \ldots \cup P_k)$ back.

Unfortunately, the condition $\langle M_4, v^{\otimes 4}\rangle \approx 0$ does not imply that $v$ is orthogonal to $\operatorname{span}(P_1 \cup \ldots \cup P_k)$: $v^\top S_i v$ can be $\approx 0$ even if $v$ is contained in $P_i$, if it is split between $P_i^+$ and $P_i^-$.
The resolution of this issue is a novel condition on \textit{the Hessian} of $\langle M_4, v^{\otimes 4}\rangle$, which allows us to identify only those $v$ for which $S_i v \approx 0$ holds for all $i$.
The condition $S_i v \approx 0$ is stronger than the condition $v^\top S_i v \approx 0$, as it implies that $v$ is close to the kernel of $S_i$, and hence that $v$ is nearly orthogonal to $P_i$.

Let us illustrate the Hessian condition in the exact case in which $v$ lies in the kernel of each $S_i$.
For such $v$, a simple calculation shows that $\langle M_4, v^{\otimes 4}\rangle = 0$ and $\nabla^2_v \langle M_4, v^{\otimes 4}\rangle = 0$.
We prove that the converse is also true:

\begin{proposition}
\label{prop:hessian2}
Suppose $v \in \R^d$ satisfies $\langle M_4, v^{\otimes 4}\rangle = 0$ and $\nabla^2_v \langle M_4, v^{\otimes 4}\rangle = 0$.
Then $S_i v = 0$ for all $i$.
\end{proposition}
\begin{proof}
The Hessian of $\langle M_4, v^{\otimes 4}\rangle$ is 
\[\nabla^2_v \langle M_4, v^{\otimes 4}\rangle = \frac{8}{k} \sum_{i=1}^k (S_i v)^{\otimes 2} + \frac{4}{k} \sum_{i=1}^k (v^\top S_i v) S_i\,.\]
The assumption $\langle M_4, v^{\otimes 4}\rangle = 0$ implies that $v^\top S_i v = 0$ for all $i$, so we can eliminate the second term.
Hence, 
\[\nabla^2_v \langle M_4, v^{\otimes 4}\rangle = \frac{8}{k} \sum_{i=1}^k (S_i v)^{\otimes 2}\,.\]
Then the assumption $\nabla^2 \langle M_4, v^{\otimes 4}\rangle = 0$ implies that $\frac{1}{k} \sum_{i=1}^k (S_i v)^{\otimes 2} = 0$, which holds only if $S_i v = 0$ for all $i$.
\end{proof}

By an approximate version of \Cref{prop:hessian2}, we show that whenever $\langle M_4, v^{\otimes 4}\rangle \approx 0$ and\linebreak $\nabla^2_v \langle M_4, v^{\otimes 4}\rangle \approx 0$, we have $\Norm{P_i v} \approx 0$ for all $i$, where by $P_i$ we denote here the orthogonal projection to the subspace $P_i$.
Furthermore, we can express this proof as a constant-degree sum-of-squares proof.

\paragraph{Sum-of-squares rounding}
Hence, we can create a system of polynomial constraints $\mathcal{A}(v)$ over a variable $v \in \R^d$ that implies via constant-degree sum-of-squares proofs that $\Norm{P_i v} \approx 0$ \mbox{for all $i$}.
We also show that $\mathcal{A}(v)$ is non-trivial: there exists a subspace of dimension $d-f(k, \max_i \dim(P_i^{\geq \tau'}))$ such that all vectors in it satisfy $\mathcal{A}(v)$ (this is a subspace of vectors close to the kernel of each $S_i$).
Recall, however, that in the end we are interested in recovering a subspace that includes $P_1 \cup \ldots \cup P_k$, i.e., the opposite of what we have now.

To recover the subspace of interest, we define a new system of constraints $\mathcal{B}(u)$ over a variable $u \in \R^d$ that imposes that $u$ is \emph{nearly orthogonal to all $v$ that satisfy $\mathcal{A}(v)$}.
Then, we show that $\mathcal{B}(u)$ implies via low-degree sum-of-squares proofs that $u$ is close to an $f(k, \max_i \dim(P_i^{\geq \tau'}))$-dimensional subspace, and furthermore all $u$ with $\Norm{P_i u} \approx 1$ for at least one $P_i$ are feasible for $\mathcal{B}(u)$.

Finally, we design and analyze a rounding algorithm for $\mathcal{B}(u)$ that recovers a subspace of dimension $f(k, \max_i \dim(P_i^{\geq \tau'}))$, which approximately contains all the vectors $u$ feasible for $\mathcal{B}(u)$.
Hence, this subspace approximately contains $P_1 \cup \ldots \cup P_k$, solving the tensor problem.
Crucially, because all the sum-of-squares proofs involved have constant degree, the time complexity of the algorithm is polynomial in $d$.

\subsection{Clustering Mixtures with Identical-Covariance Components}
\label{sec:tech-same-cov}

We explain now how these ideas also apply to the case of identical-covariance mixtures. 
Consider a distribution $\frac{1}{k} \sum_{i=1}^k N(\mu_i, \Sigma)$ with components separated in total variation distance.
We can assume without loss of generality isotropic position $\frac{1}{k} \sum_{i=1}^k \mu_i \mu_i^\top + \Sigma = I_d$.

Our strategy, as in the centered case, is to find a direction of separation along which to partially cluster the mixture and then iterate.
A simple argument shows that all directions of separation $v$ must satisfy $v^\top \Sigma v \ll 1$, so $v$ must be close to the eigenspace of $\Sigma$ with eigenvalues close to zero.
Furthermore, by the isotropic position, $\Sigma$ has at least $d-k$ eigenvalues equal to $1$.
Therefore, all directions of separation are close to a subspace of dimension $k$.

\paragraph{Identifiability via $O(\log k)$ moments}
Unfortunately, unlike the case of mixtures with centered components, for mixtures with identical-covariance components the fourth moment tensor may exactly match that of $N(0, I_d)$, giving no information about the components.
We rely instead on a result of Buhai and Steurer~\cite{buhai2023beyond}, who proved that an identical-covariance mixture of Gaussians cannot match the first $O(\log k)$ moments of $N(0, I_d)$.
By an adjustment of their argument, we show that if the mixture matches the first $O(\log k)$ moments of $N(0, 1)$ when projected onto a direction $v$, then $v$ must be orthogonal to all $\mu_i$s, and hence $v^\top \Sigma v = 1$.

Then, we consider the system of polynomial constraints over a variable $v \in \R^d$ that imposes that the first $O(\log k)$ moments of the mixture projected in direction $v$ are approximately equal to those of $N(0, 1)$.
We can prove that this system implies, via $O(\log k)$-degree sum-of-squares proofs, that $v^\top \Sigma v \approx 1$.

\paragraph{Sum-of-squares rounding}
As in the case of zero-mean components, we ended up with a system of constraints that identifies the opposite of what we want: in the end, we are interested in the eigenspace of $\Sigma$ with eigenvalues close to zero.
Hence, we apply the same strategy as before: we create a new system of constraints over a variable $u \in \R^d$ which imposes that $u$ is nearly orthogonal to all $v$ that satisfy our system of constraints.
As before, we can round this system of constraints to obtain a subspace of dimension $O_k(1)$ that approximately contains the eigenspace of $\Sigma$ with eigenvalues close to zero.
Because the sum-of-squares proofs involved here have degree $O(\log k)$, the time complexity of this algorithm ends up including a term $d^{\polylog(k)}$.

\section{Preliminaries}
While we will omit discussions of the computational model and numerical precision in this work, we note that our algorithms can be implemented to give guarantees in the standard bit complexity model.
For a detailed discussion, we direct the reader to Section 3.1 of~\cite{MR4490078-Ivkov22}. 

\subsection{Sum-of-Squares Proofs}

In this section we give a brief introduction to sum-of-square proofs, taken from the corresponding section in~\cite{buhai2023beyond}.
For a more detailed exposition, see the lecture notes~\cite{barak2016proofs} or the monograph~\cite{MR4059250-Fleming19}.

\begin{definition}[Sum-of-squares proofs]
\label{def:sos-proof-preliminaries}
Let $p(x)$ and $q_1(x), ..., q_m(x)$ be polynomials over $x \in \mathbb{R}^n$ and let $\mathcal{A} = \{q_1(x) \geq 0, ..., q_m(x) \geq 0\}$ be a system of polynomial inequalities.
A \textit{sum-of-squares proof of degree $t$} that $p(x) \geq 0$ under $\mathcal{A}$ is an identity of the form
\[p(x) = \sum_{S \subseteq [m]} \left(\sum_{i=1}^{m_S} r_{S,i}(x)^2\right) \prod_{j\in S} q_j(x)\]
for polynomials $r_{S, i}(x)$, such that $\max_{S, i} \operatorname{deg}(r_{S, i}(x)^2 \prod_{j \in S} q_j(x)) \leq t$.
\end{definition}

If there exists a sum-of-squares proof of degree $t$ that $p(x) \geq 0$ under $\mathcal{A}$, we write $\mathcal{A} \sststile{t}{x} p(x) \geq 0$.
We also use the notation $\mathcal{A} \sststile{t}{x} p(x) \geq q(x)$ if $\mathcal{A} \sststile{t}{x} p(x) - q(x) \geq 0$ and $\mathcal{A} \sststile{t}{x} p(x) \leq q(x)$ if $\mathcal{A} \sststile{t}{x} q(x) - p(x) \geq 0$.
If $\mathcal{A}=\emptyset$, we omit it altogether and write $\sststile{t}{x} p(x) \geq 0$. 
We also sometimes omit $\mathcal{A}$ if it is clear from context what axioms are assumed.
We note that if $\mathcal{A} \sststile{t}{x} p(x) \geq q(x)$ and $\mathcal{A} \sststile{t}{s} q(x) \geq r(x)$, then $\mathcal{A} \sststile{t}{x} p(x) \geq r(x)$, which allows writing chains of inequalities of the form $\mathcal{A} \sststile{t}{x} p(x) \geq s(x) \geq r(x)$.

\paragraph{Pseudo-distributions and pseudo-expectations}

\begin{definition}[Pseudo-distributions]
A \textit{pseudo-distribution} $D$ of degree $t$ is a function from $\mathbb{R}^n$ to $\mathbb{R}$ with finite support such that $\sum_{x \in \operatorname{supp}(D)} D(x) = 1$ and $\sum_{x \in \operatorname{supp}(D)} D(x) p(x)^2 \geq 0$ for all polynomials $p(x)$ with $\operatorname{deg}(p(x)^2) \leq t$.
\end{definition}

\begin{definition}[Pseudo-expectations]
Given a pseudo-distribution $D$ of degree $t$, the associated \textit{pseudo-expectation} $\tilde{\mathbb{E}}_{D(x)}$ is defined by $\tilde{\mathbb{E}}_{D(x)} f(x) = \sum_{x \in \operatorname{supp}(D)} D(x) f(x)$ for a function $f(x)$.
\end{definition}

We now define the notion of a pseudo-distribution that satisfies a set of polynomial inequalities.

\begin{definition}[Constrained pseudo-distributions]
A pseudo-distribution $D$ of degree $t$ \textit{satisfies} the set of polynomial inequalities $\mathcal{A} = \{q_1(x) \geq 0, ..., q_m(x)\geq 0\}$ if, for all $S \subseteq [m]$, $\tilde{\mathbb{E}}_{D(x)} r(x)^2 \prod_{j \in S} q_j(x) \geq 0$ for all polynomials $r(x)$ such that $\operatorname{deg}(r(x)^2 \prod_{j \in S}q_j(x)) \leq t$.

$D$ \textit{approximately satisfies} $\mathcal{A}$ up to error $\eta$ if, under the same conditions as in the previous case, $\tilde{\mathbb{E}}_{D(x)} r(x)^2 \prod_{j \in S} q_j(x) \geq -\eta \|r(x)^2\| \prod_{j \in S} \|q_j(x)\|$, where $\|p(x)\|$ denotes the $2$-norm of the vector of coefficients of the polynomial $p(x)$.
\end{definition}

The connection between pseudo-distributions and sum-of-squares proofs is made in \Cref{fact:pe-satisfies-sos}, which shows that if a pseudo-distribution satisfies a set of polynomial inequalities, then it also satisfies any other polynomial inequalities derived from this set through sum-of-squares proofs.

\begin{fact}
\label{fact:pe-satisfies-sos}
If $D$ is a pseudo-distribution of degree $t$ that satisfies $\mathcal{A}$ and if $\mathcal{A} \sststile{s}{x} p(x) \geq 0$, then ${\tilde{\mathbb{E}}_{D(x)} r(x)^2 p(x) \geq 0}$ for all polynomials $r(x)$ such that $\operatorname{deg}(r(x)^2 p(x)) \leq t$.
If $D$ approximately satisfies $\mathcal{A}$ up to error $\eta$, then, under the same conditions as in the previous case,
\[{\tilde{\mathbb{E}}_{D(x)} r(x)^2 p(x) \geq -\eta \|r(x)^2\| \|p(x)\|}\,.\]
\end{fact}

Finally, \Cref{fact:sos-to-pe} shows that there exists an algorithm with time complexity $(n+m)^{O(t)}$ to compute a pseudo-distribution of degree $t$ that approximately satisfies $\mathcal{A}$ up to error $2^{-n^{\Theta(t)}}$.

\begin{fact}
\label{fact:sos-to-pe}
For $x \in \mathbb{R}^n$, if $\mathcal{A} = \{q_1(x) \geq 0, ..., q_m(x) \geq 0\}$ is feasible and explicitly bounded,\footnote{Explicit boundedness means that $\mathcal{A}$ contains a constraint of the form $x_1^2 + ... + x_n^2 \leq B$.
In our applications it is possible to add such a constraint with $B$ large enough such that the constraint is always satisfied by the intended solution.}
then there exists an algorithm that runs in time $(n+m)^{O(t)}$ and computes a pseudo-distribution of degree $t$ that approximately satisfies $\mathcal{A}$ up to error $2^{-n^{\Theta(t)}}$.\footnote{In our applications this error is negligible.}
\end{fact}

\paragraph{Common sum-of-squares proofs}

We state now some widely used sum-of-squares proofs that we use in the current paper.

\begin{fact}[Cauchy-Schwarz, see Lemma A.1 in~\cite{MR3631006-Ma16}]
\label{fact:cauchy-schwarz}
    Let $x, y \in \R^n$ be vectors of indeterminates. Then
    \[ \proves{x, y}{4} \left\{ \langle x, y \rangle^2 \leq \norm{x}^2 \norm{y}^2 \right\} \,.\]
    Furthermore, if $y \in \R^n$ is a constant vector, then
    \[ \proves{x}{2} \left\{ \langle x, y \rangle^2 \leq \norm{x}^2 \norm{y}^2 \right\} \,.\]
\end{fact}

\begin{fact}[Almost triangle inequality, see fact A.6 in~\cite{MR3826314-Hopkins18}]
\label{fact:almost_triangle}
    Let $x_1, \dots, x_r \in \R$ be indeterminates. Then 
    \[ \proves{x_1, \dots, x_r}{2t} \left\{ \left( \sum_{i = 1}^r x_i \right)^{2t} \leq r^{2t-1} \sum_{i = 1}^r x_i^{2t} \right\}\,. \]
\end{fact}

\begin{fact}[AM-GM, see Lemma A.1 in~\cite{MR3388192-Barak15}]
\label{fact:amgm}
Let $x_1, \ldots, x_{r} \in \R$ be indeterminates. Then
\[\Set{x_1 \geq 0, \ldots, x_r \geq 0} \proves{x_1,\ldots,x_r}{r} \Set{\prod_{i=1}^{r} x_i \leq \frac{1}{r} \sum_{i=1}^r x_i^{r}}\,. \]
\end{fact}

\begin{fact}[Square]
\label{fact:square}
Let $x, y \in \R$ be indeterminates. Then
\[\Set{x + y \geq 0, x \leq y} \proves{x}{2} \Set{x^2 \leq y^2}\,.\]
\end{fact}
\begin{proof}
\[\proves{x,y}{2} y^2 - x^2 = (y-x)(x+y) \geq 0\,.\]
\end{proof}

\begin{fact}[Root, see Lemma A.3 in~\cite{kothari2017outlier}]
\label{fact:square-root}
Let $x \in \R$ be an indeterminate and let $C \geq 0$ be a constant. Then
\[\Set{x^{2t} \leq C^{2t}} \proves{x}{2t} \Set{x \leq C}\,.\]
\end{fact}

\begin{fact}[Simplification, see Lemma 9.3 in~\cite{bakshi2020outlierrobust}]
\label{fact:simplification}
Let $x, C \in \R$ be indeterminates. Then
\[\Set{x \geq 0, x^t \leq C x^{t-1}} \proves{x,C}{2t} \Set{x^{2t} \leq C^{2t}}\,.\]
\end{fact}

\begin{fact}[Spectral norm bound]
\label{fact:spectral_norm_bound}
    Let $x \in \R^n$ be a vector of indeterminates and let $A \in \R^{n \times n}$ be a constant symmetric matrix. Then
    \[ \proves{x}{2} \left\{ x^\top A x \leq \norm{A} \norm{x}^2\right\}\,.\]
\end{fact}
\begin{proof}
Let the eigendecomposition of $A$ be $\sum_{i} \lambda_i v_i v_i^\top$ with $|\lambda_i| \leq \norm{A}$ for all $i$. Then
\[\proves{x}{2} x^\top A x = \sum_i \lambda_i \langle v_i, x\rangle^2 \leq \norm{A} \sum_i \langle v_i, x\rangle^2 = \norm{A} \norm{x}^2\,.\]
\end{proof}

\begin{fact}[Bilinear form bound]
Let $x, y \in \R^n$ be vectors of indeterminates and let $A \in \R^{n \times n}$ be a constant symmetric matrix. Then
\label{fact:bilform}
\[\proves{x, y}{4} \Set{(x^\top A y)^2 \leq \norm{A}^2 \norm{x}^2 \norm{y}^2}\,.\]
\end{fact}
\begin{proof}
We have by~\Cref{fact:cauchy-schwarz} and~\Cref{fact:spectral_norm_bound} that 
\[\proves{x,y}{4} (x^\top A y)^2 \leq \norm{x}^2 \norm{Ay}^2 \leq \norm{A}^2 \norm{x}^2 \norm{y}^2\,.\]
\end{proof}

\begin{fact}[Expectation Cauchy-Schwarz]
\label{fact:expect-cauchy-schwarz}
Let $x \in \R^n$ be a vector of indeterminates and let $\bm{a} \in \R^m$ be a random variable supported on $\cA$ such that all its moments are finite. 
Let $p_a(x) = p(x,a)$ and $q_a(x) = q(x,a)$ be polynomials in $x$ and $a$ that have degrees \emph{in $x$} at most $r$ and $s$, respectively. Then we have
\[\proves{v}{2(r+s)} \Set{\Paren{\E p_{\bm{a}}(x) q_{\bm{a}}(x)}^2 \leq \E p_{\bm{a}}(x)^2 \E q_{\bm{a}}(x)^2}\,.\]
\end{fact}
\begin{proof}
Let $\mu$ be the probability measure associated with $\bm{a}$. We have 
\begin{align*}
\proves{v}{2(r+s)}
&2\E p_{\bm{a}}(x)^2 \E q_{\bm{a}}(x)^2 - 2\Paren{\E p_{\bm{a}}(x) q_{\bm{a}}(x)}^2\\
&= 2\Paren{\int_{a} p_{a}(x)^2 d\mu(a)} \Paren{\int_{a} q_{a}(x)^2 d\mu(a)} - 2\Paren{\int_{a} p_{a}(x) q_a(x) d\mu(a)}^2\\
&= \int_a \int_{a'} \Paren{ p_a(x)^2 q_{a'}(x)^2 + p_{a'}(x)^2 q_a(x)^2 - 2p_a(x)q_a(x)p_{a'}(x)q_{a'}(x)} d\mu(a)d\mu(a')\\
&= \int_a \int_{a'} \Paren{p_a(x) q_{a'}(x) - p_{a'}(x)q_a(x)}^2 d\mu(a)d\mu(a')\\
&\geq 0\,.
\end{align*}
To explain the result more explicitly, let $z_{a,a'}(x) = \Paren{p_a(x) q_{a'}(x) - p_{a'}(x)q_a(x)}^2$. We argued that $z_{a,a'}(x)$ is the square of a polynomial in $x$ of degree at most $2(r+s)$.
Therefore, for each $a,a' \in \cA$ we can write $z_{a,a'}(x) = \langle M_{a,a'}, (1,x)^{\otimes 2(r+s)}\rangle$
for some positive semidefinite matrix $M_{a,a'}$ whose entries are polynomials in $a,a'$. Then
\[\int_{a,a'} \Iprod{ M_{a,a'}, (1,x)^{\otimes 2(r+s)} }  d\mu(a)d\mu(a') = \Iprod{\int_{a,a'} M_{a,a'} d\mu(a)d\mu(a'), (1,x)^{\otimes 2(r+s)}}\,,\]
which is a sum of squares because $\int_{a,a'} M_{a,a'} d\mu(a)d\mu(a')$ is also a positive semidefinite matrix.
\end{proof}

\begin{fact}[Jensen's inequality for even powers]
\label{fact:jensen-pow2}
Let $x \in \R^n$ be a vector of indeterminates and let $\bm{a} \in \R^m$ be a random variable supported on $\cA$ such that all its moments are finite. Let  $p_a(x) = p(x,a)$ be a polynomial in $x$ and $a$ that has degree \emph{in $x$} at most $r$. Then we have
\[\proves{x}{2tr} \Set{\Paren{\E p_{\bm{a}}(x)^2}^{t} \leq \E p_{\bm{a}}(x)^{2t}}\,.\]
\end{fact}
\begin{proof}
Let $\mu$ be the probability measure associated with $\bm{a}$. We have
\begin{align*}
\proves{v}{2tr}
&\E p_{\bm{a}}(x)^{2t} - \Paren{\E p_{\bm{a}}(x)^2}^{t}\\
&= \Paren{\int_{a} p_{a}(x)^{2t} d\mu(a)} - \Paren{\int_{a} p_{a}(x)^2  d\mu(a)}^{t}\\
&= \int_{a_1, \ldots, a_{t}} \Paren{ \frac{1}{t} \Paren{p_{a_1}(x)^{2t} + \ldots + p_{a_{t}}(x)^{2t}} -  p_{a_1}(x)^2\cdots p_{a_{t}}(x)^2} d\mu(a_1) \cdots d\mu(a_{t})\,,
\end{align*}
where for each term in the integral we have by~\Cref{fact:amgm}
\begin{align*}
\proves{v}{2tr} \frac{1}{t} \Paren{p_{a_1}(x)^{2t} + \ldots + p_{a_{t}}(x)^{2t}} -  p_{a_1}(x)^2\cdots p_{a_{t}}(x)^2 \geq 0\,.
\end{align*}
To explain the result more explicitly, let
\[z_{a_1, \ldots, a_t}(x) = \frac{1}{t} \Paren{p_{a_1}(x)^{2t} + \ldots + p_{a_{t}}(x)^{2t}} -  p_{a_1}(x)^2\cdots p_{a_{t}}(x)^2\,.\]
We argued that $z_{a_1, \ldots, a_t}(x)$ is a sum of squares of polynomials in $x$ each of degree at most $2tr$. Therefore, for each $a_1,\ldots,a_t \in \cA$ we can write $z_{a_1, \ldots, a_t}(x) = \langle M_{a_1,\ldots,a_t}, (1,x)^{\otimes 2tr}\rangle$
for some positive semidefinite matrix $M_{a_1,\ldots,a_t}$ whose entries are polynomials in $a_1, \ldots, a_t$. Then
\[\int_{a_1,\ldots,a_t} \Iprod{ M_{a_1,\ldots,a_t}, (1,x)^{\otimes 2tr} }  d\mu(a_1)\cdots d\mu(a_t) = \Iprod{\int_{a_1,\ldots,a_t} M_{a_1,\ldots,a_t} d\mu(a_1)\cdots d\mu(a_t), (1,x)^{\otimes 2tr}}\,,\]
which is a sum of squares because $\int_{a_1,\ldots,a_t} M_{a_1,\ldots,a_t} d\mu(a_1)\cdots d\mu(a_t)$ is also a positive semidefinite matrix.
\end{proof}

\subsection{Gaussian Distributions}
In this section, we restate some existing results about Gaussians and mixtures of Gaussians that we use in our analysis.

First, we state a relation between total variation distance and parameter distance developed in~\cite{bakshi2020outlierrobust}:

\begin{fact}[Gaussian TV distance and parameters, Proposition A.1 in~\cite{bakshi2020outlierrobust}]
\label{fact:tv-parameters}
    Fix $\Delta > 0$ and let $\mu, \mu'$ and $\Sigma, \Sigma' \succ 0$ satisfy:
    \begin{enumerate}
        \item Mean closeness: $\forall v \in \mathbb{R}^d$, $\langle \mu - \mu', v \rangle^2 \leq \Delta^2 \cdot v^\top (\Sigma + \Sigma') v$,
        \item Spectral closeness: $\forall v \in \mathbb{R}^d$, $\frac{1}{\Delta^2} v^\top \Sigma v \leq v^\top \Sigma' v \leq \Delta^2 \cdot v^\top \Sigma v$,
        \item Relative Frobenius closeness: $\Norm{\Sigma^{-1/2} \Sigma' \Sigma^{-1/2} - I}_F^2 \leq \Delta^2 \cdot \Norm{\Sigma^{-1/2} \Sigma' \Sigma^{-1/2}}^2$.
    \end{enumerate}
    Then $d_{\TV}(N(\mu, \Sigma), N(\mu', \Sigma')) \leq 1 - \exp(-O(\Delta^2 \log \Delta))$.
\end{fact}

Second, we state the guarantees of an algorithm of~\cite{diakonikolas2019robust} for robustly learning a Gaussian:

\begin{fact}[Robust Gaussian learning, Theorem 1.2 in~\cite{diakonikolas2019robust}]
\label{fact:gaussian-robust-learning}
    Let $\mu, \Sigma$ be arbitrary and unknown, and let $\epsilon, \tau > 0$.
    There is a polynomial time algorithm which given $\epsilon$, $\tau$, and an \mbox{$\epsilon$-corrupted} set of $n$ samples from $N(\mu, \Sigma)$ with $n \geq \Tilde{\Omega}\left(\frac{d^2 \log^5(1/\tau)}{\epsilon^2}\right)$, produces $\hat{\mu}$ and $\hat{\Sigma}$ such that with probability $1-\tau$ we have
    \[d_{\TV}\left(N(\mu, \Sigma), N(\hat{\mu}, \hat{\Sigma}\right) \leq O(\epsilon \log^{3/2}(1/\epsilon))\,.\]
\end{fact}

For mixtures of Gaussians, we start by stating the identifiability and recovery results in the state-of-the-art paper of~\cite{MR4490075-Bakshi22}:

\begin{fact}[Gaussian mixture identifiability, Theorem 9.1 in~\cite{MR4490075-Bakshi22}]
\label{fact:param-identifiability}
    Let $\mathcal{M} = \sum_{i=1}^{k_1} w_i G_i$, $\mathcal{M}' = \sum_{i=1}^{k_2} w_i' G_i'$ be two mixtures of Gaussians such that $d_{\TV}(\mathcal{M}, \mathcal{M}') \leq \epsilon$.
    Then there exists a partition of $[k_1]$ into sets $R_0, R_1, \ldots, R_\ell$ and a partition of $[k_2]$ into sets $S_0, S_1, \ldots, S_\ell$ such that 
    \begin{enumerate}
        \item Let $W_i = \sum_{j \in R_i} w_j$ for $i = 0, 1, \ldots, k_1$, $W_i' = \sum_{j \in S_i} w_j'$ for $i = 0, 1, \ldots, k_2$. Then for all $i \in [\ell]$,
        \begin{align*}
            \Abs{W_i - W_i'} &\leq \poly_k(\epsilon)\,,\\
            d_{\TV}(G_j, G_{j'}') &\leq \poly_k(\epsilon)\,, \; \; \forall j \in R_i, j' \in S_i\,,
        \end{align*}
        \item $W_0, W_0' \leq \poly_k(\epsilon)$.
    \end{enumerate}
\end{fact}

\begin{theorem}[Gaussian mixture recovery, Theorem 1.6 in~\cite{MR4490075-Bakshi22}]
\label{fact:parameter-recovery}
    Given $\epsilon > 0$ and a multiset of $n = d^{O(k)} \poly_k(1/\epsilon)$ samples from a distribution $F$ on $\mathbb{R}^d$ such that $d_{\TV}(F, \mathcal{M}) \leq \epsilon$, for an unknown target $k$-GMM $\mathcal{M} = \sum_{i=1}^k w_i \mathcal{N}(\mu_i, \Sigma_i),$ the algorithm runs in time $\poly(n)\poly_k(1/\epsilon)$ and outputs a $k'$-GMM hypothesis $\hat{\mathcal{M}} = \sum_{i=1}^{k'} \hat{w_i} \mathcal{N}(\hat{\mu_i}, \hat{\Sigma_i})$, with $k' \leq k$ such that with high probability we have that there exists a partition of $[k]$ into $k' + 1$ sets $R_0, R_1, \ldots, R_{k'}$ such that 
    \begin{enumerate}
        \item Let $W_i = \sum_{j \in R_i} w_j$, $i \in \{0,1, \ldots, k'\}.$ Then, for all $i \in [k']$, we have that 
        \begin{align*}
            \Abs{W_i - \hat{w_i}} &\leq \poly_k(\epsilon)\,, \\
            d_{\TV}(\mathcal{N}(\mu_j, \Sigma_j), \mathcal{N}(\hat{\mu_i}, \hat{\Sigma_i})) &\leq \poly_k(\epsilon)\,, \; \; \forall j \in R_i \,.
        \end{align*}
        \item The total weight of exceptional components in $R_0$ is $W_0 \leq \poly_k(\epsilon)$.
    \end{enumerate}
\end{theorem}

The same paper also gives a polynomial-time algorithm for putting a mixture of Gaussians in the isotropic position:

\begin{fact}[Gaussian mixture robust isotropic position, implied by Lemma 6.7 in~\cite{MR4490075-Bakshi22}]
\label{thm:robust-isotropic-position}
Let $0 < \epsilon < 1$,\footnote{The result assumes that $\epsilon > 0$ is dimension-independent.} $k \in \N$, and $\alpha = \epsilon^{1/(10C^{k+1}(k+1)!)}$.
Let $\mathcal{M} = \sum_{i=1}^k w_i N(\mu_i, \Sigma_i)$ be a mixture with $w_i \geq \alpha$ for all $i \in [k]$.
Let $\mu$ and $\Sigma$ be the mean and covariance of $\mathcal{M}$ such that $r = \rank(\Sigma)$ and for all $i,j\in[k]$, $\Norm{\Sigma^{\dagger/2}(\Sigma_i-\Sigma_j)\Sigma^{\dagger/2}}_F \leq 1/\sqrt{\alpha}$.
Let $X$ be a set of points from $\mathcal{M}$.
Given a set $Y$, an $\epsilon$-corrupted version of $X$, of size $n \geq n_0 = d^{O(1)}$, there exists an algorithm that takes $Y$ as input and in time $n^{O(1)}$ outputs with high probability estimators $\hat\mu$ and $\hat\Sigma$ such that $\hat\Sigma = \hat{U} \hat{\Lambda} \hat{U}^\top$ is the eigenvalue decomposition, where $\hat{U} \in \R^{n \times r}$ has orthonormal columns and $\Lambda \in \R^{r \times r}$ is a diagonal matrix.
Further, we can obtain $n$ samples $Y'$ by applying the affine transformation $y_i \to \hat{U}^\top \hat{\Sigma}^{\dagger/2}(y_i-\hat{\mu})$ to each sample, such that a $(1-\epsilon)$-fraction have mean $\mu'$ and covariance $\Sigma'$ satisfying
\begin{enumerate}
    \item $\Norm{\mu'} \leq O\Paren{\Paren{1+\frac{\sqrt{\epsilon}k}{\alpha}}\sqrt{\epsilon/\alpha}}$\,,
    \item $\Paren{\frac{1}{1+(k\sqrt{\epsilon}/\alpha)}} I_r \preceq \Sigma' \preceq \Paren{\frac{1}{1-(k\sqrt{\epsilon}/\alpha)}} I_r$\,,
    \item $\Norm{\Sigma' - I_r}_F \leq O(\sqrt{\epsilon}k/\alpha)$\,,
\end{enumerate}
where $I_r$ is the $r$-dimensional identity matrix, and the remaining points are arbitrary.
Let $X'$ be the set obtained by $\hat{U}^\top \hat{\Sigma}^{\dagger/2}(x_i-\hat{\mu})$. 
Then, the points in $X'$ are distributed as a set of i.i.d. points from the mixture $\sum_{i=1}^k w_i N(\hat{U}^\top \hat{\Sigma}^{\dagger/2}(\mu_i-\hat{\mu}), \hat{U}^\top \hat{\Sigma}^{\dagger/2} \Sigma_i \hat{\Sigma}^{\dagger/2} \hat{U})$, and $Y'$ is an $\epsilon$-corruption of $X'$.
\end{fact}

Finally, \cite{MR4490075-Bakshi22} also give a polynomial-time algorithm for partially clustering a mixture of Gaussians in which two components have Frobenius separation:

\begin{fact}[Gaussian mixture robust Frobenius clustering, Theorem 4.3 in~\cite{MR4490075-Bakshi22}]
\label{fact:frobenius-clustering}
    Let $0 \leq \epsilon < \alpha / k \leq 1$ and $t \in \mathbb{N}$.
    There is an algorithm with the following guarantees: Let $\{\mu_i, \Sigma_i\}_{i \leq k}$ be means and covariances of $k$ unknown Gaussians.
    Let $Y$ be an $\epsilon$-corruption of a sample $X$ of size $n \geq \poly(d^t, k^{k+t}, \epsilon^{-1})$ from $\mathcal{M} = \sum_i w_i N(\mu_i, \Sigma_i)$.
    Suppose further that $w_i \geq \alpha > 2 \epsilon$ for each $i \in [k]$, and that for some $t \in \mathbb{N}, \beta > 0$ there exist $i, j \leq k$ such that $\Norm{\Sigma^{\dagger/2} (\Sigma_i - \Sigma_j) \Sigma^{\dagger/2}}_F^2 = \Omega\left((k^2t^4)/(\beta^{2/t} \alpha^4)\right)$, where $\Sigma$ is the covariance of the mixture $\mathcal{M}$.
    Then, the algorithm runs in time $n^{O(t)}$, and with high probability over the input, and with probability at least $2^{-O\left(\frac{1}{\alpha} \log \left( \frac{k}{\beta} \right) \right)}$ over the random choices of the algorithm, outputs a partition $Y = Y_1 \cup Y_2$ satisfying:
    \begin{enumerate}
        \item Partition respects clustering: for each $i$,
        \[\max \left\{ \frac{1}{w_i n} \Abs{Y_1 \cap X_i}, \frac{1}{w_i n} \Abs{Y_2 \cap X_i}\right\} \geq 1 - \beta - O(\epsilon/\alpha^4)\,,\]
        \item Partition is non-trivial:
        \[\max_i \frac{1}{w_i n} \Abs{X_i \cap Y_1}, \max_i \frac{1}{w_i n} \Abs{X_i \cap Y_2} \geq 1 - \beta - O(\epsilon/\alpha^4)\,.\]
    \end{enumerate}
\end{fact}

\subsection{Hypothesis Selection}

Our algorithm will produce a large list of candidate mixtures, such that at least one of them is close to the ground truth.
The following hypothesis selection result allows us to select a good mixture from this list of candidate mixtures:

\begin{fact}[Robust tournament, Lemma 2.9 in~\cite{kane2021robust}]
\label{fact:robust-tournament}
    Let $X$ be an unknown distribution and let $H_1, \ldots, H_n$ be distributions with explicitly computable probability density functions that can be efficiently sampled from.
    Assume furthermore that $\min_{1 \leq i \leq n}(d_{\TV}(X, H_i)) \leq \epsilon$.
    Then there exists an algorithm that given access to $\poly(n,\epsilon^{-1})$ $\epsilon$-noisy samples from $X$ along with $H_1, \ldots, H_n$ computes in time $\poly(n, \epsilon^{-1})$ a $1 \leq m \leq n$ so that with high probability 
    \[ d_{\TV}(X, H_m) = O(\epsilon) \,.\]
\end{fact}

\subsection{Partial Clustering Definitions}

In this section we introduce some definitions and notation relating to clustering that we use in our analysis.

\begin{definition}[Containment]
    For a collection of samples $\mathcal{T}$ coming from a mixture of components and for a subset $S \subseteq \mathcal{T}$, a component $i$ is contained in $S$ if a $(1-w_{\min})$-fraction of the samples in $\mathcal{T}$ from component $i$ are in $S$.
    Let $\mathsf{comp}(S)$ denote the indices of components contained in $S$.
\end{definition}

\begin{definition}[Corruptions]
    For a collection of samples $\mathcal{T}$ coming from a mixture of components and for a subset $S \subseteq \mathcal{T}$, let $\mathsf{corr}(S)$ denote the number of corrupted samples in $S$, defined as the number of samples in $S$ that do not belong to a component in $\mathsf{comp}(S)$ plus the number of samples in $\mathcal{T} \setminus S$ that belong to a component in $\mathsf{comp}(S)$.
\end{definition}

\begin{definition}[$\epsilon$-good partial clustering]
    For a collection of samples $\mathcal{T}$ coming from a mixture of components, we call a set $\mathcal{S} = \{S_1, \ldots, S_{k'}\}$ with $S_1, \ldots, S_{k'} \subseteq \mathcal{T}$ a partial clustering of the samples if $S_1, ..., S_{k'}$ are disjoint, $S_1 \cup \ldots \cup S_{k'} = \mathcal{T}$, and $\mathsf{comp}(S_1), \ldots, \mathsf{comp}(S_{k'}) > 0$.
    We say that a partial clustering $\mathcal{S}$ is $(1-\epsilon)$-good if all $S \in \mathcal{S}$ have $\mathsf{corr}(S)/|S| \leq \epsilon$.
\end{definition}

\begin{definition}[Refinement of a partial clustering]
    For a collection of samples $\mathcal{T}$ coming from a mixture of components, given two partial clusterings of the samples $\mathcal{S}, \mathcal{S}'$, we say that $\mathcal{S}'$ is a refinement of $\mathcal{S}$ if for all $S' \in \mathcal{S}'$ there exists some $S \in \mathcal{S}$ such that $S' \subseteq S$ and if $\Abs{\mathcal{S'}} > \Abs{\mathcal{S}}$.
\end{definition}

\section{Efficiently Identifying Relevant Subspaces via SoS}
\label{section:subspace-rounding}
In this section we present two general tools for identifying ``relevant'' low-dimensional subspaces given input data.
These results form the technical core of our dimension reduction subroutine. 

For the first result, suppose we have a system of polynomial constraints $\cA$ (built from the input data) that implies that all solutions to $\cA$ are (approximately) contained in a low-dimensional subspace.
Then, if this implication holds in the low-degree sum-of-squares proof system, we show that there exists an efficient algorithm that recovers a low-dimensional subspace that (approximately) contains all solutions to $\cA$.

\begin{theorem}[Identifying subspace approximately containing all solutions]
\label{thm:sos-rounding}
There is an algorithm that takes input a system of polynomial inequalities $\mathcal{A} = \{\norm{v}^2 = 1, p_1(v) \geq 0, \ldots, p_m(v) \geq 0\}$ in indeterminate $v \in \R^d$ and outputs an orthogonal projection matrix $Q \in \R^{d \times d}$ with the following guarantee:

Suppose there exists an orthogonal projection matrix $P \in \R^{d \times d}$ of rank $r$ such that\linebreak $\mathcal{A} \proves{v}{t} \{\|P v\|^2 \geq 1-\epsilon\}$ for some $\epsilon < 1$ small enough.
Then, for $\gamma \geq 2\sqrt{2} \epsilon^{1/8}$, the algorithm runs in time $(dm)^{O(t)} \cdot O(1/\gamma)^{r}$ and outputs an orthogonal projection matrix $Q \in \R^{d \times d}$ with $\rank(Q) \leq O(1/\gamma)^{r}$ such that with high probability over the random choices of the algorithm, for every unit vector $v\in \R^d$ satisfying $\cA$, there exists a unit vector $v'$ such that $Qv' = v'$ and $\Norm{v-v'} \leq \gamma$.
\end{theorem}

For the second result, suppose instead that $\cA$ implies that all solutions to $\cA$ are (approximately) \emph{orthogonal} to a subspace of interest.
If, in addition, we have the converse fact that there exists a low-dimensional subspace such that all vectors orthogonal to it are solutions to $\cA$, then we give an efficient algorithm that recovers a low-dimensional subspace that (approximately) contains the original subspace of interest.

In our proofs we actually need a slightly more general form of this result, in which solutions to $\cA$ are (approximately) orthogonal to \emph{multiple} low-dimensional subspaces.
We state this more general version.

\begin{theorem}[Identifying subspace approximately orthogonal to all solutions]
\label{thm:sos-rounding-2}
There is an algorithm that takes input a system of polynomial inequalities $\mathcal{A} = \{\norm{v}^2 \leq 1, p_1(v) \geq 0, \ldots, p_m(v) \geq 0\}$ in indeterminate $v \in \R^d$ and outputs an orthogonal projection matrix $Q \in \R^{d \times d}$ with the following guarantee:

Suppose there exist orthogonal projection matrices $P_1, P_2, \ldots, P_k \in \R^{d \times d}$ such that ${\mathcal{A} \proves{v}{t} \{\|P_i v\|^2 \leq \epsilon\}}$ for all $i \in [k]$ and for some $\epsilon < 1$ small enough.
Suppose further that there exists an orthogonal projection matrix $R \in \R^{d \times d}$ of rank $r$ such that \[\Set{\norm{v}^2 \leq 1, Rv=0} \proves{v}{t} \cA\,.\] 

Then, for $\gamma \geq 4 \epsilon^{1/8}$, the algorithm runs in time $(dm)^{O(t^3)} \cdot O(1/\gamma)^{r}$ and outputs an orthogonal projection matrix $Q \in \R^{d \times d}$ with $\rank(Q) \leq O(1/\gamma)^{r}$ such that with high probability over the random choices of the algorithm, for every unit vector $v\in \R^d$ such that $\|P_iv\|^2 \geq 1-\epsilon$ for some $i \in [k]$, there exists a unit vector $v'$ such that $Qv' = v'$ and $\Norm{v-v'} \leq \gamma$.

Furthermore, if in each sum-of-squares proof ${\mathcal{A} \proves{v}{t} \{\|P_i v\|^2 \leq \epsilon\}}$ each monomial $r(v)^2 \prod_{i \in S} p_i(v)$ has $|S| \leq q$, then the time complexity is $(dm)^{O(t^2(q+1))} \cdot O(1/\gamma)^{r}$.
\end{theorem}

\subsection{Proof of~\Cref{thm:sos-rounding}}

We prove~\Cref{thm:sos-rounding} by analyzing the following algorithm:

\begin{algorithm}[H]\label{alg:rounding-subspace}
    \SetAlgoLined
    $Q = 0^{d\times d}$\;
    \For{$i=1$ to $O(1/\gamma)^{r}$}{
        \If{$\exists$ degree-$t$ pseudo-expectation $\tilde{\mathbb{E}}$ that satisfies $\mathcal{A} \cup \{\norm{Qv}^2 \leq 1-\gamma^2/2\}$}
        {Let $M = \tilde{\mathbb{E}} vv^\top$\;
        Fix an eigendecomposition of $M$, and let $\mathcal{D}_M$ be the distribution over the eigenvectors in the decomposition such that the probability of each eigenvector is proportional to its corresponding eigenvalue\;
        Sample $w$ from $\mathcal{D}_M$ conditioned on $\norm{Qw}^2 \leq 1-\gamma^2/4$\; 
        Update $Q$ to be the orthogonal projection to the subspace spanned by $Q$ and $w$\;}
        \Else {\Return{Q}\;}
}
\caption{Subspace rounding}
\end{algorithm}
\begin{proof}[Proof of~\Cref{thm:sos-rounding}]
First note that $\rank(Q) \leq O(1/\gamma)^r$, because at each iteration of the for loop we only increase the dimension of the subspace associated with $Q$ by $1$.

Next, we prove that as long as the algorithm finds a degree-$t$ pseudo-expectation $\tilde{\E}$ that satisfies\linebreak $\mathcal{A} \cup \{\norm{Qv}^2 \leq 1-\gamma^2/2\}$, it samples with probability $1-\epsilon^{1/4}$ some $w$ such that $\norm{Pw}^2 \geq 1-\sqrt{\epsilon}$ and $\norm{Qw}^2 \leq 1-\gamma^2/4$.
Because $\tilde{\E}$ is a pseudo-expectation of degree at least $2$, we have that $M \succeq 0$.
Also, because $\cA \proves{v}{t} \norm{v}^2=1$, we have that $\Tr(M) = 1$.
Fix an eigenvalue decomposition of $M$ as $M = \lambda_1 z_1z_1^\top + \ldots + \lambda_d z_d z_d^\top$ where $\norm{z_i}  = 1$ for all $i$.
Then $\Tr(M)=1$ implies that $\lambda_1 + ... + \lambda_d = 1$. 

Because $\mathcal{A} \proves{v}{t} \|P v\|^2 \geq 1-\epsilon$,  we have that $\langle P, M \rangle \geq 1-\epsilon$, so $\sum_{i=1}^d \lambda_i \langle P, z_iz_i^\top\rangle \geq 1-\epsilon$.
Let $S = \{i \mid \langle P, z_iz_i^\top\rangle \geq 1-\sqrt{\epsilon}\}$.
Then we must have $\sum_{i \in S} \lambda_i \geq 1 -\sqrt{\epsilon}$ in order to avoid a contradiction.
Similarly, we have that $\langle Q, M\rangle \leq 1-\gamma^2/2$, so $\sum_{i=1}^d \lambda_i \langle Q, z_iz_i^\top\rangle \leq 1-\gamma^2/2$.
Let $S' = \{i \mid \langle Q, z_iz_i^\top \rangle \leq 1-\gamma^2/4\}$.
Then we must have $\sum_{i \in S'} \lambda_i \geq \gamma^2/4$ in order to avoid a contradiction. Combining these,
\[\sum_{i \in S \cap S'} \lambda_i \geq \gamma^2/4 - \sqrt{\epsilon} \geq \gamma^2/8\,,\]
where we used that $\gamma^2 \geq 8\sqrt{\epsilon}$.
We argue now that the algorithm samples $w$ with probability $1-\epsilon^{1/4}$ as some $z_i$ with $i \in S \cap S'$, which implies that $\norm{Pw}^2 \geq 1-\sqrt{\epsilon}$ and $\norm{Qw}^2 \leq 1-\gamma^2/4$.
By definition, the algorithm always samples $w$ as some $z_i$ with $i \in S'$.
In addition, $\sum_{i \not\in S} \lambda_i \leq \sqrt{\epsilon}$.
Therefore, the algorithm also samples $w$ as some $z_i$ with $i \in S$ with probability at least $\frac{\gamma^2/8}{\gamma^2/8 + \sqrt{\epsilon}} \geq 1 - \epsilon^{1/4}$, where we used that $\gamma^2 \geq 8\epsilon^{1/4}$.

Next, we prove that if at some iteration there does not exist any degree-$t$ pseudo-expectation that satisfies $\mathcal{A} \cup \{\norm{Qv}^2 \leq 1-\gamma^2/2\}$, then the guarantees of the theorem hold. 
Indeed, such a pseudo-expectation exists as long as there exists any unit vector $v$ that satisfies $\mathcal{A}$ and $\norm{Qv}^2 \leq 1-\gamma^2/2$.
If no unit vector $v$ satisfies $\mathcal{A}$, then the algorithm stops at the first iteration and returns $0^{d\times d}$, which trivially satisfies the guarantees of the theorem.
On the other hand, if at some iteration $\norm{Qv}^2 > 1-\gamma^2/2$ for all unit vectors $v$ that satisfy $\cA$, then each such $v$ has inner product at least $1-\gamma^2/2$ with some unit vector in the subspace associated with $Q$, so each such $v$ has distance at most $\gamma$ from some unit vector in the subspace associated with $Q$.
This is the desired guarantee.

Finally, we show that with high probability there exists an iteration of the for loop at which there exists no degree-$t$ pseudo-expectation that satisfies $\mathcal{A} \cup \{\norm{Qv}^2 \leq 1-\gamma^2/2\}$.
Call an iteration good if the algorithm samples a unit vector $w$ with ${\norm{Pw}^2 \geq 1-\sqrt{\epsilon}}$ and $\norm{Qw}^2 \leq 1-\gamma^2/4$.
We start by upper bounding the number of good iterations of the algorithm.
In a good iteration, $w$ is $\sqrt{2}\epsilon^{1/4}$-close to the subspace associated with $P$ and $\gamma/2$-far from the subspace associated with $Q$.
Let $z$ be the projection of $w$ onto the subspace associated with $P$.
Then $z$ is $(\gamma/2-\sqrt{2}\epsilon^{1/4})$-far from the subspace associated with $Q$, so it is also $(\gamma/2-\sqrt{2}\epsilon^{1/4})$-far from any $w$ produced in a previous good iteration, so by the triangle inequality it is also $(\gamma/2-2\sqrt{2}\epsilon^{1/4})$-far from any $z$ produced in a previous good iteration.
We bound $\gamma/2-2\sqrt{2}\epsilon^{1/4} \geq \gamma/4$ using that $\gamma^2 \geq 128 \epsilon^{1/2}$.
Then the number of good iterations is bounded by the size of a packing of the subspace associated with $P$ with balls of radius $\gamma/4$, which is at most $O(1/\gamma)^r$.
To finish, we upper bound the number of bad iterations.
An iteration is bad with probability at most $\epsilon^{1/4}$, so then by standard concentration bounds with high probability the number of bad iterations is bounded by the number of good iterations.
Therefore with high probability there exists an iteration among the first $O(1/\gamma)^r$ ones where the algorithm cannot find a degree-$t$ pseudo-expectation that satisfies $\mathcal{A} \cup \{\norm{Qv}^2 \leq 1-\gamma^2/2\}$, and as we argued this implies that with high probability the guarantees of the theorem hold.
\end{proof}

\subsection{Proof of~\Cref{thm:sos-rounding-2}}

We first state a lemma that, given the system of polynomial constraints $\cA$, shows that we can construct a system of polynomial constraints $\cB$ over unit vectors such that (1) if a vector satisfies $\cB$, then it is nearly contained in the subspace described by $R$, and (2) if a vector is nearly contained in one of the subspaces described by $P_1, \ldots, P_k$, then it satisfies $\cB$.
Then the proof of~\Cref{thm:sos-rounding-2} consists of applying the rounding algorithm in~\Cref{thm:sos-rounding} to the system $\cB$.

\begin{lemma}[Sum-of-squares orthogonal complement]
\label{lem:tmp-sos-complement}
There is an algorithm that takes input a system of polynomial inequalities $\mathcal{A}(v) = \{\norm{v}^2 \leq 1, p_1(v) \geq 0, \ldots, p_m(v) \geq 0\}$ in indeterminate $v \in \R^d$ and outputs a system of polynomial inequalities $\mathcal{B}(u)$ in indeterminate $u \in \mathbb{R}^d$ with the following guarantee:

Suppose there exist orthogonal projection matrices $P_1, P_2, \ldots, P_k \in \R^{d\times d}$ such that $\mathcal{A}(v) \proves{v}{t} \{\|P_iv\|^2 \leq \epsilon\}$ for all $i \in [k]$ and for some $\epsilon < 1$ small enough.
Suppose further that there exists an orthogonal projection matrix $R \in \R^{d \times d}$ of rank $r$ such that \[\Set{\norm{v}^2 \leq 1, Rv=0} \proves{v}{t} \cA(v)\,.\] 

Then the algorithm runs in time $(dm)^{O(t)}$ and outputs a system of polynomial inequalities $\mathcal{B}(u)$ of size $(dm)^{O(t)}$ in indeterminate $u \in \mathbb{R}^d$ such that 
\begin{enumerate}
    \item $\mathcal{B}(u) \proves{u}{2} \{\norm{u}^2 = 1\}$\,,
    \item $\mathcal{B}(u) \proves{u}{t^2} \{\|R u\|^2 \geq 1-4\epsilon\}$\,,
    \item All unit vectors $w \in \R^d$ such that $\norm{P_i w}^2 \geq 1-\epsilon$ for some $i \in [k]$ are feasible for $\mathcal{B}(w)$\,.
\end{enumerate}   
Furthermore, if in each sum-of-squares proof ${\mathcal{A}(v) \proves{v}{t} \{\|P_i v\|^2 \leq \epsilon\}}$ each monomial $r(v)^2 \prod_{i \in S} p_i(v)$ has $|S| \leq q$, then $\mathcal{B}(u) \proves{u}{t(q+1)} \{\|R u\|^2 \geq 1-4\epsilon\}$.
\end{lemma}
\begin{proof}
    Let $\mathcal{B}(u)$ be the following system of polynomial inequalities in indeterminate $u \in \mathbb{R}^d$:
    \begin{enumerate}
        \item $\norm{u}^2 = 1$\,,
        \item There exists a degree-$t$ sum-of-squares proof in indeterminate $v \in \R^d$
        \[ \cA(v) \proves{v}{t} \Set{\langle v, u \rangle^2 \leq 4 \epsilon}\,.\]
    \end{enumerate}
    We note that the second constraint can be encoded with $(dm)^{O(t)}$ polynomial inequalities:
    see for example the discussion about succinct representations in Section 4.3.4 in~\cite{MR4059250-Fleming19}.
    
    We now show that $\cB(u)$ has the desired properties. Property (1) is immediate from the first constraints in $\cB(u)$.

    For property (2), because $\Norm{I_d-R} \leq 1$ we have by~\Cref{fact:spectral_norm_bound} that $\cB(u) \proves{u}{2} \Set{\Norm{(I_d-R)u}^2 \leq \Norm{u}^2 =1}$.
    Further, because $R^2=R$ we have that $R (I_d-R) = 0$ and thus $\proves{u}{0} \Set{R (I_d-R)u = 0}$.
    Therefore, because $\Set{\norm{v}^2 \leq 1, Rv=0} \proves{v}{t} \cA(v)$, by replacing $v$ by $(I_d-R)u$ we get that ${\cB(u) \proves{u}{t} \cA((I_d-R)u)}$.
    By composition with the second constraint in $\cB(u)$, we get that $\mathcal{B}(u) \proves{u}{t^2} \Set{\langle u, (I_d - R)u \rangle^2 \leq 4 \epsilon}$.
    We can rewrite this as $\mathcal{B}(u) \proves{u}{t^2} \Set{\norm{Ru}^2 \geq 1-4\epsilon}$, as desired. 

    Finally, for property (3), let $w \in \R^d$ be a unit vector such that $\norm{P_i w}^2\geq 1-\epsilon$ for some $i \in [k]$.
    Then
    \begin{align*}
    \cA(v) \proves{v}{t} \langle v, w \rangle^2
    &= \langle P_i v + (I_d-P_i) v, w\rangle^2\\
    &\leq 2 \langle P_i v, w\rangle^2 + 2 \langle (I_d - P_i) v, w \rangle^2\\
    &= 2 \langle P_i v, w\rangle^2 + 2 \langle v, (I_d - P_i) w \rangle^2\\
    &\leq 4\epsilon\,,
    \end{align*}
    where in the first inequality we used~\Cref{fact:almost_triangle} and in the last inequality we used Cauchy-Schwarz (\Cref{fact:cauchy-schwarz}) to bound the first term in the sum using that $\cA(v) \proves{v}{t} \Set{\norm{P_i v}^2 \leq \epsilon}$ and the second term in the sum using that $\norm{(I_d - P_i) w}^2 \leq \epsilon$.
    Therefore, $w$ is feasible for $\mathcal{B}(w)$.

    Furthermore, suppose in each sum-of-squares proof ${\mathcal{A}(v) \proves{v}{t} \{\|P_i v\|^2 \leq \epsilon\}}$ each monomial\linebreak $r(v)^2 \prod_{i \in S} p_i(v)$ has $|S| \leq q$.
    Note that in this case in the sum-of-squares proof that we give in the proof of property (3) each monomial $r(v)^2 \prod_{i \in S} p_i(v)$ also has $|S| \leq q$.
    Therefore we can require the same of the sum-of-squares proof in the second constraint of $\cB(u)$.
    Then this implies by composition that the sum-of-squares proof $\mathcal{B}(u) \proves{u}{} \{\|R u\|^2 \geq 1-4\epsilon\}$ has degree at most $t(q+1)$.
\end{proof}

\begin{proof}[Proof of~\Cref{thm:sos-rounding-2}]
By~\Cref{lem:tmp-sos-complement}, there is an algorithm that runs in time $(dm)^{O(t)}$ and constructs a system of polynomial inequalities $\mathcal{B}$ of size $(dm)^{O(t)}$ in indeterminate unit vectors $u \in \mathbb{R}^d$ such that $\mathcal{B}(u) \proves{u}{t^2} \{\|R u\|^2 \geq 1-4\epsilon\}$ and such that all unit vectors $w$ with $\norm{P_i w}^2 \geq 1-\epsilon$ for some $i \in [n]$ are feasible for $\cB(w)$.
Then the result follows by applying the algorithm in~\Cref{thm:sos-rounding} to $\cB$, leading to time complexity $(dm)^{O(t^3)} \cdot O(1/\gamma)^r$.

Furthermore, if in each sum-of-squares proof ${\mathcal{A}(v) \proves{v}{t} \{\|P_i v\|^2 \leq \epsilon\}}$ each monomial $r(v)^2 \prod_{i \in S} p_i(v)$ has $|S| \leq q$, then~\Cref{lem:tmp-sos-complement} implies that $\mathcal{B}(u) \proves{u}{t(q+1)} \{\|R u\|^2 \geq 1-4\epsilon\}$, and then~\Cref{thm:sos-rounding} leads to time complexity $(dm)^{O(t^2(q+1))} \cdot O(1/\gamma)^r$.
\end{proof}

\section{Clustering Mixtures of Centered Gaussians}

In this section, we show how our sum-of-squares based dimension reduction method yields significantly improved algorithms for outlier-robustly clustering mixtures of well-separated centered Gaussians.
Our notion of separation is in a distance (studied in prior works such as~\cite{bakshi2020outlierrobust,dhkk20}) natural for non-spherical mixtures, which for centered Gaussian components corresponds to total variation separation of $1-f(1/w_{\min})$ where $w_{\min}$ is the minimum mixing weight of any component in the input mixture (see~\Cref{fact:tv-parameters}).

\begin{definition}[Centered parameter distance]
We say that two centered distributions with covariances $\Sigma_1, \Sigma_2 \in \R^{d \times d}$ are $\Delta$-separated if at least one of the following conditions hold for $\{i,j\} = \{1,2\}$: 
\begin{itemize}
    \item Spectral separation: $\exists v \in \R^d$ such that $v^\top \Sigma_i v > \Delta^2 \cdot v^\top \Sigma_j v$,
    \item Relative Frobenius separation: $\Norm{\Sigma_i^{-1/2} \Sigma_j \Sigma_i^{-1/2} - I_d}^2_F > \Delta^2 \cdot \Norm{\Sigma_i^{-1/2} \Sigma_j \Sigma_i^{-1/2}}^2$.
\end{itemize}
\end{definition}

Our algorithm has sample and time complexity that are \emph{fixed polynomials} in the underlying dimension $d$.
The best-known prior works require samples and time  
$\gg d^{O(k)}$~\cite{MR4490075-Bakshi22}.
Importantly, unlike prior work~\cite{bakshi2020outlierrobust}, our algorithm \emph{does not need} certifiable anti-concentration properties of the component distributions, and indeed this is the main technical reason why we obtain a fixed polynomial (as opposed to $d^{\poly(k)}$) sample and time complexity.

\begin{theorem}[Main theorem, centered components]
\label{thm:zero-mean-main}
Let $\mathcal{M}$ be a $d$-dimensional mixture of $k$ Gaussians $\sum_{i=1}^k w_i N(0, \Sigma_i)$ with $w_{\min} = \min_{i} w_i$ and $\Sigma_1, \ldots, \Sigma_k \succ 0$.
Furthermore, assume all components $i \neq j$ are \mbox{$\Delta$-separated}.
Also let $\mathcal{M}'$ be a distribution satisfying $d_{\TV}(\mathcal{M}', \mathcal{M}) \leq \epsilon^*$.
Suppose $\Delta, (\epsilon^*)^{-1} \geq f(w_{\min}^{-1})$ for some function $f$.
Then, given $\poly(d, w_{\min}^{-1}, \Delta)$ i.i.d. samples from $\mathcal{M}'$, there exists an algorithm that runs in time $g(w_{\min}^{-1}, \Delta) \cdot \poly(d)$ for some function $g$ and outputs with high probability a partition of the samples $\hat{S}_1, \ldots, \hat{S}_k$ such that, if we let $S_i$ be the set of samples generated according to the $i$-th component, then with high probability (up to a permutation of $\hat{S}_1, \ldots, \hat{S}_k$)\footnote{It suffices to take $f(x) = \exp(\exp(\tilde{O}(x^2)))$ and $g(x,y)=\exp(\exp(\exp(\exp(\tilde{O}(\min(x,y)^2)))))$.}
\[\min_i \frac{|\hat{S}_i \cap S_i|}{|S_i|} \geq 1-O(kw_{\min}^{-1}\max\{\epsilon^*, \Delta^{-\Omega(1/k)}\})\,.\]
\end{theorem}

\begin{remark}
\label{rem:equiv-obl-adp}
By a recent result of Blanc and Valiant~\cite{blanc2024adaptive} that proved an equivalence between oblivious and adaptive adversaries, we can also generalize~\Cref{thm:zero-mean-main} to the adaptive adversary setting.
In the adaptive adversary setting, a set $X$ of $\poly(d, w_{\min}^{-1}, \Delta)$ i.i.d. samples is drawn from $\mathcal{M}$, and then the input consists of an arbitrary set of samples $Y$ with $|Y|=|X|$ such that $|Y \cap X| \geq (1-\epsilon^*)|X|$.
We need to appeal to their result because, otherwise, in the adaptive adversary setting we could not guarantee that the samples used in the clustering selection step are independent of the samples used to construct the candidate clusterings; this difficulty also appears in~\cite{MR4490075-Bakshi22}.
To apply~\cite{blanc2024adaptive}, we would discretize our data and then apply their Theorem 2, as described in their Remark 2.
Because their result only incurs a loss in sample complexity proportional to the \emph{logarithm} of the size of the domain, we can afford a domain size of $(1/\delta)^d$, even with $\delta$ exponentially small in $d$, while only incurring a loss in sample complexity polynomial in $d$.
\end{remark}

\begin{corollary}[Subspace clustering]
\label{cor:zero-mean-subspace}
In the same setting as~\Cref{thm:zero-mean-main}, if all components $i \neq j$ have $\operatorname{colspan}(\Sigma_i) \neq \operatorname{colspan}(\Sigma_j)$, then the guarantees of~\Cref{thm:zero-mean-main} hold with any $\Delta \geq f(w_{\min}^{-1})$ for some function $f$.\footnote{We stated~\Cref{thm:zero-mean-main} for positive definite covariance matrices, but we can reduce to that case by convolving all samples with i.i.d. Gaussians $N(0, \delta I_d)$ for a $\delta$ that is inverse exponentially small in $n$.} 
\end{corollary}

We prove~\Cref{thm:zero-mean-main} in~\Cref{sec:proof-of-zero-mean}, and we give here a short overview of the proof.

If two components of the mixture have relative Frobenius separation, then an algorithm of~\cite{MR4490075-Bakshi22} partially clusters the mixture with samples and time $f(w_{\min}^{-1}) \cdot \poly(d)$.
This allows us to assume that there is no relative Frobenius separation and therefore that all components of the mixture are spectrally separated.
In this setting, we give an algorithm that finds directions in which the mixture is spectrally separated.
Finally, we partially cluster the mixture along these directions, and we iterate these steps until we obtain $k$ clusters.

The bulk of our work is in~\Cref{sec:zero-mean-sep-vars}, where we recover a low-dimensional subspace containing directions of spectral separation.
Then in~\Cref{sec:zero-mean-main} we use this subspace-finding subroutine to obtain an algorithm that satisfies the guarantees of~\Cref{thm:zero-mean-main}.

\subsection{Finding Directions with Spectral Separation}
\label{sec:zero-mean-sep-vars}
In this section we show how to recover directions with spectral separation when the components do not have Frobenius separation.
When the mixture is in the isotropic position and has lower bounded mixing weights, the directions of spectral separation are roughly equivalent to the directions in which one component has very small variance.
Thus, we phrase our guarantees in terms of recovering directions in which at least one component has small variance.

Our algorithms will only be able to ensure the following relaxation of isotropy that will still be enough for our purposes. 
\begin{definition}[$\gamma$-approximate isotropic position, centered components]
    For $\gamma \leq 1$, a centered mixture is in $\gamma$-approximate isotropic position if $(1-\gamma) I_d \preceq \sum_{i=1}^k w_i \Sigma_i \preceq (1+\gamma) I_d$ and ${\Norm{\sum_{i=1}^k w_i \Sigma_i - I_d}_F \leq \gamma}$. We say that the mixture is (exactly) isotropic when $\gamma=0$.
\end{definition}

\begin{theorem}[Subspace finding theorem, centered components]
\label{thm:zero-mean-subspace}
Consider a $d$-dimensional mixture of $k$ Gaussians $\sum_{i=1}^k w_i N(0, \Sigma_i)$ with $w_{\min} = \min_{i} w_i$ and $\Sigma_1, \ldots, \Sigma_k \succ 0$ in $\gamma$-approximate isotropic position.
Suppose $\norm{\Sigma_i - I_d}_F^2 \leq r$ for all $i \in [k]$, where $r \geq 1$.

Let $\epsilon > 0$ with $\epsilon \leq \Omega(w_{\min}^{16})$, and suppose $\gamma \leq \Omega(w_{\min}^4 \epsilon^2)$.
Then, given $\poly(d, w_{\min}^{-1}, \epsilon^{-1})$ samples from the mixture with an $\epsilon$-fraction of corruptions, there exists an algorithm that runs in time $f(w_{\min}^{-1}, \epsilon^{-1}, r) \cdot \poly(d)$ and outputs an orthogonal projection matrix $Q \in \R^{d \times d}$ with\linebreak$\rank(Q) \leq g(w_{\min}^{-1}, \epsilon^{-1}, r)$ such that with high probability, for every unit vector $v \in \R^d$ such that $v^\top \Sigma_i v \leq \Omega(w_{\min}^{-4} \epsilon^{1/128})$ for some $i \in [k]$, there exists a unit vector $v'$ such that $Qv' = v'$ and $\norm{v-v'} \leq O(w_{\min}^{-1/2}\epsilon^{1/1024})$.
\end{theorem}

\textbf{Proof outline.} The proof of~\Cref{thm:zero-mean-subspace} has three main components:
     \begin{enumerate}
         \item In~\Cref{section:zero-mean-exact-subspace-identifiability} we construct a system of polynomial inequalities $\mathcal{A}$ from the \emph{exact} fourth moments of the mixture that identifies the intersection of the eigenspaces of $\Sigma_i$ that have eigenvalues $\approx 1$. Specifically, every vector that satisfies $\mathcal{A}$ is close for all $i \in [k]$ to one eigenvector of $\Sigma_i$ with eigenvalue $\approx 1$, and in the converse, every vector that is simultaneously for all $i \in [k]$ an eigenvector of $\Sigma_i$ with eigenvalue $\approx 1$ satisfies $\mathcal{A}$.
        \item In~\Cref{section:zero-mean-approx-moment-identifiability} we show we can construct a system of polynomial inequalities $\widehat{\mathcal{A}}$ from the \emph{approximate} fourth moments of the mixture with the same properties as $\cA$.
         \item In~\Cref{section:orthogonal-complement-subspace} we argue that $\widehat{\mathcal{A}}$ satisfies the conditions of the rounding algorithm analyzed in~\Cref{thm:sos-rounding-2}.
     \end{enumerate}

We will address these steps and then return to the proof of~\Cref{thm:zero-mean-subspace}.
We start by defining the polynomials that form the basis of our system of constraints.

\begin{definition}
We define the following homogeneous polynomials in $v \in \R^d$ whose coefficient tensors are functions of $w_1, \ldots, w_k$ and $\Sigma_1, \ldots, \Sigma_k$:
\begin{itemize}
    \item $p(v) = \sum_{i=1}^k w_i (v^\top \Sigma_i v)^2$\,,
    \item $p''(v) = 8 \sum_{i=1}^k w_i (\Sigma_i v)^{\otimes 2} + 4 \sum_{i=1}^k w_i (v^\top \Sigma_i v) \Sigma_i$\,,
    \item $q(v) = \sum_{i=1}^k w_i (v^\top (\Sigma_i - I_d) v)^2$\,.
\end{itemize}
\end{definition}

\begin{lemma}
\label{lemma:polynomials}
Let $\bm{x} \in \mathbb{R}^d$ be distributed according to the ground truth mixture.
Suppose the mixture is in exact isotropic position.
Then
\begin{itemize}
    \item $p(v) = \langle \E \bm{x}^{\otimes 4}, v^{\otimes 4}\rangle / 3 = \E \langle \bm{x}, v\rangle^4 / 3$\,,
    \item $u^\top p''(v)u = 4 \langle \E \bm{x}^{\otimes 4}, v^{\otimes 2} \otimes u^{\otimes 2}\rangle$\,,
    \item $q(v) = \langle \E \bm{x}^{\otimes 4} - I_d^{\otimes 2}, v^{\otimes 4}\rangle / 3 = \E \langle \bm{x}, v\rangle^4 / 3 - \norm{v}^4$\,.
\end{itemize}
\end{lemma}
\begin{proof}
The first claim follows from the fact that the fourth moment of a 1-D Gaussian random variable $N(0, \sigma^2)$ is $3\sigma^4$.
For the second claim, we use that $p''(v)$ is the Hessian of $p(v)$ to obtain
\[p''(v) = (\E \langle \bm{x}, v\rangle^4 / 3)'' = ( 4 \E \langle \bm{x}, v\rangle^3 \bm{x} / 3 )' = 4 \E \langle \bm{x}, v\rangle^2 \bm{x} \bm{x}^\top\,,\]
so 
\[u^\top p''(v) u = 4 \E \langle \bm{x}, v\rangle^2 \langle \bm{x}, u\rangle^2 = 4 \langle \E \bm{x}^{\otimes 4}, v^{\otimes 2} \otimes u^{\otimes 2}\rangle\,.\]
For the third claim, 
\[q(v) = \sum_{i=1}^k w_i (v^\top (\Sigma_i - I_d) v)^2 = \sum_{i=1}^k w_i (v^\top \Sigma_i v)^2 - 2 \|v\|^2 \sum_{i=1}^k w_i (v^\top \Sigma_i v) + \|v\|^4\,,\]
so using that $\sum_{i=1}^k w_i \Sigma_i = I_d$, we have
\[q(v) = \sum_{i=1}^k w_i (v^\top \Sigma_i v)^2 - \|v\|^4 = p(v) - \norm{v}^4\,.\]
\end{proof}

\subsubsection{Identifying Eigenspaces with Eigenvalues $\approx 1$ Using Exact Moments}
\label{section:zero-mean-exact-subspace-identifiability}

\begin{definition}
\label{def:constr-centered}
Let $\bm{x} \in \mathbb{R}^d$ be distributed according to the ground truth mixture.
We define $\mathcal{A}(v, \epsilon)$ to be the following system of polynomial inequalities in indeterminate $v \in \R^d$:
\begin{enumerate}
    \item $\norm{v}^2 \leq 1$\,,
    \item $\Paren{\E \langle \bm{x}, v\rangle^4 / 3 - \norm{v}^4}^2 \leq \epsilon^2 \norm{v}^8$\,,
    \item There exists a degree-$O(1)$ sum-of-squares proof in indeterminate $u \in \R^d$ 
    \[\proves{u}{O(1)} \Set{\left( 4 \langle \E \bm{x}^{\otimes 4}, v^{\otimes 2} \otimes u^{\otimes 2}\rangle - 4 \norm{v}^2\norm{u}^2 - 8\langle v, u\rangle^2\right)^2 \leq \epsilon(\norm{v}^8 + \norm{u}^8)}\,.\]
\end{enumerate}
\end{definition}
See~\Cref{lem:equiv_constr_centered} and the discussion in the paragraph succeeding it for some intuition about the second and third constraints.
We note that the third constraint can be encoded with $\poly(d)$ polynomial inequalities:
see, for example, the discussion about succinct representations in Section 4.3.4 in~\cite{MR4059250-Fleming19}.

\begin{lemma}
\label{lem:equiv_constr_centered}
Suppose the mixture is in exact isotropic position.
Then $\cA(v, \epsilon)$ is equivalent to the following system of polynomial inequalities in indeterminate $v \in\R^d$:
\begin{enumerate}
    \item $\norm{v}^2 \leq 1$\,,
    \item $q(v)^2 \leq \epsilon^2 \norm{v}^8$\,,
    \item There exists a degree-$O(1)$ sum-of-squares proof in indeterminate $u \in \R^d$ 
    \[\proves{u}{O(1)} \Set{\left( u^\top \left(p''(v) - 4 \Vert v \Vert^2 I_d - 8 v^{\otimes 2}\right) u \right)^2 \leq \epsilon \Paren{\norm{v}^8+\norm{u}^8}}\,.\]
\end{enumerate}
\end{lemma}
\begin{proof}
The claims follow from~\Cref{lemma:polynomials}.
\end{proof}

The third constraint is intended to impose an upper bound on the spectral norm of the matrix $p''(v) - 4 \Vert v \Vert^2 I_d - 8 v^{\otimes 2}$. However, the spectral norm of a matrix is not a polynomial in its entries, so we cannot impose such a constraint directly. Instead, we impose the stronger constraint that there exists a sum-of-squares proof in some indeterminate $u \in \R^d$ that the quadratic $u^\top (p''(v) - 4 \Vert v \Vert^2 I_d - 8 v^{\otimes 2}) u$ is bounded.
Note that such a constraint implies a spectral norm bound on the desired matrix.

For ease of notation, we write $\cA$ without $v$ or $\epsilon$ when understood from context.

We show now that $\cA(v)$ implies that $v$ is approximately an eigenvector of eigenvalue $\approx 1$ for all $\Sigma_i$s.

\begin{lemma}[Identifiability of approximate eigenvectors]
\label{lem:eval-ident}
Suppose the mixture is in exact isotropic position.
Then for all $i \in [k]$ 
\begin{align*}
    \mathcal{A}(v, \epsilon) \proves{v}{O(1)} \Set{\|(\Sigma_i - I_d)v\|^2 \leq w_{\min}^{-3} \epsilon^{1/16}}.
\end{align*}
\end{lemma}
\begin{proof}
We start by observing that the constraint $q(v)^2 \leq \epsilon^2 \norm{v}^8$, which is equivalent to
\[\cA \proves{v}{O(1)} \Paren{\sum_{i=1}^k w_i (v^\top (\Sigma_i - I_d) v)^2}^2 \leq \epsilon^2 \norm{v}^8\,, \]
also implies that, for all $i \in [k]$, $\cA \proves{v}{O(1)} w_i^2 (v^\top (\Sigma_i - I_d) v)^4 \leq \epsilon^2 \norm{v}^8$.

Let us turn to the third constraint. By~\Cref{fact:square} and~\Cref{fact:almost_triangle} we have
\begin{align*}
\cA \proves{v,u}{O(1)} & \Paren{ u^\top \left(8 \sum_{i=1}^k w_i (\Sigma_i v)^{\otimes 2} + 4 \sum_{i=1}^k w_i (v^\top \Sigma_i v) \Sigma_i - 4 \Vert v \Vert^2 I_d - 8 v^{\otimes 2}\right) u }^4\\
&\leq \epsilon^2 \Paren{\norm{v}^8+\norm{u}^8}^2 \leq 2\epsilon^2 \Paren{\norm{v}^{16}+\norm{u}^{16}}
\end{align*}
and then again by~\Cref{fact:almost_triangle} we have
\begin{align*}
\cA \proves{v,u}{O(1)} & 16 \epsilon^2 (\norm{v}^{16} + \norm{u}^{16})\\
&\geq \Paren{u^\top \Paren{8 \sum_{i=1}^k w_i (\Sigma_i v)^{\otimes 2} - 8v^{\otimes 2}} u}^4 - 8\Paren{u^\top \Paren{4\sum_{i=1}^k w_i (v^\top \Sigma_i v) \Sigma_i - 4\norm{v}^2 I_d } u}^4\,.
\end{align*}
We now bound the second term on the right-hand side in order to obtain an upper bound on the first term on the right-hand side. We have
\begin{align*}
\cA \proves{v, u}{O(1)} &\Paren{u^\top \Paren{4\sum_{i=1}^k w_i (v^\top \Sigma_i v) \Sigma_i - 4\norm{v}^2 I_d } u}^4
\stackrel{(1)}{=} \Paren{u^\top \Paren{4\sum_{i=1}^k w_i (v^\top (\Sigma_i - I_d) v) \Sigma_i} u}^4\\
&=\Paren{4\sum_{i=1}^k w_i (v^\top (\Sigma_i - I_d) v) (u^\top \Sigma_i u)}^4
\stackrel{(2)}{\leq} 256k^3 \sum_{i=1}^k w_i^4 (v^\top(\Sigma_i - I_d)v)^4 \Paren{u^\top \Sigma_i u}^4\\
&\stackrel{(3)}{\leq} 256k^3\epsilon^2 \norm{v}^8 \sum_{i=1}^k w_i^2 \Paren{u^\top \Sigma_i u}^4
\stackrel{(4)}{\leq} O(k^4 w_{\min}^{-2} \epsilon^2) \norm{v}^8 \norm{u}^8\\
&\leq O(k^4 w_{\min}^{-2} \epsilon^2) \Paren{\norm{v}^{16} + \norm{u}^{16}}\,,
\end{align*}
where in (1) we used that $\sum_{i=1}^k w_i \Sigma_i = I_d$, in (2) we used~\Cref{fact:almost_triangle}, in (3) we used that $\cA \proves{v}{O(1)} w_i^2 (v^\top (\Sigma_i - I_d) v)^4 \leq \epsilon^2 \norm{v}^8$, in (4) we used~\Cref{fact:spectral_norm_bound} and that $\norm{\Sigma_i} \leq O(w_i^{-1})$.
Then we conclude that 
\begin{align*}
\cA \proves{v,u}{O(1)} \Paren{u^\top \Paren{8 \sum_{i=1}^k w_i (\Sigma_i v)^{\otimes 2} - 8v^{\otimes 2}} u}^4
&\leq 16 \epsilon^2 (\norm{v}^{16} + \norm{u}^{16}) + O(k^4 w_{\min}^{-2} \epsilon^2) \Paren{\norm{v}^{16} + \norm{u}^{16}}\\
&\leq O(k^4 w_{\min}^{-2} \epsilon^2) \Paren{\norm{v}^{16} + \norm{u}^{16}}\,.
\end{align*}

From now on, we aim to adjust this inequality to obtain a bound on $\norm{(\Sigma_i - I_d) v}^2$ for all $i \in [k]$.
Replace $u$ by $(\norm{v}^2 I_d - vv^\top)u$.
An easy calculation shows that
\begin{align*}
&u^\top \left(\norm{v}^2 I_d - vv^\top\right)\left( 8 \sum_{i=1}^k w_i (\Sigma_i v)^{\otimes 2} - 8v^{\otimes 2}\right) \left( \norm{v}^2 I_d - vv^\top\right) u\\
&\quad = u^\top \Paren{8 \sum_{i=1}^k w_i \left(\left(\norm{v}^2 I_d - vv^\top\right) \Sigma_i v\right)^{\otimes 2}} u \,.
\end{align*}
Then, by the bound obtained so far with $u$ replaced by $(\norm{v}^2 I_d - vv^\top)u$, and using that $\cA \proves{v}{O(1)} \norm{v}^2 I_d - vv^\top \preceq \norm{v}^2 I_d$, we get
\begin{align*}
\cA \proves{v,u}{O(1)} \Paren{u^\top \Paren{8 \sum_{i=1}^k w_i \left(\left(\norm{v}^2 I_d - vv^\top\right) \Sigma_i v\right)^{\otimes 2}}u}^4 \leq O(k^4 w_{\min}^{-2} \epsilon^2) \Paren{\norm{v}^{16} + \norm{v}^{32}\norm{u}^{16}}\,.
\end{align*}
This also implies that, for all $i \in [k]$,
\begin{align*}
\cA \proves{v,u}{O(1)} \Paren{u^\top \Paren{8 w_i \left(\left(\norm{v}^2 I_d - vv^\top\right) \Sigma_i v\right)^{\otimes 2}}u}^4 \leq O(k^4 w_{\min}^{-2} \epsilon^2) \Paren{\norm{v}^{16} + \norm{v}^{32}\norm{u}^{16}}\,.
\end{align*}
Replace $u$ by $\left(\norm{v}^2 I_d - vv^\top\right) \Sigma_i v$.
Then, using that $\cA \proves{v}{O(1)} \norm{v}^2 \leq 1$ and that ${\norm{\Sigma_i} \leq O(w_{\min}^{-1})}$ implies $\cA \proves{v}{O(1)} \norm{\left(\norm{v}^2 I_d - vv^\top\right) \Sigma_i v}^2 \leq O(w_{\min}^{-2})$, we get 
\begin{align*}
\cA \proves{v}{O(1)} \Norm{\left(\norm{v}^2 I_d - vv^\top\right) \Sigma_i v}^{16} \leq  O(k^4 w_{\min}^{-22} \epsilon^2) \,,
\end{align*}
so by taking roots using~\Cref{fact:square-root} we get
\begin{align*}
\cA \proves{v}{O(1)} \Norm{\left(\norm{v}^2 I_d - vv^\top\right) \Sigma_i v}^{2} \leq  O(k^{1/2} w_{\min}^{-11/4} \epsilon^{1/4})\,.
\end{align*}
Now we bound $\norm{(\Sigma_i-I_d)v}^2$ in terms of this inequality. 
Note that 
\[\left(\norm{v}^2 I_d - vv^\top\right) \Sigma_i v = \norm{v}^2 \Sigma_i v - (v^\top \Sigma_i v) v = \norm{v}^2 (\Sigma_i - I_d) v - (v^\top (\Sigma_i - I_d) v) v\,,\]
so by~\Cref{fact:almost_triangle} and using that $\cA \proves{v}{O(1)} w_i (v^\top(\Sigma_i - I_d)v)^2 \leq \epsilon \norm{v}^4$, 
\begin{align*}
\cA \proves{v}{O(1)} \Norm{\norm{v}^2 (\Sigma_i-I_d)v}^2
&\leq 2 \Norm{\left(\norm{v}^2 I_d - vv^\top\right) \Sigma_i v}^{2} + 2(v^\top (\Sigma_i-I_d) v)^2 \norm{v}^2\\
&\leq O(k^{1/2} w_{\min}^{-11/4} \epsilon^{1/4}) + O(w_{\min}^{-1} \epsilon)\\
&\leq O(k^{1/2} w_{\min}^{-11/4} \epsilon^{1/4})\,,
\end{align*}
so
\[\cA \proves{v}{O(1)} \norm{v}^4 \Norm{(\Sigma_i-I_d)v}^2 \leq O(k^{1/2} w_{\min}^{-11/4} \epsilon^{1/4})\,.\]
Because $\norm{\Sigma_i-I_d} \leq O(w_{\min}^{-1})$ we have $\cA \proves{v}{O(1)} \norm{v}^2 \geq O(w_{\min}^{2}) \Norm{(\Sigma_i-I_d)v}^2$, so by~\Cref{fact:square} $\cA \proves{v}{O(1)} \norm{v}^4 \geq O(w_{\min}^{4}) \Norm{(\Sigma_i-I_d)v}^4$, so combined with the previous bound
\begin{align*}
\cA \proves{v}{O(1)} \Norm{(\Sigma_i - I_d) v}^6 \leq O(k^{1/2} w_{\min}^{-27/4}\epsilon^{1/4})\,.
\end{align*}
Trivially, because $\norm{\Sigma_i-I_d} \leq O(w_{\min}^{-1})$, we have by~\Cref{fact:spectral_norm_bound} that $\cA \proves{v}{O(1)} \Norm{(\Sigma_i-I_d)v}^2 \leq O(w_{\min}^{-2})$, so also
\begin{align*}
\cA \proves{v}{O(1)} \Norm{(\Sigma_i - I_d) v}^8 \leq O(k^{1/2} w_{\min}^{-35/4}\epsilon^{1/4})\,.
\end{align*}
so by~\Cref{fact:square-root}
\begin{align*}
\cA \proves{v}{O(1)} \Norm{(\Sigma_i - I_d) v}^2 \leq O(k^{1/8} w_{\min}^{-35/16} \epsilon^{1/16})\,.
\end{align*}
Finally, this is further upper bounded by $w_{\min}^{-3} \epsilon^{1/16}$.
\end{proof}

Next, we prove an easy corollary of~\Cref{lem:eval-ident}, namely that $\cA(v)$ implies that $v$ is close to the subspace of eigenvectors of eigenvalue $\approx 1$ for all $\Sigma_i$s.

\begin{lemma}[Identifiability of approximate eigenspaces]
\label{lem:evec-ident}
Suppose the mixture is in exact isotropic position.
Let $P_i$ be the orthogonal projection to the subspace of eigenvectors of $\Sigma_i$ such that their eigenvalues lie outside $[1 - \delta, 1+\delta]$. Then for all $i \in [k]$ 
\begin{align*}
    \mathcal{A}(v, \epsilon) \proves{v}{O(1)} \{\|P_i v\|^2 \leq w_{\min}^{-3} \epsilon^{1/16}\delta^{-2}\}.
\end{align*}
\end{lemma}
\begin{proof}
Let $S_- = (-\infty, 1-\delta)$ and $S_+ = (1+\delta, \infty)$. Fix some $\Sigma_i$ and let its eigenvalue decomposition be $\Sigma_i = \sum_{j=1}^d \lambda_j s_j s_j^\top$. By~\Cref{lem:eval-ident}, $\cA \proves{v}{O(1)} \norm{(\Sigma_i-I_d)v}^2 \leq \epsilon'$ where $\epsilon' = w_{\min}^{-3}\epsilon^{1/16}$. Then
\[ \mathcal{A} \proves{v}{O(1)} \Norm{\sum_{j=1}^d (\lambda_j - 1) s_j s_j^\top v}^2 \leq \epsilon' \,, \]
which is equivalent to 
\[ \mathcal{A} \proves{v}{O(1)} \sum_{j=1}^d (\lambda_j - 1)^2 (s_j^\top v)^2 \leq \epsilon'\,.\]
Therefore we also have
\[ \mathcal{A} \proves{v}{O(1)} \delta^2 \sum_{j: \lambda_j \in S_- \cup S_+} (s_j^\top v)^2 \leq \epsilon'\,,\]
so
\[ \mathcal{A} \proves{v}{O(1)} v^\top \Paren{ \sum_{j: \lambda_j \in S_- \cup S_+} s_j s_j^{\top} } \Paren{ \sum_{j: \lambda_j \in S_- \cup S_+} s_j s_j^{\top} } v = \sum_{j: \lambda_j \in S_- \cup S_+} (s_j^\top v)^2 \leq \epsilon' / \delta^2\,,\]
so
\[ \mathcal{A} \proves{v}{O(1)} \norm{P_iv}^2 \leq \epsilon'/\delta^2\,.\]
\end{proof}

We also prove the converse: if $v$ is in the subspace of eigenvectors of eigenvalue $\approx 1$ for all $\Sigma_i$s, then $v$ satisfies $\cA(v)$.

\begin{lemma}[Sum-of-squares feasibility]
\label{lem:zero-mean-sos-feasibility}
Suppose the mixture is in exact isotropic position.
Let $P_i$ be the orthogonal projection to the subspace of eigenvectors of $\Sigma_i$ such that their eigenvalues lie outside $[1 - \delta, 1+\delta]$. Let $P$ be the orthogonal projection to the span of the subspaces associated with $P_i$ for all $i \in [k]$. Then
\[\Set{\norm{v}^2 \leq 1, Pv=0} \proves{v}{O(1)} \cA\Paren{v, O(k^2\delta^2)}\,.\]
\end{lemma}
\begin{proof}
For ease of notation, let $\cC = \Set{\norm{v}^2 \leq 1, Pv=0}$.
We note that, by the definition of $P$, we have ${(I_d-P) \Sigma_i (I_d-P) = I_d + E_i}$ where $\norm{E_i} \leq \delta$.
We will prove that $v$ satisfies each of the three constraints of $\cA$.
First, we have trivially that $\cC \proves{v}{2} \Set{\norm{v}^2 \leq 1}$.

Second, using that $\cC \proves{v}{O(1)} \Set{(I_d-P)v=v}$ and~\Cref{fact:spectral_norm_bound} we have
\[\cC \proves{v}{O(1)} q(v)^2 = \Paren{\sum_{i=1}^k w_i \left(v^\top \left(\Sigma_i - I_d\right) v\right)^2}^2 = \Paren{\sum_{i=1}^k w_i \Paren{v^\top E_i v}^2}^2 \leq \delta^4 \norm{v}^8\,.\]

Third, let us inspect
\begin{align*}
p''(v) - 4\norm{v}^2I_d - 8v^{\otimes 2}
&= 8\sum_{i=1}^k w_i (\Sigma_i v)^{\otimes 2} + 4\sum_{i=1}^k w_i (v^\top \Sigma_i v) \Sigma_i - 4\norm{v}^2I_d - 8v^{\otimes 2}\,.
\end{align*}

For the first and fourth terms, we have
\begin{align*}
\cC \proves{v,u}{O(1)}
&\Paren{u^\top \Paren{8\sum_{i=1}^k w_i (\Sigma_i v)^{\otimes 2} - 8 v^{\otimes 2}}u}^2\\
&= \Paren{u^\top \Paren{8\sum_{i=1}^k w_i (I_d + (\Sigma_i - I_d)) vv^\top (I_d + (\Sigma_i - I_d)) - 8 v^{\otimes 2}}u}^2\\
&= \Paren{u^\top \Paren{8\sum_{i=1}^k w_i \Paren{(\Sigma_i - I_d) vv^\top + vv^\top (\Sigma_i - I_d) + (\Sigma_i - I_d) vv^\top (\Sigma_i - I_d)}}u}^2\\
&= \Paren{8\sum_{i=1}^k w_i \Paren{ u^\top (\Sigma_i - I_d) v (v^\top u) + (u^\top v)v^\top (\Sigma_i - I_d) u + u^\top (\Sigma_i - I_d) vv^\top (\Sigma_i - I_d)u}}^2\\
&\stackrel{(1)}{\leq} 192 k \sum_{i=1}^k w_i^2 \Big((u^\top (\Sigma_i - I_d) v)^2(v^\top u)^2+(u^\top v)^2(v^\top (\Sigma_i - I_d) u)^2\\
&\quad\quad +(u^\top (\Sigma_i - I_d) v)^2(v^\top (\Sigma_i - I_d) u)^2\Big)\\
&\stackrel{(2)}{\leq} O\Paren{k\delta^2} \norm{v}^4 \norm{u}^4 \leq O\Paren{k\delta^2 \Paren{\norm{v}^8 + \norm{u}^8}}\,,
\end{align*}
where in (1) we used~\Cref{fact:almost_triangle} and in (2) we used that by~\Cref{fact:cauchy-schwarz} and~\Cref{fact:spectral_norm_bound} we have $\cC \proves{v,u}{O(1)} (u^\top (\Sigma_i - I_d) v)^2 \leq \norm{u}^2 \norm{(\Sigma_i - I_d)v}^2 = \norm{u}^2 \norm{E_i v}^2 \leq \delta^2\norm{v}^2\norm{u}^2$. For the second and third terms, using that $\sum_{i=1}^k w_i \Sigma_i = I_d$, we have 
\begin{align*}
\cC \proves{v,u}{O(1)}
&\Paren{u^\top \Paren{4\sum_{i=1}^k w_i (v^\top \Sigma_i v) \Sigma_i - 4\norm{v}^2I_d}u}^2\\
&= \Paren{u^\top \Paren{4\sum_{i=1}^k w_i (v^\top (I_d+E_i) v) \Sigma_i - 4\norm{v}^2 \sum_{i=1}^k w_i \Sigma_i}u}^2\\
&= \Paren{u^\top \Paren{4\sum_{i=1}^k w_i (v^\top E_i v) \Sigma_i}u}^2\\
&= \Paren{4\sum_{i=1}^k w_i (v^\top E_i v) (u^\top \Sigma_i u)}^2\\
&\stackrel{(1)}{\leq} 16k \sum_{i=1}^k w_i^2 (v^\top E_i v)^2 (u^\top \Sigma_i u)^2\\
&\stackrel{(2)}{\leq} O(k^2 \delta^2) \norm{v}^4\norm{u}^4 \leq O\Paren{k^2 \delta^2 \Paren{\norm{v}^8 + \norm{u}^8}}\,,
\end{align*}
where in (1) we used~\Cref{fact:almost_triangle} and in (2) we used~\Cref{fact:spectral_norm_bound} and that $\norm{E_i} \leq \delta$, $\norm{\Sigma_i} \leq w_{i}^{-1}$. Combining the two upper bounds, we get by~\Cref{fact:almost_triangle} that
\begin{align*}
\cC \proves{v,u}{O(1)} \Paren{u^\top \Paren{p''(v) - 2\norm{v}^2I_d - 4v^{\otimes 2}} u}^2
&\leq O\Paren{k^2 \delta^2 \Paren{\norm{v}^8 + \norm{u}^8}}\,.
\end{align*}
\end{proof}

\subsubsection{Identifying Eigenspaces with Eigenvalues $\approx 1$ Using Approximate Moments}
\label{section:zero-mean-approx-moment-identifiability}
In this section we show how to construct a system of polynomial constraints $\widehat{\cA}$ with properties similar to $\cA$ when only having access to approximate moments. Before doing so, we state two results on the closeness of the empirical moments to the population moments.

\begin{fact}[Theorem 1.3 and Lemma 5.4 in~\cite{kothari2017outlier}]
\label{thm:sos-mom-est}
Given $\poly(d, w_{\min}^{-1}, \epsilon^{-1})$ samples from the mixture with an $\epsilon$-fraction of corruptions, where $\epsilon \leq \Omega(w_{\min}^2)$, there exists an algorithm that runs in time $\poly(d, w_{\min}^{-1}, \epsilon^{-1})$ and outputs symmetric tensor moment estimates $\widehat{M}_2 \in \R^{d^2}$ and $\widehat{M}_4 \in \R^{d^4}$ of the tensor moments of the mixture such that with high probability, for all vectors $v \in \R^d$,
\[\langle M_2 - \widehat{M}_2, v^{\otimes 2}\rangle^2 \leq O(w_{\min}^{-1}\sqrt{\epsilon}) \cdot \langle M_2, v^{\otimes 2}\rangle^2\]
and 
\[\langle M_4 - \widehat{M}_4, v^{\otimes 4}\rangle^2 \leq O(w_{\min}^{-2}\sqrt{\epsilon}) \cdot \langle M_2, v^{\otimes 2}\rangle^4\,.\]
Furthermore, there exist degree-$O(1)$ sum-of-squares proofs in $v$ of these inequalities.
\end{fact}

We obtain the following simple corollary for mixtures in $\gamma$-approximate isotropic position:

\begin{corollary}
\label{cor:sos-mom-est-iso}
Suppose the mixture is in $\gamma$-approximate isotropic position. Then the same result as in~\Cref{thm:sos-mom-est} holds with upper bounds
\[\langle M_2 - \widehat{M}_2, v^{\otimes 2}\rangle^2 \leq O(w_{\min}^{-1}\sqrt{\epsilon}) \norm{v}^4\]
and 
\[\langle M_4 - \widehat{M}_4, v^{\otimes 4}\rangle^2 \leq O(w_{\min}^{-2}\sqrt{\epsilon}) \norm{v}^8\,.\]
\end{corollary}
\begin{proof}
Because of the $\gamma$-approximate isotropic position, we have that $\norm{M_2} \leq 1+\gamma$.
Then, for any even $t \geq 2$, using~\Cref{fact:spectral_norm_bound},
\[\proves{v}{O(t)} \langle M_2, v^{\otimes 2}\rangle^t = (v^\top M_2 v)^t \leq \norm{M_2}^t \norm{v}^{2t} \leq (1+\gamma)^t \norm{v}^{2t}\,.\]
Setting $t=2$ and $t=4$ leads to the desired conclusions.
\end{proof}

\begin{lemma}[Empirical fourth moment Hessian bound]
\label{lem:sos-hes-est}
Suppose we have some $\widehat{M}_4 \in \R^{d^4}$ such that
\[\proves{v}{O(1)} \Set{\langle M_4 - \widehat{M}_4, v^{\otimes 4}\rangle^2 \leq \epsilon \norm{v}^8}\,.\]
Then 
\[\proves{v, u}{O(1)} \langle M_4 - \widehat{M}_4, v^{\otimes 2} \otimes u^{\otimes 2} \rangle^2 \leq O(\epsilon(\norm{v}^8 + \norm{u}^8))\,.\]
\end{lemma}
\begin{proof}
We observe 
\begin{align*}
&\langle M_4 - \widehat{M}_4, (v + u)^{\otimes 4}\rangle + \langle M_4 - \widehat{M}_4, (v - u)^{\otimes 4} \rangle\\
&\quad = 2 \langle M_4 - \widehat{M}_4, v^{\otimes 4} \rangle + 2 \langle M_4 - \widehat{M}_4, u^{\otimes 4} \rangle + 12 \langle M_4 - \widehat{M}_4, v^{\otimes 2} \otimes u^{\otimes 2} \rangle\,.
\end{align*}
Then, using~\Cref{fact:almost_triangle},
\begin{align*}
\proves{u,v}{O(1)} &\langle M_4 - \widehat{M}_4, v^{\otimes 2} \otimes u^{\otimes 2} \rangle^2\\
&\quad \leq O\Bigg(\langle M_4 - \widehat{M}_4, (v + u)^{\otimes 4}\rangle^2 + \langle M_4 - \widehat{M}_4, (v - u)^{\otimes 4} \rangle^2 \\
&\quad\quad + \langle M_4 - \widehat{M}_4, v^{\otimes 4} \rangle^2 + \langle M_4 - \widehat{M}_4, u^{\otimes 4} \rangle^2\Bigg)\,.
\end{align*}
We can now apply the assumption to obtain
\begin{align*}
\proves{u,v}{O(1)} &\langle M_4 - \widehat{M}_4, v^{\otimes 2} \otimes u^{\otimes 2} \rangle^2 \leq O\Paren{\epsilon \norm{v + u}^8 + \epsilon \norm{v - u}^8 + \epsilon \norm{v}^8 + \epsilon \norm{u}^8}\\
&\leq O(\epsilon)(\norm{v}^8+\norm{u}^8)\,.
\end{align*}
\end{proof}

In~\Cref{lem:zero-mean-feasibility} we prove that we can construct a system of constraints $\widehat\cA(v)$ that is roughly equivalent to $\cA(v)$.

\begin{lemma}[Approximate moment feasibility]
\label{lem:zero-mean-feasibility}
Suppose the mixture is in $\gamma$-approximate isotropic position.
Given $\poly(d, w_{\min}^{-1}, \epsilon^{-1})$ samples from the mixture with an $\epsilon$-fraction of corruptions, where $\epsilon \leq \Omega(w_{\min}^4)$, there exists an algorithm that runs in time $\poly(d, w_{\min}^{-1}, \epsilon^{-1})$ and computes a system of polynomial inequalities $\widehat\cA(v, \epsilon)$ of size $\poly(d)$ in indeterminate $v \in \R^d$ such that with high probability
\[\widehat\cA(v, \epsilon) \proves{v}{O(1)} \cA(v, O(w_{\min}^{-1}\epsilon^{1/4}))\]
and 
\[\cA(v, \epsilon) \proves{v}{O(1)} \widehat\cA(v, O(w_{\min}^{-1}\epsilon^{1/4}))\,.\]
\end{lemma}
\begin{proof}
We compute by~\Cref{cor:sos-mom-est-iso} in time $\poly(d, w_{\min}^{-1}, \epsilon^{-1})$ some $\widehat{M}_4$ such that
\[\proves{v}{O(1)} \langle M_4 - \widehat{M}_4, v^{\otimes 4}\rangle^2 \leq O(w_{\min}^{-2} \sqrt{\epsilon}) \norm{v}^8\,.\]
Let $\epsilon' = O(w_{\min}^{-1}\epsilon^{1/4})$. Then we construct the following system of polynomial inequalities $\widehat\cA(v, \epsilon')$ in indeterminate $v \in \R^d$: 
\begin{enumerate}
    \item $\norm{v}^2 \leq 1$\,,
    \item $\Paren{\langle \widehat{M}_4 / 3, v^{\otimes 4}\rangle - \norm{v}^4}^2 \leq (\epsilon')^2 \norm{v}^8$\,,
    \item there exists a degree-$O(1)$ sum-of-squares proof in indeterminate $u \in \R^d$
    \[\proves{u}{O(1)} \Set{\left( 4 \langle \widehat{M}_4, v^{\otimes 2} \otimes u^{\otimes 2}\rangle - 4 \norm{v}^2\norm{u}^2 - 8\langle v, u\rangle^2\right)^2 \leq \epsilon' \Paren{\norm{v}^8+\norm{u}^8}}\,.\]
\end{enumerate}
We start by proving that $\widehat\cA(v, \epsilon) \proves{v}{O(1)} \cA(v, O(\epsilon'))$. We have trivially that $\widehat\cA(v, \epsilon) \proves{v}{O(1)} \norm{v}^2 \leq 1$. For the second constraint, we have by~\Cref{fact:almost_triangle}
\begin{align*}
\widehat\cA(v, \epsilon) &\proves{v}{O(1)} \Paren{\langle M_4 / 3, v^{\otimes 4}\rangle - \norm{v}^4}^2\\
&\leq 2\Paren{\langle \widehat{M}_4 / 3, v^{\otimes 4}\rangle - \norm{v}^4}^2 + 2\langle M_4 / 3 - \widehat{M}_4 / 3, v^{\otimes 4}\rangle^2\\
&\leq O((\epsilon')^2) \norm{v}^8\,.
\end{align*}
For the third constraint, we have by~\Cref{fact:almost_triangle} and~\Cref{lem:sos-hes-est}
\begin{align*}
\widehat\cA(v, \epsilon) \proves{v, u}{O(1)}
&\Paren{ 4 \langle M_4, v^{\otimes 2} \otimes u^{\otimes 2}\rangle - 4 \norm{v}^2\norm{u}^2 - 8\langle v, u\rangle^2}^2\\
&\leq 2 \Paren{ 4 \langle \widehat{M}_4, v^{\otimes 2} \otimes u^{\otimes 2}\rangle - 4 \norm{v}^2\norm{u}^2 - 8\langle v, u\rangle^2}^2 + 2\Paren{ 4 \langle M_4 - \widehat{M}_4, v^{\otimes 2} \otimes u^{\otimes 2} \rangle}^2\\
&\leq O\Paren{\epsilon'\Paren{\norm{v}^8+\norm{u}^8}}\,.
\end{align*}
Therefore $\widehat\cA(v, \epsilon) \proves{v}{O(1)} \cA(v, O(\epsilon'))$.

A similar set of inequalities gives the other direction.
\end{proof}

Finally, we show that even if the mixture is only in $\gamma$-approximate isotropic position, $\cA(v)$ is still roughly equivalent to the system of constraints that we would have if the mixture were in exact isotropic position.

\begin{lemma}[Approximate isotropic position]
\label{lem:zero-mean-approx-isotropic}
Suppose the mixture is in $\gamma$-approximate isotropic position, with $\gamma \leq \Omega(w_{\min}^2)$.
Let $\tilde\Sigma = \sum_{i=1}^k w_i \Sigma_i$, and let $\tilde{\cA}(v, \epsilon)$ be the system of polynomial inequalities in indeterminate $v\in \R^d$ in~\Cref{def:constr-centered} \emph{for the mixture} $\tilde\Sigma^{-1/2} \bm{x}$ where $\bm{x}$ is distributed according to the ground truth mixture, that is, for the ground truth mixture put into exact isotropic position.
Then
\[\cA(v, \epsilon) \proves{v}{O(1)} \tilde{\cA}(v, O(\epsilon + w_{\min}^{-2} \gamma))\]
and 
\[\tilde\cA(v, \epsilon) \proves{v}{O(1)} \cA(v, O(\epsilon + w_{\min}^{-2} \gamma))\,.\]
\end{lemma}
\begin{proof}
We start by proving that $\cA(v, \epsilon) \proves{v}{O(1)} \tilde{\cA}(v, O(\epsilon + w_{\min}^{-2} \gamma))$.
We have trivially that $\cA(v, \epsilon) \proves{v}{O(1)} \norm{v}^2 \leq 1$.

For the second constraint, 
we have by~\Cref{fact:almost_triangle} and~\Cref{fact:cauchy-schwarz}
\begin{align*}
\cA(v, \epsilon) &\proves{v}{O(1)} \Paren{\E \langle \tilde\Sigma^{-1/2} \bm{x}, v\rangle^4 / 3 - \norm{v}^4}^2\\
&\leq 2 \Paren{\E \langle \bm{x}, v\rangle^4 / 3 - \norm{v}^4}^2 + 2 \Paren{\E \langle \tilde\Sigma^{-1/2} \bm{x}, v\rangle^4 / 3 - \E \langle \bm{x}, v\rangle^4 / 3}^2\\
&\leq 2\epsilon^2 \norm{v}^8 + 2 \Paren{\E \langle \tilde\Sigma^{-1/2} \bm{x}, v\rangle^4 / 3 - \E \langle \bm{x}, v\rangle^4 / 3}^2\\
&= 2\epsilon^2 \norm{v}^8 +  2 \Paren{\E (\langle \bm{x}, v\rangle + \langle (\tilde\Sigma^{-1/2} - I_d) \bm x, v\rangle)^4 / 3 - \E \langle \bm{x}, v\rangle^4 / 3}^2\\
&\leq 2\epsilon^2 \norm{v}^8 + O\Paren{ \sum_{\ell=0}^3 \Paren{\E \langle \bm{x}, v \rangle^\ell \langle (\tilde\Sigma^{-1/2} - I_d) \bm x, v\rangle^{4-\ell}}^2  }\\
&\leq 2\epsilon^2 \norm{v}^8 + O\Paren{ \sum_{\ell=0}^3 \E \langle \bm{x}, v \rangle^{2\ell} \E \langle (\tilde\Sigma^{-1/2} - I_d) \bm x, v\rangle^{8-2\ell}}\,,
\end{align*}
where the last line uses~\Cref{fact:expect-cauchy-schwarz}.
Because of the $\gamma$-approximate isotropic position we have $\norm{\Sigma_i} \leq O(w_{\min}^{-1})$ for all $i \in [k]$, so we can bound for any constant $t$ using~\Cref{fact:spectral_norm_bound}
\begin{align*}
\proves{v}{O(1)} \E \langle \bm x, v\rangle^{2t}
\leq O\Paren{ \sum_{i=1}^k w_i (v^\top \Sigma_i v)^t } \leq O(w_{\min}^{-t}) \norm{v}^{2t}\,.
\end{align*}
Similarly, using that $\norm{(\tilde\Sigma^{-1/2}-I_d)\Sigma_i(\tilde\Sigma^{-1/2}-I_d)} \leq O(w_{\min}^{-1}\gamma^2)$,
\begin{align*}
\proves{v}{O(1)} \E\langle (\tilde\Sigma^{-1/2}-I_d)\bm x, v\rangle^{2t}
\leq O\Paren{ \sum_{i=1}^k w_i (v^\top (\tilde\Sigma^{-1/2}-I_d)^\top \Sigma_i(\tilde\Sigma^{-1/2}-I_d) v)^t } \leq O(w_{\min}^{-t} \gamma^{2t}) \norm{v}^{2t}\,.
\end{align*}
Therefore we get
\begin{align*}
\cA(v, \epsilon) &\proves{v}{O(1)} \Paren{\E \langle \tilde\Sigma^{-1/2} \bm{x}, v\rangle^4 / 3 - \norm{v}^4}^2 \leq 2\epsilon^2 \norm{v}^8 + O(w_{\min}^{-4} \gamma^{2}) \norm{v}^8\,.
\end{align*}
For the third constraint, we have by~\Cref{fact:almost_triangle}
\begin{align*}
\cA(v, \epsilon)
&\proves{v,u}{O(1)} \Paren{ 4 \langle \E (\tilde\Sigma^{-1/2} \bm{x})^{\otimes 4}, v^{\otimes 2} \otimes u^{\otimes 2}\rangle - 4 \norm{v}^2\norm{u}^2 - 8\langle v, u\rangle}^2\\
&\leq 2\Paren{ 4 \langle \E \bm{x}^{\otimes 4}, v^{\otimes 2} \otimes u^{\otimes 2}\rangle - 4 \norm{v}^2\norm{u}^2 - 8\langle v, u\rangle}^2\\
&\quad+ 2 \Paren{ 4 \langle \E (\tilde\Sigma^{-1/2} \bm{x})^{\otimes 4}, v^{\otimes 2} \otimes u^{\otimes 2}\rangle - 4 \langle \E \bm{x}^{\otimes 4}, v^{\otimes 2} \otimes u^{\otimes 2}\rangle}^2\\
&\leq 2\epsilon (\norm{v}^8+\norm{u}^8) + 2 \Paren{ 4 \langle \E (\tilde\Sigma^{-1/2} \bm{x})^{\otimes 4}, v^{\otimes 2} \otimes u^{\otimes 2}\rangle - 4 \langle \E \bm{x}^{\otimes 4}, v^{\otimes 2} \otimes u^{\otimes 2}\rangle}^2\,.
\end{align*}
We will make use of~\Cref{lem:sos-hes-est}, so it will suffice to bound the second term above for the case $v=u$.
We have 
\begin{align*}
\cA(v, \epsilon)
&\proves{v}{O(1)} 2 \Paren{ 4 \langle \E (\tilde\Sigma^{-1/2} \bm{x})^{\otimes 4}, v^{\otimes 4}\rangle - 4 \langle \E \bm{x}^{\otimes 4}, v^{\otimes 4}\rangle}^2\\
&\leq O\Paren{\E \langle \tilde\Sigma^{-1/2} \bm{x}, v\rangle^4/3 - \E \langle \bm{x}, v\rangle^4\rangle/3}^2\\
&\leq O\Paren{w_{\min}^{-4}\gamma^2} \norm{v}^8
\end{align*}
where the last inequality follows by the bound we proved for the second constraint. Then by~\Cref{lem:sos-hes-est} an analogous bound also applies to the term with $v^{\otimes 2} \otimes u^{\otimes 2}$. Therefore we get
\begin{align*}
\cA(v, \epsilon)
&\proves{v}{O(1)} \Paren{ 4 \langle \E (\tilde\Sigma^{-1/2} \bm{x})^{\otimes 4}, v^{\otimes 2} \otimes u^{\otimes 2}\rangle - 2 \norm{v}^2\norm{u}^2 - 4\langle v, u\rangle}^2\\
&\leq 2\epsilon (\norm{v}^8+\norm{u}^8) + O(w_{\min}^{-4} \gamma^2) (\norm{v}^8+\norm{u}^8)\,.
\end{align*}
Hence, from the bounds we got for the two constraints, we get that $\cA(v, \epsilon) \proves{v}{O(1)} \tilde{\cA}(v, O(\epsilon + w_{\min}^{-2} \gamma))$.

A similar set of inequalities gives the other direction.
\end{proof}

\subsubsection{Final Constraint Construction}
\label{section:orthogonal-complement-subspace}
In this section, we show that we can construct $\widehat{\cA}$ that satisfies the conditions of the rounding algorithm analyzed in~\Cref{thm:sos-rounding-2}.

\begin{lemma}[Centered mixtures rounding conditions]
\label{lem:zero-mean-complement}
Suppose the mixture is in $\gamma$-approximate isotropic position.
Let $\cA(v, \epsilon)$ be the system of polynomial inequalities for the mixture.
Let $\tilde\Sigma = \sum_{i=1}^k w_i \Sigma_i$, and let $\tilde{\cA}(v, \epsilon)$ be the system of polynomial inequalities in indeterminate $v\in \R^d$ in~\Cref{def:constr-centered} \emph{for the mixture} $\tilde\Sigma^{-1/2} \bm{x}$ where $\bm{x}$ is distributed according to the ground truth mixture, that is, for the ground truth mixture put into exact isotropic position.

Let $P_i$ be the orthogonal projection to the subspace of eigenvectors of $\tilde\Sigma^{-1/2} \Sigma_i \tilde\Sigma^{-1/2}$ such that their eigenvalues lie outside $[1 - \delta, 1+\delta]$. 
Also let $\bar{P}_i$ be the orthogonal projection to the subspace of eigenvectors of $\tilde\Sigma^{-1/2} \Sigma_i \tilde\Sigma^{-1/2}$ such that their eigenvalues lie outside $[1 - \delta^{1/2048}, 1+\delta^{1/2048}]$.
Finally, let $R \in \R^{d\times d}$ be the orthogonal projection to the span of the subspaces associated with $P_i$ for all $i \in [k]$.

Suppose $\epsilon \leq \Omega(w_{\min}^{16})$, $\gamma \leq \Omega(w_{\min}^6 \epsilon)$, and $\delta \sim w_{\min}^3 \epsilon^2$ small enough.
Given $\poly(d, w_{\min}^{-1}, \epsilon^{-1})$ samples from the mixture with an $\epsilon$-fraction of corruptions, there exists an algorithm that runs in time $\poly(d, w_{\min}^{-1}, \epsilon^{-1})$ and computes a system of polynomial inequalities $\widehat\cA$ of size $\poly(d)$ in indeterminate $v \in \R^d$, including $\norm{v}^2 \leq 1$, such that with high probability:
    \begin{enumerate}
        \item $\widehat\cA \proves{v}{O(1)} \Set{\Norm{\bar{P}_i v}^2 \leq O(w_{\min}^{-4} \epsilon^{1/128})}$ for all $i \in [k]$,
        \item $\Set{\norm{v}^2 \leq 1, Rv=0} \proves{v}{O(1)} \widehat\cA$.
    \end{enumerate}
\end{lemma}
\begin{proof}
    By~\Cref{lem:zero-mean-feasibility} and~\Cref{lem:zero-mean-approx-isotropic} we can construct a system of polynomial inequalities $\widehat\cA(v, \epsilon)$ in time $\poly(d, w_{\min}^{-1}, \epsilon^{-1})$ such that
    \begin{align*}
    \widehat{\cA}(v, \epsilon)
    &\proves{v}{O(1)} \cA(v, O(w_{\min}^{-1}\epsilon^{1/4}))\\
    &\proves{v}{O(1)} \tilde\cA(v, O(w_{\min}^{-1}\epsilon^{1/4} + w_{\min}^{-2} \gamma)\,.
    \end{align*}
    
    Let $\epsilon' = O(w_{\min}^{-1}\epsilon^{1/4} + w_{\min}^{-2} \gamma)$.
    Then by~\Cref{lem:evec-ident}, for all $i \in [k]$,
    \[\widehat{\cA}(v, \epsilon) \proves{v}{O(1)} \{\|\bar{P}_iv\|^2 \leq O(w_{\min}^{-3}(\epsilon')^{1/16}\delta^{-1/1024})\}\,,\]
    where using that $\delta^{1/1024} \geq \Omega((\epsilon')^{1/32})$ we have the bound 
    \[\widehat{\cA}(v, \epsilon) \proves{v}{O(1)} \{\|\bar{P}_iv\|^2 \leq O(w_{\min}^{-3} (\epsilon')^{1/32})\}\,,\]
    where using the bounds on $\epsilon$ and $\gamma$ we get the bound
    \[\widehat{\cA}(v, \epsilon) \proves{v}{O(1)} \{\|\bar{P}_iv\|^2 \leq O(w_{\min}^{-4} \epsilon^{1/128})\}\,.\]

    For the second claim, we start by noting that, by~\Cref{lem:zero-mean-sos-feasibility}, there exists a sum-of-squares proof in indeterminate $v \in \R^d$ that, if $\norm{v}^2 \leq 1$ and $Rv=0$, then $v$ is feasible for $\tilde\cA(v, O(k^2\delta^2))$.
    Then by~\Cref{lem:zero-mean-feasibility} and~\Cref{lem:zero-mean-approx-isotropic}
    \begin{align*}
    &\tilde\cA(v, \Theta(k^2\delta^2))\\
    &\quad \proves{v}{O(1)} \cA(v, O(k^2\delta^2 + w_{\min}^{-2} \gamma))\\
    &\quad \proves{v}{O(1)} \hat\cA\Paren{v, O(w_{\min}^{-1} (k^2\delta^2 + w_{\min}^{-2} \gamma)^{1/4})}\,.
    \end{align*}
    When $\gamma \leq \Omega(w_{\min}^6 \epsilon)$ and $\delta \leq \Omega(w_{\min}^3 \epsilon^2)$, this implies $\widehat{\cA}(v,\epsilon)$.
\end{proof}

We also prove that the subspace we aim to find has dimension bounded by ${\max_i \norm{\Sigma_i - I_d}_F^2}$.

\begin{lemma}[Rank bound]
\label{lem:rank}
Let $A_1, \ldots, A_k \in \R^{d \times d}$.
For all $i \in [k]$, let $P_i \in \R^{d \times d}$ be the orthogonal projection to the subspace of eigenvectors of $A_i$ such that their eigenvalues lie outside $[1 - \delta, 1+\delta]$.
Let $P \in \R^{d\times d}$ be the orthogonal projection to the span of the subspaces associated with $P_i$ for all $i \in [k]$.
Suppose that $\|A_i-I_d\|_F^2 \leq r$ for all $i \in [k]$. Then  $\rank(P) \leq k r/\delta^2$.
\end{lemma}
\begin{proof}
First, we note that $\rank(P) \leq \sum_{i=1}^k \rank(P_i)$ and thus it suffices to bound each $\rank(P_i)$. For some fixed $i \in [k]$, let $\lambda_1, \ldots, \lambda_d$ be the eigenvalues of $A_i - I_d$. We have that 
\begin{align*}
    r \geq \|A_i-I_d\|_F^2 = \sum_{j=1}^d \lambda_j^2 \geq \delta^2 \rank(P_i)\,.
\end{align*}
Then $\rank(P_i) \leq r/\delta^2$ for all $i \in [k]$ and therefore $\rank(P) \leq k r/\delta^2$.
\end{proof}

\subsubsection{Proof of~\Cref{thm:zero-mean-subspace}}

Let $\tilde\Sigma = \sum_{i=1}^k w_i \Sigma_i$.
Let $P_i$ be the orthogonal projection to the subspace of eigenvectors of $\tilde\Sigma^{-1/2}\Sigma_i\tilde\Sigma^{-1/2}$ such that their eigenvalues lie outside $[1 - \delta, 1+\delta]$, where $\delta \sim w_{\min}^3 \epsilon^2$.
Also let $\bar{P}_i$ be the orthogonal projection to the subspace of eigenvectors of $\tilde\Sigma^{-1/2}\Sigma_i\tilde\Sigma^{-1/2}$ such that their eigenvalues lie outside $[1 - \delta^{1/2048}, 1+\delta^{1/2048}]$.
Finally, let $R \in \R^{d\times d}$ be the orthogonal projection to the span of the subspaces associated with $P_i$ for all $i \in [k]$.

By~\Cref{lem:zero-mean-complement}, we can compute in time $\poly(d, w_{\min}^{-1}, \epsilon^{-1})$ a system of polynomial inequalities $\widehat\cA$ of size $\poly(d)$ in indeterminate $v \in \R^d$, including $\norm{v}^2 \leq 1$, such that 
\begin{itemize}
    \item $\widehat\cA \proves{v}{O(1)} \Set{\Norm{\bar{P}_i v}^2 \leq O(w_{\min}^{-4} \epsilon^{1/128})}$ for all $i \in [k]$,
    \item $\Set{\norm{v}^2 \leq 1, Rv=0} \proves{v}{O(1)} \widehat\cA$.
\end{itemize}
Furthermore, we have for all $i \in [k]$ that
\begin{align*}
\Norm{\tilde\Sigma^{-1/2}\Sigma_i\tilde\Sigma^{-1/2} - I_d}_F
&= \Norm{\tilde\Sigma^{-1/2}\Sigma_i\tilde\Sigma^{-1/2} - \tilde\Sigma^{-1/2} \tilde\Sigma \tilde\Sigma^{-1/2}}_F
\leq \Norm{\tilde\Sigma^{-1/2}}^2 \Norm{\Sigma_i - \tilde\Sigma}_F\\
&\leq O\Paren{\Norm{\Sigma_i - \tilde\Sigma}_F}
\leq O\Paren{\Norm{\Sigma_i - I_d}_F + \Norm{\tilde\Sigma - I_d}_F}
\leq O(\sqrt{r} + \gamma)\,.
\end{align*}
Then, by~\Cref{lem:rank}, $\operatorname{rank}(R) \leq O(k (r+\gamma^2)/\delta^2) \leq O(w_{\min}^{-8} \epsilon^{-4} r)$.

Then, by applying~\Cref{thm:sos-rounding-2} to $\widehat\cA$ with orthogonal projection matrices $\bar{P}_1, \ldots, \bar{P}_k$, and $R$, we can compute in time $f(w_{\min}^{-1}, \epsilon^{-1}, r) \cdot \poly(d)$ a $D$-dimensional subspace such that, for every unit vector $v \in \R^d$ such that $\norm{\bar{P}_i v}^2 \geq 1-\alpha$ for some $i \in [k]$, there exists a unit vector in the subspace that is $\beta$-close to it, where 
\[D = O(w_{\min}^{1/2} \epsilon^{-1/1024})^{O(w_{\min}^{-8} \epsilon^{-4} r)}\,,\]
\[\alpha = \Omega(w_{\min}^{-4} \epsilon^{1/128})\,,\]
\[\beta = O(w_{\min}^{-1/2}\epsilon^{1/1024})\,.\]

We argue now that the unit vectors $v \in \R^d$ with $\norm{\bar{P}_i v}^2 \geq 1-\alpha$ are the ones we are interested in.
First, we verify that every unit vector $v \in \R^d$ with $v^\top \tilde\Sigma^{-1/2} \Sigma_i \tilde\Sigma^{-1/2} v \leq \Omega(\alpha)$ for some $i \in [k]$ also satisfies $\norm{\bar{P}_i v}^2 \geq 1- \alpha$.
By the definition of $\bar{P}_i$, we have that $\norm{(I_d-\bar{P}_i)v}^2 \leq O(v^\top \tilde\Sigma^{-1/2} \Sigma_i \tilde\Sigma^{-1/2} v)$, so if $v^\top \tilde\Sigma^{-1/2} \Sigma_i \tilde\Sigma^{-1/2} v \leq \Omega(\alpha)$ then $\norm{(I_d-\bar{P}_i) v}^2 \leq \Omega(\alpha)$ so $\norm{\bar{P}_i v}^2 \geq 1-\Omega(\alpha) \geq 1-\alpha$.
Second, we want a result about $v^\top \Sigma_i v$, not $v^\top \tilde\Sigma^{-1/2} \Sigma_i \tilde\Sigma^{-1/2} v$, but we note that
\begin{align*}
|v^\top \Sigma_i v - v^\top \tilde\Sigma^{-1/2} \Sigma_i \tilde\Sigma^{-1/2} v|
&= |v^\top (\Sigma_i - \tilde\Sigma^{-1/2} \Sigma_i \tilde\Sigma^{-1/2}) v|
\leq \norm{v}^2 \norm{\Sigma_i - \tilde\Sigma^{-1/2} \Sigma_i \tilde\Sigma^{-1/2}}\\
&= \norm{\Sigma_i - (I_d + E) \Sigma_i (I_d + E)}
\leq 2\norm{E} \norm{\Sigma_i} + \norm{E}^2 \norm{\Sigma_i}\\
&\leq O(w_{\min}^{-1} \gamma)\,,
\end{align*}
where we used that $E = \tilde{\Sigma}^{-1/2} - I_d$ satisfies $\norm{E} \leq O(\gamma)$ and $\norm{\Sigma_i} \leq w_{\min}^{-1}$.
Therefore, because $w_{\min}^{-1} \gamma \ll \alpha$, any unit vector $v \in \R^d$ with $v^\top \Sigma_i v \leq \Omega(\alpha)$ also satisfies $v^\top \tilde\Sigma^{-1/2} \Sigma_i \tilde\Sigma^{-1/2} v \leq \Omega(\alpha)$.
This finishes the proof.

$\qed$

\subsection{Algorithm for Clustering Mixtures of Centered Gaussians}\label{sec:zero-mean-main}
We first present the algorithm and its main subroutines and then conclude by proving~\Cref{thm:zero-mean-main}.
Our algorithm has two main subroutines:
\begin{enumerate}
    \item \textbf{Partial clustering refinement:} This subroutine takes a partial clustering $\mathcal{S}$ and produces a list of candidate refinements with the following properties:
    \begin{enumerate}
        \item Every partial clustering $\mathcal{S}'$ in the list is a refinement of $\mathcal{S}$,
        \item The size of the output list is bounded by $f(w_{\min}^{-1}, \Delta)$ for some function $f$,
        \item If $\mathcal{S}$ is ``good", then with high probability the output list contains a partial clustering that is good.
    \end{enumerate}
    \item \textbf{Clustering selection:} This subroutine takes a list of candidate clusterings that contains at least one ``good" clustering and outputs a good clustering.
\end{enumerate}

\begin{algorithm}[H]
    \SetAlgoLined
    \SetKwInOut{Input}{input}
    \SetKwInOut{Output}{output}

    \Input{a collection of samples $\mathcal{T}$, the number of components $k$, the minimum mixing weight $w_{\min}$, the separation parameter $\Delta$}
    \Output{a clustering $\mathcal{S}$ of $\mathcal{T}$}

    Partition $\mathcal{T} = \mathcal{T}_1 \cup \mathcal{T}_2$ by putting each sample in $\mathcal{T}_1$ independently with probability $1/2$\footnote{This increases the fraction of outliers in $\mathcal{T}_1$ and $\mathcal{T}_2$ with high probability by at most a constant factor.}\;
    Let $\mathcal{C} = \Set{\mathcal{S}}$ where $\mathcal{S}$ is the trivial partial clustering with one subset $\mathcal{T}_1$\;
    \While{$\exists\, \mathcal{S} \in \mathcal{C}$ such that $\vert \mathcal{S} \vert < k$}{
        Remove $\mathcal{S}$ from $\mathcal{C}$\;
        Let $L$ be the output of the partial clustering refinement algorithm (\Cref{alg:zero-mean-refinement}) on $\mathcal{S}$\;
        Add all $\mathcal{S}' \in L$ to $\mathcal{C}$\;
    }
    Run the clustering selection algorithm (\Cref{claim:list-reduction}) on $\mathcal{C}$ with upper bound $\Delta^{-1/k^{(4k^2)}}$ on the fraction of corruptions and with new samples $\mathcal{T}_2$, and let $\mathcal{S}$ be the returned clustering\;
    \Return{$\mathcal{S}$}\;
\caption{Mixtures of centered Gaussians learning algorithm}
\label{alg:zero-mean-cluster}
\end{algorithm}

\begin{algorithm}[H]
    \SetAlgoLined
    \SetKwInOut{Input}{input}
    \SetKwInOut{Output}{output}

    \Input{a candidate partial clustering $\mathcal{S}$, the number of components $k$, the minimum mixing weight $w_{\min}$, the separation parameter $\Delta$}
    \Output{a list of candidate refinements of $\mathcal{S}$}
    Let $L = \varnothing$\;
    \For{each $S \in \mathcal{S}$ and each guess of $\epsilon$ with $2^{-\Delta}$-bit complexity that satisfies $\epsilon \geq \Delta^{-1}$}{
        \tcp{Frobenius clustering}
        \For{$2^{O(w_{\min}^{-1} \log (k/\epsilon^{1/4}))}$ rounds}{
            Let $S = S_1 \cup S_2$ be the result of the Frobenius partial clustering algorithm (\Cref{fact:frobenius-clustering}) on $S$ with parameters $\epsilon$, $\alpha=w_{\min}$, $t=4$, and $\beta=\epsilon$\;
            Add $(\mathcal{S}\setminus S) \cup \{S_1, S_2\}$ to $L$
        }
        Put the samples of $S$ in isotropic position with the algorithm corresponding to~\Cref{thm:robust-isotropic-position}\;
        \tcp{Spectral clustering}
        Let $Q$ be the subspace returned by the zero-mean subspace recovery algorithm (\Cref{thm:zero-mean-subspace}) run on $S$ with fraction of outliers set to $\epsilon^{1/12}$ (but abort and set $Q=\emptyset$ if the algorithm performs more than $f(\Delta) \cdot \poly(d)$ steps for some function $f$)\;
        If $\dim(Q) > f(\Delta)$ for some function $f$, set $Q=\emptyset$\;
        \For{each unit vector $v$ in a $\Delta^{-1}$-net of $Q$}{
            \For{each $\epsilon$-resolution choice of $\tau$ in $\left[\epsilon, 2w_{\min}^{-1}\right]$}{
                Partition the samples $S = S_1 \cup S_2$ based on whether their projection on $v$ is in the interval $[-\tau, \tau]$ or outside it\;
                Add $(\mathcal{S}\setminus S) \cup \{S_1, S_2\}$ to $L$\;
            }
        }
    }
    \Return{$L$}\;
\caption{Mixtures of centered Gaussians partial clustering refinement algorithm}
\label{alg:zero-mean-refinement}
\end{algorithm}

We first state a few facts that we use in the analysis of the algorithm.

\begin{fact}[Lemma 5.1 in~\cite{diakonikolas2023spectral}]
\label{fact:misc-linalg-12}
    Consider two arbitrary positive definite matrices $\Sigma_1 \in \R^{d \times d}$ and $\Sigma_2 \in \R^{d \times d}$, and suppose there exists a positive definite matrix $H \in \R^{d \times d}$ such that
    \[\Norm{I_d - H^{-1/2} \Sigma_1 H^{-1/2}}_F \leq \rho, \quad \Norm{I_d - H^{-1/2} \Sigma_2 H^{-1/2}}_F \leq \rho\,.\]
    Then, for an arbitrary positive definite matrix $\Sigma \in \R^{d \times d}$, we have
    \[ \Norm{\Sigma^{-1/2} \Sigma_1 \Sigma^{-1/2} - \Sigma^{-1/2} \Sigma_2 \Sigma^{-1/2}}_F \leq 5 \rho \max \left(\Norm{\Sigma^{-1/2} \Sigma_1 \Sigma^{-1/2}},\Norm{\Sigma^{-1/2} \Sigma_2 \Sigma^{-1/2}}\right) \,. \]
\end{fact}

\begin{claim}[Quadratic form with approximate vectors]
\label{lem:quadratic-form-closeness}
Let $\Sigma \in \R^{d \times d}$ be a matrix with $\norm{\Sigma} \leq \Delta$. Suppose $v \in \R^d$ is a unit vector and $w \in \R^d$ is a vector such that $\norm{v-w} \leq \delta \leq 1$. Then ${|v^\top \Sigma v - w^\top \Sigma w| \leq O(\delta\Delta)}$.
\end{claim}
\begin{proof}
We have 
\begin{align*}
w^\top \Sigma w
&= (v+(w-v))^\top \Sigma (v+(w-v))
\leq v^\top \Sigma v + 2\norm{w-v} \norm{\Sigma} + \norm{w-v}^2 \norm{\Sigma}
\leq v^\top \Sigma v + O(\delta \Delta)
\end{align*}
and 
\begin{align*}
w^\top \Sigma w
&= (v+(w-v))^\top \Sigma (v+(w-v))
\geq v^\top \Sigma v - 2\norm{w-v} \norm{\Sigma} - \norm{w-v}^2 \norm{\Sigma}
\geq v^\top \Sigma v - O(\delta \Delta)\,.
\end{align*}
\end{proof}

\begin{claim}[Variance thresholding]
\label{lem:spectral-separation-threshold}
Consider two arbitrary positive definite matrices $\Sigma_1 \in \R^{d \times d}$ and $\Sigma_2 \in \R^{d \times d}$.
Let $v \in \R^d$ and suppose $v^\top \Sigma_1 v < v^\top \Sigma_2 v$. Denote ${\sigma_1 = \sqrt{v^\top \Sigma_1 v}}$ and ${\sigma_2 = \sqrt{v^\top \Sigma_2 v}}$. Then 
\[\Pr_{\bm x \sim N(0, \Sigma_1)}\Paren{\langle \bm x, v\rangle \not\in [-\sqrt{\sigma_1\sigma_2}, \sqrt{\sigma_1\sigma_2}]} \leq \frac{\sigma_1}{\sigma_2}\]
and 
\[\Pr_{\bm x \sim N(0, \Sigma_2)}\Paren{\langle \bm x, v\rangle \in [-\sqrt{\sigma_1\sigma_2}, \sqrt{\sigma_1\sigma_2}]} \leq O\Paren{\sqrt{\frac{\sigma_1}{\sigma_2}}}\,.\]
\end{claim}
\begin{proof}
We use that, for an arbitrary positive definite matrix $\Sigma \in \R^{d \times d}$ and $\bm{x} \sim N(0, \Sigma)$, $\langle \bm x, v\rangle$ is distributed according to $N(0, v^\top \Sigma v)$.
Then, for the first inequality, we have by Chebyshev's inequality 
\[\Pr_{\bm x \sim N(0, \Sigma_1)}\Paren{|\langle \bm x, v\rangle| \geq \sigma_1 \cdot \sqrt{\frac{\sigma_2}{\sigma_1}}} \leq \frac{\sigma_1}{\sigma_2}\,.\]
For the second inequality, we have by the anti-concentration properties of Gaussians that
\begin{align*}
\Pr_{\bm x \sim N(0, \Sigma_2)}\Paren{|\langle \bm x, v\rangle| \leq \sigma_2 \cdot \sqrt{\frac{\sigma_1}{\sigma_2}}}
&\leq O\Paren{\sqrt{\frac{\sigma_1}{\sigma_2}}}\,.
\end{align*}
\end{proof}

\subsubsection{Analysis of Partial Clustering Refinement (\Cref{alg:zero-mean-refinement})}
We now show that the refinement subroutine has the desired properties.
First, we show that the subroutine produces a good partial clustering in the case of relative Frobenius separation.

\begin{lemma}[Frobenius clustering result]
\label{lem:refinement-frobenius-correctness}
    Let $\mathcal{S}$ be a $(1-\epsilon)$-good partial clustering.
    Suppose there exists some $S \in \mathcal{S}$ with $|S| \geq \poly(d,k^k,\Delta)$ that contains at least two components $i, j \in \mathsf{comp}(S)$ such that $\Norm{\Sigma_S^{-1/2}(\Sigma_i-\Sigma_j)\Sigma_S^{-1/2}}_F \geq \epsilon^{-1}$, where $\Sigma_S = \sum_{i \in \mathsf{comp}(S)} w'_i \Sigma_i$ for $w'_i = \frac{w_i}{\sum_{i \in \mathsf{comp}(S)} w_i}$.
    Then, for $\Delta^{-1} \leq \epsilon \leq w_{\min}^{O(1)}$, the output of~\Cref{alg:zero-mean-refinement} contains with high probability some $\mathcal{S}'$ that is a \mbox{$(1-\epsilon')$-good} partial clustering for $\epsilon' = O(w_{\min}^{-1}\epsilon^{1/4})$.
\end{lemma}
\begin{proof}
    Consider the iteration of the loop in which we select a cluster $S$ and a value $\epsilon$ that correspond to the assumptions.
    We show that with high probability in this case one of the iterations of the Frobenius clustering inner loop will produce a $(1-\epsilon')$-good refinement of $\mathcal{S}$.

    By the assumptions, there exist some $i, j \in \mathsf{comp}(S)$ such that $\Norm{\Sigma_S^{-1/2}(\Sigma_i-\Sigma_j)\Sigma_S^{-1/2}}_F \geq \epsilon^{-1}$.
    Then to run the partial clustering algorithm in~\Cref{fact:frobenius-clustering} with $\epsilon$, $\alpha=w_{\min}$, $t=1$, and $\beta=\epsilon^{1/4}$,
    we require that $\Norm{\Sigma_S^{-1/2}(\Sigma_i - \Sigma_j)\Sigma_S^{-1/2}}_F^2 \geq \Omega(k^2/(\beta^2\alpha^4)) = \Omega(k^2w_{\min}^{-4} \epsilon^{-1/2})$, which is smaller than $\epsilon^{-1}$ for our choice of parameters.
    Therefore the partial clustering algorithm succeeds with probability $2^{-O\Paren{w_{\min}^{-1}\log(k/\epsilon^{1/4})}}$.
    
    If the partial clustering algorithm succeeds, then it partitions the samples into two while making an error on at most an $O\Paren{\beta+\epsilon/\alpha^4}$-fraction of the samples from each component, which is dominated by $O(\beta)$.
    Therefore in each part the fraction of outliers increases by at most $O(w_{\min}^{-1}\beta) = O(w_{\min}^{-1}\epsilon^{1/4})$, and the new clustering is $(1-\epsilon')$-good for $\epsilon' = O(\epsilon^{1/4})$.
    Furthermore, since we run the algorithm $2^{O\Paren{w_{\min}^{-1}\log(k/\epsilon^{1/4})}}$ times, with high probability at least one of these runs succeeds and the corresponding iteration adds a $(1-\epsilon')$-good refinement to the output list.
    
    Finally, we note that we do not need the exact value of $\epsilon$ and an upper bound within $2^{-\Delta}$ of the truth suffices.
\end{proof}

Second, we give a similar result in the case of spectral separation.

\begin{lemma}[Spectral clustering result]
\label{lem:refinement-spectral-correctness}
    Let $\mathcal{S}$ be a $(1-\epsilon)$-good partial clustering.
    Suppose there exists some $S \in \mathcal{S}$ with the following properties:
    \begin{enumerate}
        \item $|S| \geq \poly(d, \Delta)$,
        \item $S$ contains at least two components $i, j \in \mathsf{comp}(S)$ such that there exists some unit vector $v \in \R^d$ with $v^\top \Sigma_i v > \Delta^{1/2} \cdot v^\top \Sigma_j v$,
        \item The mixture corresponding to components in $\mathsf{comp}(S)$ is in approximate isotropic position: $(1-\epsilon^{3})I_d \preceq \sum_{i \in \mathsf{comp}(S)} w'_i \Sigma_i \preceq (1+\epsilon^{3})I_d$ and $\Norm{\sum_{i \in \mathsf{comp}(S)} w'_i \Sigma_i}_F \leq \epsilon^{3}$, where $w'_i = \frac{w_i}{\sum_{i \in \mathsf{comp}(S)} w_i}$,
        \item For all $i \in \mathsf{comp}(S)$, we have that $\norm{\Sigma_i - I_d}_F^2 \leq f(\Delta)$ for some function $f$.
    \end{enumerate}
    Then, for $\Delta^{-1} \leq \epsilon \leq w_{\min}^{O(k)}$, the output of~\Cref{alg:zero-mean-refinement} contains with high probability some $\mathcal{S}'$ that is a $(1-\epsilon')$-good partial clustering for $\epsilon' = O(w_{\min}^{-2} \epsilon^{1/(4096k)})$. 
\end{lemma}
\begin{proof}
    Consider the iteration of the loop in which we select a cluster $S$ and a value $\epsilon$ that correspond to the assumptions.
    We show that with high probability in this case the inner loop will produce a $(1-\epsilon')$-good refinement of $\mathcal{S}$.

    When we run the subspace recovery algorithm in~\Cref{thm:zero-mean-subspace}, we obtain with high probability a subspace such that, for each unit $v$ for which $v^\top \Sigma_i v \leq \Omega(w_{\min}^{-4}\epsilon^{1/128})$ for some $i \in \mathsf{comp}(S)$, there exists a unit vector in the subspace that is $O(w_{\min}^{-1/2}\epsilon^{1/1024})$-close to it. 

    By the assumptions, we have that for some $i, j\in\mathsf{comp}(S)$ there exists a unit vector $v \in \R^d$ with $v^\top \Sigma_i v > \Delta^{1/2} \cdot v^\top \Sigma_j v$. 
    By the approximate isotropic position, we can assume that ${v^\top \Sigma_i v \leq 2w_{\min}^{-1}}$, so $v^\top \Sigma_j v \leq O(w_{\min}^{-1}\Delta^{-1/2})$.
    Then the guarantee from~\Cref{thm:zero-mean-subspace} shows that there exists a unit vector in the recovered subspace that is $O(w_{\min}^{-1/2}\epsilon^{1/1024})$-close to $v$ as long as ${v^\top \Sigma_j v \leq \Omega(w_{\min}^{-4}\epsilon^{1/128})}$, which we satisfy because $v^\top \Sigma_j v \leq O(w_{\min}^{-1}\Delta^{-1/2}) \ll w_{\min}^{-4} \epsilon^{1/128}$.
    Then there also exists a unit vector $v'$ in the \mbox{$\Delta^{-1}$-net} of the subspace that is $O(w_{\min}^{-1/2}\epsilon^{1/1024}+\Delta^{-1})$-close to $v$, which is dominated by $O(w_{\min}^{-1/2}\epsilon^{1/1024})$.

    Then, by~\Cref{lem:quadratic-form-closeness} and using that $\norm{\Sigma_j} \leq 2w_{\min}^{-1}$, we have that $v'$ satisfies $(v')^\top \Sigma_j (v') \leq  O(w_{\min}^{-1}\Delta^{-1/2}) + O(w_{\min}^{-3/2}\epsilon^{1/1024})$, which is dominated by $O(w_{\min}^{-3/2}\epsilon^{1/1024})$.
    Then, using that because of the approximate isotropic position $\max_i (v')^\top \Sigma_i (v') \geq 1/2$, if we sort all components in increasing order by their variance in direction $v'$, there exist two consecutive indices $i$ and $i+1$ such that $(v')^\top \Sigma_{i+1} v' \geq \Omega(w_{\min}^{-3/2} \epsilon^{1/1024})$ and $(v')^\top \Sigma_i v'/{(v')^\top \Sigma_{i+1} v' \leq O(w_{\min}^{-3/(2k)} \epsilon^{1/(1024k)})}$. Consider partitioning the samples projected along direction $v'$ with the threshold $\tau = \sqrt{(v')^\top \Sigma_{i} v' \cdot (v')^\top \Sigma_{i + 1} v'}$. Then by~\Cref{lem:spectral-separation-threshold} an independent sample from the mixture is partitioned erroneously with probability at most $O(w_{\min}^{-3/(8k)} \epsilon^{1/(4096k)})$.

    Then, by Hoeffding's inequality, we have that with probability $\exp\Paren{-|S| \cdot O(w_{\min}^{-3/(4k)} \epsilon^{1/(2048k)})}$ at most an $O(w_{\min}^{-3/(8k)} \epsilon^{1/(4096k)})$-fraction of the samples are partitioned erroneously.
    In the worst case in which one part contains only one component with weight $w_{\min}$, the fraction of outliers in it is at most $O(w_{\min}^{-2} \epsilon^{1/(4096k)})$. 

    A problem we omitted is that the bound on the probability that a sample is partitioned erroneously holds for a fixed direction independent of the samples, but $v'$ can depend on the samples.
    We can show however by a union bound that the result holds for all $v'$.
    Consider a $\poly(d^{-1}, \Delta^{-1})$-net of $\R^d$, which has size $\poly(d, \Delta)^d$.
    The probability that a sample is partitioned erroneously in a direction that is $\poly(d^{-1}, \Delta^{-1})$-close to $v'$ is the same up to constants to the probability for direction $v'$, so the bound by Hoeffding's inequality is also essentially the same.
    Then to union bound over all the directions in the net it suffices to have $|S| \geq \poly(d, \Delta)$.

    It may also happen that $(v')^\top \Sigma_i v'$ is \emph{much} smaller than $(v')^\top \Sigma_{i+1} v'$, and that the threshold $\tau = \sqrt{(v')^\top \Sigma_{i} v' \cdot (v')^\top \Sigma_{i + 1} v'}$ falls below the interval in which we search for it. However, we obtain the same guarantees also with the threshold $\tau = \sqrt{\max\Set{(v')^\top \Sigma_{i} v', \epsilon} \cdot (v')^\top \Sigma_{i + 1} v'}$, so this is not an issue.

    Furthermore, we note that we obtain essentially the same probability of partitioning a sample erroneously even if the threshold is within a constant multiplicative factor from the threshold $\tau$ we analyzed.
    Also note that the threshold is trivially upper bounded by $2w_{\min}^{-1}$, because $\norm{\Sigma_i} \leq 2w_{\min}^{-1}$ for all $i \in [k]$.

    Thus, when we loop over every unit vector in the net, we find a direction $v'$ for which we produce a $(1-\epsilon')$-good partial clustering for $\epsilon' = O(w_{\min}^{-2} \epsilon^{1/(4096k)})$.

    Finally, we note that we do not need the exact value of $\epsilon$ and an upper bound within $2^{-\Delta}$ of the truth suffices.
\end{proof}

We also show that the refinement procedure produces a list of dimension-independent size.

\begin{lemma}[Partial clustering refinement produces a small list]
\label{lem:refinement-small}
    For all partitions $\mathcal{S}$ of the sample universe such that $\vert \mathcal{S} \vert < k$, the size of the list returned by~\Cref{alg:zero-mean-refinement} has size at most $f(w_{\min}^{-1}, \Delta)$ for some function $f$.
\end{lemma}
\begin{proof}
    We first consider the number of candidate refinements added to the list in a single iteration of the loop.
    In the first part of the loop, we run the Frobenius clustering algorithm $2^{O(w_{\min}^{-1} \log (k/\epsilon))}$ times and for each run of the algorithm we add one possible refinement to the output list, where $\epsilon \geq \Delta^{-1}$.
    In the second part of the loop, we run the centered mixture subspace recovery algorithm and for each $v$ in a $\Delta^{-1}$-net of the outputted subspace we add at most $f(w_{\min}^{-1}, \Delta)$ possible refinements to the output list for some function $f$.
    The dimension of the subspace outputted by the centered mixture subspace recovery algorithm in~\Cref{thm:zero-mean-subspace} is at most some function of $\Delta$.
    Then, the size of the $\Delta^{-1}$-net of the subspace is also bounded by some function of $\Delta$. 
    Finally, the number of iterations of the loop is bounded by a function exponential in $\Delta$.
\end{proof}

Finally, we combine the previous lemmas to give the overall guarantees of the refinement procedure.

\begin{lemma}[Correctness of partial clustering refinement]
\label{lem:zero-mean-refinement-correctness}
    Given a partial clustering $\mathcal{S}$ with $\vert \mathcal{S} \vert < k$, \Cref{alg:zero-mean-refinement} outputs a list of partial clusterings $L$  with the following properties:
    \begin{enumerate}
        \item Every $\mathcal{S}' \in L$ is a refinement of $\mathcal{S}$,
        \item The size of $L$ is bounded by $f(w_{\min}^{-1}, \Delta)$ for some function $f$,
        \item If $\mathcal{S}$ is a $(1-\epsilon)$-good partial clustering for $\Delta^{-1} \leq \epsilon \leq w_{\min}^{(2k^{2k})}$ with $|S| \geq \poly(d,k^k,\Delta)$ for all $S \in \mathcal{S}$, then the output of the algorithm contains with high probability at least one refinement $\mathcal{S}'$ that is a $(1-\epsilon')$-good clustering with $\epsilon' = O(w_{\min}^{-1} \epsilon^{1/(4k^{2k})})$.
    \end{enumerate}
\end{lemma}
\begin{proof}
    The first property is immediate from the definition of the algorithm. The second property follows by~\Cref{lem:refinement-small}.

    For the rest of the proof, suppose that $\mathcal{S}$ is a $(1-\epsilon)$-good partial clustering with $|S| \geq \poly(d,k^k,\Delta)$ for all $S \in \mathcal{S}$.
    Because the size of $\mathcal{S}$ is less than $k$, for some cluster $S \in \mathcal{S}$ we have $|\mathsf{comp}(S)| > 1$.
    For such a cluster with more than one component, let ${\Sigma_S = \sum_{i \in \mathsf{comp}(S)} w'_i \Sigma_i}$ for ${w'_i = \frac{w_i}{\sum_{i \in \mathsf{comp}(S)} w_i}}$.
    If there exist two $i, j\in \mathsf{comp}(S)$ such that $\Norm{\Sigma_S^{-1/2}(\Sigma_i-\Sigma_j)\Sigma_S^{-1/2}}_F \geq \epsilon^{-1/k^{2k}}$, then we satisfy the conditions of~\Cref{lem:refinement-frobenius-correctness} and with high probability the Frobenius clustering produces a $(1-O(w_{\min}^{-1}\epsilon^{1/(4k^{2k})}))$-good clustering.

    Else, $\Norm{\Sigma_S^{-1/2}(\Sigma_i-\Sigma_j)\Sigma_S^{-1/2}}_F \leq \epsilon^{-1/k^{2k}}$ for all $i,j \in \mathsf{comp}(S)$, so the robust isotropic position algorithm corresponding to~\Cref{thm:robust-isotropic-position} succeeds with high probability and we can assume the mixture is in approximate isotropic position, with new mixture mean and covariance ${\Norm{\mu'} \leq \epsilon^{1/4}}$, ${(1-\epsilon^{1/4}) I_d \preceq \Sigma' \preceq (1+\epsilon^{1/4}) I_d}$, and ${\Norm{\Sigma' - I_d}_F \leq \epsilon^{1/4}}$.
    Denote by $\Sigma_1', \ldots, \Sigma_k'$ the component covariances after this robust isotropic position transformation.
    Then we also have that \linebreak ${\Norm{(\Sigma')^{-1/2}(\Sigma_i'-\Sigma_j')(\Sigma')^{-1/2}}_F \leq \epsilon^{-1/k^{2k}}}$ for all $i, j \in \mathsf{comp}(S)$, so
    \begin{align*}
    \Norm{\Sigma_i'-\Sigma_j'}_F
    &= \Norm{(\Sigma')^{1/2}(\Sigma')^{-1/2}(\Sigma_i'-\Sigma_j')(\Sigma')^{-1/2}(\Sigma')^{1/2}}_F\\
    &\leq \Norm{(\Sigma')^{1/2}}^2 \cdot \Norm{(\Sigma')^{-1/2}(\Sigma_i'-\Sigma_j')(\Sigma')^{-1/2}}_F\\
    &\leq 2\epsilon^{-1/k^{2k}}\,.
    \end{align*}
    Note that the approximate isotropic position implies that $\Norm{\sum_{i \in \mathsf{comp}(S)} w'_i \Sigma_i' - I_d}_F \leq \epsilon^{1/4}$, which is false unless $\Norm{\Sigma_i' - I_d}_F \leq 4\epsilon^{-1/k^{2k}}$ for all $i \in \mathsf{comp}(S)$. This is because, using that $\Norm{\Sigma_i'-\Sigma_j'}_F \leq 2\epsilon^{-1/k^{2k}}$ for all $i, j \in \mathsf{comp}(S)$, we get that for any fixed $i \in \mathsf{comp}(S)$ all $\Sigma_j'$ with $j \in \mathsf{comp}(S)$ are $2\epsilon^{-1/k^{2k}}$-close to $\Sigma_i'$, so $\sum_{j \in \mathsf{comp}(S)} w'_j \Sigma_j'$ is also $2\epsilon^{-1/k^{2k}}$-close to $\Sigma_i'$. But if $\Norm{\Sigma_i' - I_d}_F > 4\epsilon^{-1/k^{2k}}$, then by the triangle inequality $\sum_{j \in \mathsf{comp}(S)} w'_j \Sigma_j'$ has distance more than $2\epsilon^{-1/k^{2k}}$ from $I_d$, which is a contradiction.
    Therefore, for the rest of the proof we assume that $\Norm{\Sigma_i' - I_d}_F \leq 4\epsilon^{-1/k^{2k}}$ for all $i \in \mathsf{comp}(S)$.
    
    Recall that all the components are $\Delta$-separated, so for each $i \neq j \in \mathsf{comp}(S)$ at least one of the following holds:
    \begin{itemize}
        \item Spectral separation: $\exists v \in \R^d$ such that $v^\top \Sigma_i' v < \frac{1}{\Delta^2} v^\top \Sigma_j' v$ or $v^\top \Sigma_i' v > \Delta^2 \cdot v^\top \Sigma_j' v$\,,
        \item Relative Frobenius separation: $\Norm{(\Sigma_i')^{-1/2} \Sigma_j' (\Sigma_i')^{-1/2} - I_d}^2_F > \Delta^2 \cdot \Norm{(\Sigma_i')^{-1/2} \Sigma_j' (\Sigma_i')^{-1/2}}^2$\,.
    \end{itemize}

We argue now that, because $\Norm{\Sigma_i' - I_d}_F \leq 4\epsilon^{-1/k^{2k}}$ for all $i \in \mathsf{comp}(S)$, there must exist two components with spectral separation.
By~\Cref{fact:misc-linalg-12}, for two $i, j \in \mathsf{comp}(S)$ we have
\begin{align*}
&\norm{(\Sigma_i')^{-1/2}\Sigma_j' (\Sigma_i')^{-1/2} - I_d}_F^2\\
&\quad \leq 25 \max\Set{\Norm{\Sigma_i' - I_d}_F^2, \Norm{\Sigma_j' - I_d}_F^2} \cdot  \max\Set{1, \Norm{(\Sigma_i')^{-1/2}\Sigma_j' (\Sigma_i')^{-1/2}}^2}\\
&\quad \leq 400 \epsilon^{-2/k^{2k}} \cdot \max\Set{1, \Norm{(\Sigma_i')^{-1/2}\Sigma_j' (\Sigma_i')^{-1/2}}^2}\\
&\quad = 400 \epsilon^{-2/k^{2k}} \cdot \max\Set{1, \frac{1}{\Norm{(\Sigma_i')^{-1/2}\Sigma_j' (\Sigma_i')^{-1/2}}^2}} \cdot \Norm{(\Sigma_i')^{-1/2}\Sigma_j' (\Sigma_i')^{-1/2}}^2\,.
\end{align*}
If $\Norm{(\Sigma_i')^{-1/2}\Sigma_j' (\Sigma_i')^{-1/2}}^2 > 400\epsilon^{-2/k^{2k}}\Delta^{-2}$, we have from the above that components $i, j$ do not satisfy relative Frobenius separation with parameter $\Delta$, so they must be spectrally separated with parameter $\Delta$.
On the other hand, if $\Norm{(\Sigma_i')^{-1/2}\Sigma_j' (\Sigma_i')^{-1/2}}^2 \leq 400\epsilon^{-2/k^{2k}}\Delta^{-2}$, it means that 
\[(\Sigma_i')^{-1/2}\Sigma_j' (\Sigma_i')^{-1/2} \preceq \Paren{20\epsilon^{-1/k^{2k}}\Delta^{-1}} I_d\,, \]
so
\[\Sigma_j' \preceq \Paren{20\epsilon^{-1/k^{2k}}\Delta^{-1}} \Sigma_i'\,,\]
so for any vector $v \in \R^d$ we have $v^\top \Sigma_j' v \leq \Paren{20\epsilon^{-1/k^{2k}}\Delta^{-1}} v^\top \Sigma_i' v$, so components $i,j$ are spectrally separated with parameter $\epsilon^{1/(2k^{2k})}\Delta^{1/2}/\sqrt{20} \geq \Delta^{1/4}$.

Therefore there exist two components that are spectrally separated with parameter at least $\Delta^{1/4}$.
Then we apply~\Cref{lem:refinement-spectral-correctness} with the assumption that we have a $(1-\epsilon^{1/12})$-good partial clustering (so that the approximate isotropic position requirement of~\Cref{lem:refinement-spectral-correctness} is satisfied), and conclude that with high probability the outputted list contains a ${(1-O(w_{\min}^{-2} \epsilon^{1/(50000k)}))}$-good clustering.

Thus, in both the case of Frobenius separation and spectral separation, we produce with high probability at least a $(1-O(w_{\min}^{-1} \epsilon^{1/(4k^{2k})}))$-good clustering.
\end{proof}

\subsubsection{Proof of~\Cref{thm:zero-mean-main}}
\label[appendix]{sec:proof-of-zero-mean}
Note that all partial clusterings that are contained in $\mathcal{C}$ in~\Cref{alg:zero-mean-cluster} at any point are naturally associated with a tree where the children of a partial clustering $\mathcal{S}$ are the partial clustering $\mathcal{S}'$ produced when we run the refinement algorithm on $\mathcal{S}$.
Also note that the number of clusters increases by one at each level and is bounded by $k$ and thus the depth is at most $k$.

For correctness, note that by~\Cref{lem:zero-mean-refinement-correctness} with high probability this tree contains at least one path from the root to a leaf such that at each level if the node is a $(1-\epsilon)$-good partial clustering with $\epsilon \geq \Delta^{-1}$ then the next node in the path is an $(1-\epsilon')$-good partial clustering with $\epsilon' = O(w_{\min}^{-1} \epsilon^{1/(4k^{2k})})$. Conditioning on the existence of a path from the root to level $i$, we can extend the path to level $i+1$ with high probability by~\Cref{lem:zero-mean-refinement-correctness}.
By taking a union bound over the $k-1$ steps of~\Cref{alg:zero-mean-refinement} needed to extend the path from the root to the leaves, we have that this path will exist with high probability. 
We can assume we start with $\epsilon = \max\Set{O(\epsilon^*), \Delta^{-1}}$, such that with high probability the samples in $\mathcal{T}_1$ have at most an $\epsilon$-fraction of outliers.
Then, we have that at level $i$ the partial clustering is a $\left(1-O\Paren{w_{\min}^{-2} \epsilon^{1/((4k^{2k}))^i}}\right)$-good clustering when $\epsilon^{-1} \geq f(w_{\min}^{-1})$ is large enough, which we lower bound by $1-\epsilon^{1/k^{(4k^2)}}$ for all $i \leq k$. 

Thus for $(\epsilon^*)^{-1} \geq f(w_{\min}^{-1})$ large enough there exists some leaf that corresponds to a candidate clustering in the final list $\mathcal{C}$ with few outliers, as required by~\Cref{claim:list-reduction}.
Then by~\Cref{claim:list-reduction} we have that we output a $(1-O(kw_{\min}^{-1}\max\{\epsilon, \Delta^{-\Omega(1/k)}\}))$-good clustering of the input samples.

We argue now that the loop in~\Cref{alg:zero-mean-cluster} terminates in $f(w_{\min}^{-1}, \Delta)$ iterations and that the resulting list of candidate clusterings has size at most $f( w_{\min}^{-1}, \Delta)$ for some function $f$.
For all $\mathcal{S}$ which are not leaves in this tree (and thus have size strictly less than $k$), by~\Cref{lem:zero-mean-refinement-correctness} their number of children is bounded by some function of $w_{\min}^{-1}$ and $\Delta$.
Therefore, the total number of partial clusters ever included in $\mathcal{C}$ is bounded by some function of $w_{\min}^{-1}$ and $\Delta$, and since we process each of these at most once the number of total iterations is also bounded by some function of $w_{\min}^{-1}$ and $\Delta$. 

Furthermore, each iteration of the loop in~\Cref{alg:zero-mean-cluster} takes time at most $g(w_{\min}^{-1}, \Delta) \cdot \poly(d)$ for some function $g$, and thus overall the loop, and the entire algorithm, finishes in time $g(w_{\min}^{-1}, \Delta) \cdot \poly(d)$.

$\qed$

\section{Clustering Mixtures of Identical-Covariance Gaussians}

We define now a notion of separation for non-spherical distributions with identical covariances.
For identical-covariance Gaussian components this corresponds to total variation separation of $1-f(1/w_{\min})$ where $w_{\min}$ is the minimum mixing weight of any component in the input mixture (see~\Cref{fact:tv-parameters}).

\begin{definition}[Identical-covariance parameter distance]
We say that two distributions with means $\mu_1, \mu_2 \in \R^d$ and identical covariance $\Sigma \in \R^{d \times d}$ are $\Delta$-separated if there exists a vector $v \in \R^d$ such that $\langle \mu_1 - \mu_2, v\rangle^2 > 2\Delta^2 \cdot v^\top \Sigma v$.
\end{definition}

We state now our main result for clustering mixtures of identical-covariance Gaussians.

\begin{theorem}[Main theorem, identical-covariance components]
\label{thm:same-cov-clustering}
Let $\mathcal{M}$ be a $d$-dimensional mixture of $k$ Gaussians $\sum_{i=1}^k w_i N(\mu_i, \Sigma)$ with $w_{\min} = \min_i w_i$ and $\Sigma \succ 0$.
Furthermore, assume all components $i \neq j$ are $\Delta$-separated.
Also let $\mathcal{M}'$ be a distribution satisfying $d_{\TV}(\mathcal{M}', \mathcal{M}) \leq \epsilon^*$.
Suppose $\Delta, (\epsilon^*)^{-1} \geq f(w_{\min}^{-1})$ for some function $f$.
Then, given $\poly(d^{\log w_{\min}^{-1}}, (w_{\min}^{-1})^{\log w_{\min}^{-1}}, \Delta)$ i.i.d. samples from $\mathcal{M}'$, there exists an algorithm that runs in time $g(w_{\min}^{-1}, \Delta) \cdot \poly(d^{\log^2 w_{\min}^{-1}})$ for some function $g$ and outputs with high probability a partition of the samples $\hat{S}_1, \ldots, \hat{S}_k$ such that, if we let $S_i$ be the set of samples generated according to the $i$-th component, then with high probability  (up to a permutation of $\hat{S}_1, \ldots, \hat{S}_k$)\footnote{It suffices to take $f(x) = \exp(\exp(\tilde{O}(x^2)))$ and $g(x,y)=\exp(\exp(\exp(\exp(\tilde{O}(\min(x,y)^2)))))$.}
\[\min_i \frac{|\hat{S}_i \cap S_i|}{|S_i|} \geq 1-O(kw_{\min}^{-1}\max\{\epsilon^*, \Delta^{-\Omega(1/k)}\})\,.\]
\end{theorem}

\begin{remark}
\label{rem:equiv-obl-adp2}
By a recent result of Blanc and Valiant~\cite{blanc2024adaptive} that proved an equivalence between oblivious and adaptive adversaries, we can also generalize~\Cref{thm:same-cov-clustering} to the adaptive adversary setting.
See~\Cref{rem:equiv-obl-adp} for a more detailed discussion.
\end{remark}

The proof uses the fact that, when the distribution is in isotropic position, the norms of the means are bounded by $w_{\min}^{-1/2}$, and therefore in directions of separation $v$ the component variance $v^\top \Sigma v$ must be very small.
In fact, the mixture is separable in all directions $v$ in which $v^\top \Sigma v$ is very small, so it suffices to find such directions in order to partially cluster the mixture.

In~\Cref{sec:pp-find-sub} we give an algorithm that recovers a low-dimensional subspace containing directions of separation.
Then in~\Cref{sec:pp-full-alg} we use this subspace-finding subroutine to obtain an algorithm that satisfies the guarantees of~\Cref{thm:same-cov-clustering}.

\subsection{Finding Directions with Small Variance}
\label{sec:pp-find-sub}

Similarly to the case of centered components, we define an approximate notion of isotropic position. 
\begin{definition}[$\gamma$-approximate isotropic position, identical-covariance components]
    For $\gamma \leq 1$, a mixture is in $\gamma$-approximate isotropic position if we have that $\Norm{\sum_{i=1}^k w_i \mu_i} \leq \gamma$ and $(1-\gamma) I_d \preceq \sum_{i=1}^k w_i \mu_i \mu_i^\top + \Sigma \preceq (1+\gamma) I_d$. We say that the mixture is (exactly) isotropic when $\gamma=0$.
\end{definition}

\begin{theorem}[Subspace finding theorem, identical-covariance components]
\label{thm:same-cov-subspace-rounding}
    Consider a $d$-dimensional mixture of $k$ Gaussians $\sum_{i=1}^k w_i N(\mu_i, \Sigma)$ with $w_{\min} = \min_{i} w_i$ and $\Sigma \succ 0$ in $\gamma$-approximate isotropic position.

    Let $\epsilon > 0$ with $\epsilon \leq \Omega(w_{\min}^{160 \log w_{\min}^{-1}})$, and suppose $\gamma \leq \Omega(w_{\min}^{40 \log w_{\min}^{-1}} \epsilon)$.
    Then, given\linebreak $\poly(d^{\log w_{\min}^{-1}}, (w_{\min}^{-1})^{\log w_{\min}^{-1}}, \epsilon^{-1})$ samples from the mixture with an \mbox{$\epsilon$-fraction} of corruptions, there exists an algorithm that runs in time $f(w_{\min}^{-1}, \epsilon^{-1}) \cdot \poly(d^{\log^2 w_{\min}^{-1}})$ and outputs an orthogonal projection matrix $Q \in \R^{d \times d}$ with $\rank(Q) \leq g(w_{\min}^{-1}, \epsilon^{-1})$ such that with high probability, for every unit vector $v \in \R^d$ that lies in the subspace of eigenvectors of $\Sigma$ with eigenvalue at most $1/2$, there exists a unit vector $v'$ such that $Qv' = v'$ and ${\norm{v-v'} \leq O(w_{\min}^{-1/8} \epsilon^{1/(640\log w_{\min}^{-1})})}$.
\end{theorem}

\textbf{Proof outline.} There are three main components to the proof of~\Cref{thm:zero-mean-subspace}:
     \begin{enumerate}
         \item In~\Cref{sec:pp-1} we construct a system of polynomial inequalities $\cA$ from the \emph{exact} $O(\log w_{\min}^{-1})$ moments of the mixture that identifies the directions in which the variance is $\approx 1$.
        \item In~\Cref{sec:pp-2} we show we can construct a system of polynomial inequalities $\cA'$ from the \emph{approximate} $O(\log w_{\min}^{-1})$ moments of the mixture with the same properties as $\cA$.
         \item In~\Cref{sec:pp-3} we argue that $\widehat{\mathcal{A}}$ satisfies the conditions of the rounding algorithm analyzed in~\Cref{thm:sos-rounding-2}.
     \end{enumerate}

At the end we will return to the proof of~\Cref{thm:same-cov-subspace-rounding}.

As a preliminary, we prove a property of generalized Hermite polynomials, which is used in our proofs.

\begin{lemma}[Generalized Hermite polynomials]
\label{lem:gen_hermite}
Let $\bm{x} \in \mathbb{R}$ be distributed according to a mixture of $k$ univariate Gaussians $\sum_{i=1}^k w_i N(\mu_i, \sigma^2)$. Let $H_{e_t}(x, s)$ be the $t$-th probabilist's Hermite polynomial homogenized by $s$.\footnote{For example, $H_{e_4}(x) = x^4 - 6x^2 + 3$ and $H_{e_4}(x, s) = x^4 - 6x^2s^2 + 3s^4$.} Then, for $s \geq 0$ such that $s^2 \leq \sigma^2$, $\mathbb{E}H_{e_t}(\bm{x}, s)$ is equal to the $t$-th moment of the mixture of $k$ univariate Gaussians $\sum_{i=1}^k w_i N(\mu_i, \sigma^2 - s^2)$.
\end{lemma}
\begin{proof}
For $s=0$ the conclusion is trivial.

Otherwise, for $\sigma^2 \geq 1$, it is standard that $\mathbb{E} H_{e_t}(\bm{x}, 1)$ is equal to the $t$-th moment of the mixture $\sum_{i=1}^k w_i N(\mu_i, \sigma^2 - 1)$.
Then for $\sigma^2 \geq s^2$ we have that $\mathbb{E} H_{e_t}(\bm{x} /s, 1)$ is equal to the $t$-th moment of the mixture $\sum_{i=1}^k w_i N(\mu_i/s, \sigma^2/s^2 - 1)$, so $s^t \cdot \mathbb{E} H_{e_t}(\bm{x}/s, 1)$ is equal to the $t$-th moment of the mixture $\sum_{i=1}^k w_i N(\mu_i, \sigma^2 - s^2)$.
Finally, the polynomial $H_{e_t}(x, s)$ is equal to the polynomial $s^t \cdot H_{e_t}(x/s, 1)$, so $\mathbb{E} H_{e_t}(\bm{x}, s) = s^t \cdot \mathbb{E} H_{e_t}(\bm{x}/s, 1)$.
\end{proof}

\subsubsection{Identifying Eigenspaces with Large Eigenvalues Using Exact Moments}
\label{sec:pp-1}
\begin{definition}
\label{def:pp-system}
Let $\bm{x} \in \mathbb{R}^d$ be distributed according to the ground truth mixture. We define $\cA(v, \epsilon)$ to be the following system of polynomial inequalities in indeterminate $v \in \R^d$:
\begin{enumerate}
    \item $\norm{v}^2 \leq 1$\,,
    \item $\Paren{\mathbb{E} \langle \bm{x}, v\rangle^{2t} - (2t-1)!!\|v\|^{2t}}^2 \leq \epsilon \|v\|^{4t}$, for all $t=1,\ldots,10 \log w_{\min}^{-1}$\,.
\end{enumerate}
\end{definition}
For ease of notation, we write $\cA$ without $v$ or $\epsilon$ when understood from context.

We show now that $\cA(v)$ implies that the variance in direction $v$ is very close to $1$.

\begin{lemma}[Identifiability of large-variance directions]
\label{lem:pp-large-directions}
Suppose the mixture is in exact isotropic position.
Then
\[\mathcal{A}(v, \epsilon) \proves{v}{O(\log w_{\min}^{-1})} \Set{v^\top (I_d - \Sigma) v \leq O((\log w_{\min}^{-1}) \epsilon^{1/(20 \log w_{\min}^{-1})})}\,.\]
Furthermore, each monomial $r(v)^2 \prod_{i \in S} p_i(v)$ in the sum-of-squares proof has $|S| \leq O(1)$.
\end{lemma}
\begin{proof}
Projecting $\bm{x}$ onto $v$ produces a mixture in which each component has variance $v^\top \Sigma v$.
Therefore by~\Cref{lem:gen_hermite} we have that $\mathbb{E}H_{e_{2t}}(\langle \bm{x}, v\rangle, \sqrt{v^\top \Sigma v})$ is equal to the $2t$-th moment of the mixture produced by projecting $\bm{x}$ on $v$ but in which each component has variance $0$.
Then, using that $\mathbb{E} H_{e_{2t}}(\langle \bm{x}, v\rangle, \sqrt{v^\top \Sigma v})$ is a degree-$2t$ polynomial in $v$, we have for all $t \leq 10 \log w_{\min}^{-1}$ the polynomial identity
\[\proves{v}{O(\log w_{\min}^{-1})} \mathbb{E} H_{e_{2t}}(\langle \bm{x}, v\rangle, \sqrt{v^\top \Sigma v}) = \sum_{i=1}^k w_i \langle \mu_i, v\rangle^{2t}\,.\]
We aim now to also obtain another approximation for $\mathbb{E} H_{e_{2t}}(\langle \bm{x}, v\rangle, \sqrt{v^\top \Sigma v})$ given the constraints in $\cA$.
If $\mathbb{E} \langle \bm{x}, v\rangle^{2q}$ were equal to $(2q-1)!!\|v\|^{2q}$ for each $0 \leq q \leq t$, then the mixture obtained by projecting $\bm{x}$ on $v$ would match the first $t$ moments of $N(0, \norm{v}^2)$, so by~\Cref{lem:gen_hermite} we would also have that $\mathbb{E} H_{e_{2t}}(\langle \bm{x}, v\rangle, \sqrt{v^\top \Sigma v}) = (2t-1)!!(\norm{v}^2 - v^\top \Sigma v)^t$.
We have instead the following approximate version: $\cA \proves{v}{O(\log w_{\min}^{-1})} (\mathbb{E} \langle \bm{x}, v\rangle^{2q} - (2q-1)!!\|v\|^{2q})^2 \leq \epsilon \|v\|^{4q}$.
Therefore, using that $\mathbb{E} H_{e_{2t}}(\langle \bm{x}, v\rangle, \sqrt{v^\top \Sigma v})$ can be expanded as a sum of terms $C_{t,q} \mathbb{E} \langle \bm{x}, v\rangle^{2q} (v^\top \Sigma v)^{t-q}$ with ${-(2t)^t \leq C_{t,q} \leq (2t)^t}$ and that $0 \preceq \Sigma \preceq I_d$, we can bound by~\Cref{fact:almost_triangle} and~\Cref{fact:spectral_norm_bound}
\begin{align*}
\mathcal{A}
\proves{v}{O(\log w_{\min}^{-1})} &\Paren{ \mathbb{E} H_{e_{2t}}(\langle \bm{x}, v\rangle, \sqrt{v^\top \Sigma v}) - (2t-1)!!(\norm{v}^2-v^\top \Sigma v)^{t} }^2\\
&= \Paren{\sum_{q=0}^t \Paren{C_{t,q} \mathbb{E} \langle \bm{x}, v\rangle^{2q} (v^\top \Sigma v)^{t-q} - C_{t,q} (2q-1)!! \norm{v}^{2q} (v^\top \Sigma v)^{t-q}}}^2\\
&\leq \epsilon (t+1) (2t)^{2t} \norm{v}^{4t} \leq \epsilon (4t)^{2t} \norm{v}^{4t} \leq \epsilon (4t)^{2t}\,.
\end{align*}
Putting these two observations about $\mathbb{E} H_{e_{2t}}(\langle \bm{x}, v\rangle, \sqrt{v^\top \Sigma v})$ together, we obtain that for all $t \leq 10 \log w_{\min}^{-1}$
\[\mathcal{A} \proves{v}{O(\log w_{\min}^{-1})} \Paren{ \sum_{i=1}^k w_i \langle \mu_i, v\rangle^{2t} - (2t-1)!!(\norm{v}^2-v^\top \Sigma v)^{t} }^2 \leq \epsilon (4t)^{2t}\,,\]
so by~\Cref{fact:square-root}
\[\mathcal{A} \proves{v}{O(\log w_{\min}^{-1})} \Abs{ \sum_{i=1}^k w_i \langle \mu_i, v\rangle^{2t} - (2t-1)!!(\norm{v}^2-v^\top \Sigma v)^{t} } \leq \sqrt{\epsilon} (4t)^{t}\,.\]
In particular, as an upper bound we get 
\begin{align*}
\mathcal{A} \proves{v}{O(\log w_{\min}^{-1})} \sum_{i=1}^k w_i \langle \mu_i, v\rangle^{2 t} 
&\leq (2t - 1)!! (\norm{v}^2-v^\top \Sigma v)^{t} + \sqrt{\epsilon} (4t)^t\\
&\leq (2t)^{t} (\norm{v}^2-v^\top \Sigma v)^{t}  + \sqrt{\epsilon} (4t)^t\,,
\end{align*}
which raised to the power $10$ gives for $t \leq \log w_{\min}^{-1}$ that
\begin{equation}
\label{eq:ub}
\mathcal{A} \proves{v}{O(\log w_{\min}^{-1})} \Paren{\sum_{i=1}^k w_i \langle \mu_i, v\rangle^{2 t}}^{10} \leq (4t)^{10t} (\norm{v}^2-v^\top \Sigma v)^{10 t}  + \epsilon^{5} (8t)^{10t} \,.
\end{equation}
As a lower bound, we get for $t \leq \log w_{\min}^{-1}$ that
\begin{align}
\label{eq:lb}
\mathcal{A} \proves{v}{O(\log w_{\min}^{-1})} \sum_{i=1}^k w_i \langle \mu_i, v\rangle^{20 t} 
&\geq (20t - 1)!!(\norm{v}^2-v^\top \Sigma v)^{10 t}  - \sqrt{\epsilon} (40t)^{10t}\nonumber\\
&\geq (10 t)^{10 t}(\norm{v}^2-v^\top \Sigma v)^{10 t}  - \sqrt{\epsilon} (40t)^{10t}\,.
\end{align}
Putting the upper bound and the lower bound together, and using~\Cref{fact:almost_triangle}, we get
\begin{align*}
\mathcal{A} \proves{v}{O(\log w_{\min}^{-1})} \sum_{i=1}^k w_i \langle \mu_i, v\rangle^{20 t}
&\leq w_{\min}^{-9} \Paren{\sum_{i=1}^k w_i \langle \mu_i, v\rangle^{2 t}}^{10}\\
&\leq w_{\min}^{-9} (4t)^{10t} (\norm{v}^2-v^\top \Sigma v)^{10 t}  + w_{\min}^{-9} \epsilon^{5} (8t)^{10t}\\
&\leq w_{\min}^{-9} \frac{(4t)^{10t}}{(10t)^{10t}} \Paren{\sum_{i=1}^k w_i \langle \mu_i, v\rangle^{20 t} + \sqrt{\epsilon} (40t)^{10t}}  + w_{\min}^{-9}  \epsilon^{5} (8t)^{10t}\\
&\leq w_{\min}^{-9}\Paren{\frac{4}{10}}^{10 t} \sum_{i=1}^k w_i \langle \mu_i, v\rangle^{20 t} + w_{\min}^{-9}  \sqrt{\epsilon} (32t)^{10t}\,,
\end{align*}
where the second inequality uses Equation~\ref{eq:ub} and the third inequality uses Equation~\ref{eq:lb}.
Finally, using that $w_{\min}^{-9} \cdot (4/10)^{20 t} \ll 1$ for $t = \log w_{\min}^{-1}$, we get for $t = \log w_{\min}^{-1}$ that 
\[\mathcal{A} \proves{v}{O(\log w_{\min}^{-1})} \sum_{i=1}^k w_i \langle \mu_i, v\rangle^{20 t} \leq 2 w_{\min}^{-9} \cdot \sqrt{\epsilon}\cdot (32t)^{10t}\,.\]
Then by~\Cref{fact:jensen-pow2} also
\[\mathcal{A} \proves{v}{O(\log w_{\min}^{-1})} \Paren{\sum_{i=1}^k w_i \langle \mu_i, v\rangle^{2}}^{10t} \leq 2 w_{\min}^{-9} \cdot \sqrt{\epsilon}\cdot (32t)^{10t}\,,\]
so by~\Cref{fact:square-root}
\[\mathcal{A} \proves{v}{O(\log w_{\min}^{-1})} \sum_{i=1}^k w_i \langle \mu_i, v\rangle^{2} \leq O(t \epsilon^{1/(20t)}) \leq O((\log w_{\min}^{-1}) \epsilon^{1/ (20 \log w_{\min}^{-1})})\,.\]
Finally, because the exact isotropic position ensures $\sum_{i=1}^k w_i \mu_i \mu_i^\top + \Sigma = I_d$,
\[\mathcal{A} \proves{v}{O(\log w_{\min}^{-1})} v^\top (I_d - \Sigma) v = \sum_{i=1}^k w_i \langle \mu_i, v\rangle^2 \leq O((\log w_{\min}^{-1}) \epsilon^{1/ (20 \log w_{\min}^{-1})})\,.\]
\end{proof}

Next, we prove an easy corollary of~\Cref{lem:pp-large-directions}, showing that $\cA(v)$ implies that $v$ is close to the subspace of eigenvectors of eigenvalue $\approx 1$ of $\Sigma$.

\begin{lemma}[Identifiability of approximate eigenspaces]
\label{lem:pp-subspace}
Suppose the mixture is in exact isotropic position.
Let $P$ be the orthogonal projection to the subspace of eigenvectors of $\Sigma$ whose eigenvalues are at most $1 - \delta$. Then
\begin{align*}
    \mathcal{A}(v, \epsilon) \proves{v}{O(1)} \{\|P v\|^2 \leq O((\log w_{\min}^{-1}) \epsilon^{1/(20 \log w_{\min}^{-1})} \delta^{-1})\}.
\end{align*}
Furthermore, each monomial $r(v)^2 \prod_{i \in S} p_i(v)$ in the sum-of-squares proof has $|S| \leq O(1)$.
\end{lemma}
\begin{proof}
Consider the eigenvalue decomposition $\Sigma = \sum_{j=1}^d \lambda_j s_j s_j^\top$. By~\Cref{lem:pp-large-directions}, we have that $\cA \proves{v}{O(\log w_{\min}^{-1})} v^\top (I_d - \Sigma) v \leq \epsilon'$ where $\epsilon' = O((\log w_{\min}^{-1})\epsilon^{1/(20 \log w_{\min}^{-1})})$. Then
\[ \mathcal{A} \proves{v}{O(\log w_{\min}^{-1})} \sum_{j=1}^d (1-\lambda_j) \langle s_j, v\rangle^2 \leq \epsilon' \,. \]
Therefore we also have
\[ \mathcal{A} \proves{v}{O(\log w_{\min}^{-1})} \delta \sum_{j: \lambda_j \leq 1 - \delta} \langle s_j, v\rangle^2 \leq \epsilon'\,,\]
so
\[ \mathcal{A} \proves{v}{O(\log w_{\min}^{-1})} v^\top \Paren{ \sum_{j: \lambda_j \leq 1 - \delta} s_j s_j^{\top} } \Paren{ \sum_{j: \lambda_j \leq 1 - \delta} s_j s_j^{\top}} v = \sum_{j: \lambda_j \leq 1 - \delta} \langle s_j, v\rangle^2 \leq \epsilon' / \delta\,,\]
so
\[ \mathcal{A} \proves{v}{O(\log w_{\min}^{-1})} \norm{Pv}^2 \leq \epsilon' / \delta\,.\]
\end{proof}

We also prove the converse: if $v$ is in the subspace of eigenvectors of eigenvalue $\approx 1$ of $\Sigma$, then $v$ satisfies $\cA(v)$.

\begin{lemma}[Sum-of-squares feasibility]
\label{lem:pp-feasibility}
Suppose the mixture is in exact isotropic position.
Let $P$ be the orthogonal projection to the subspace of eigenvectors of $\Sigma$ whose eigenvalues are at most $1 - \delta$.
Then
\[\Set{\norm{v}^2 \leq 1, Pv = 0} \proves{v}{O(\log w_{\min}^{-1})} \cA\Paren{v, O\Paren{\Paren{w_{\min}^{-1}}^{40 \log w_{\min}^{-1}} \delta^{20 \log w_{\min}^{-1}}}}\,.\]
Furthermore, each monomial $r(v)^2 \prod_{i \in S} p_i(v)$ in the sum-of-squares proof has $|S| \leq O(1)$.
\end{lemma}
\begin{proof}
For ease of notation, let $\cC = \Set{\norm{v}^2 \leq 1, Pv=0}$.
We prove that $v$ satisfies the conditions of $\cA$.
First, we have trivially that $\cC \proves{v}{2} \Set{\norm{v}^2 \leq 1}$.

Second, note that $\cC \proves{v}{2} \Set{(I_d-P)v =v}$. Consider the eigenvalue decomposition $\Sigma = \sum_{j=1}^d \lambda_j s_j s_j^\top$. Then 
\[\cC \proves{v}{2} v^\top \Sigma v \geq (1-\delta) \sum_{j: \lambda_j > 1-\delta} \langle s_j, v\rangle^2 = (1-\delta) \norm{(I_d-P)v}^2 = (1-\delta)\norm{v}^2\,.\]
Then, using that $\sum_{i=1}^k w_i \langle \mu_i, v\rangle^2 + v^\top \Sigma v = \norm{v}^2$, we also get
\[\cC \proves{v}{2} \sum_{i=1}^k w_i \langle \mu_i, v\rangle^2 \leq \delta \norm{v}^2\,,\]
so $\cC \proves{v}{2} \langle \mu_i, v\rangle^2 \leq w_{\min}^{-1} \delta \norm{v}^2$ for all $i \in [k]$.
Then, for all $t=1,\ldots, 10 \log w_{\min}^{-1}$, using also that $0 \preceq \Sigma \preceq I_d$, we have by~\Cref{fact:spectral_norm_bound} the upper bound
\begin{align*}
\cC \proves{v}{O(\log w_{\min}^{-1})} \mathbb{E}\langle \bm{x}, v\rangle^{2t}
&= \sum_{i=1}^k w_i \sum_{s=0}^t \binom{2t}{2s} \langle \mu_i, v\rangle^{2s} (v^\top \Sigma v)^{t-s} (2t-2s-1)!!\\
&\leq \sum_{i=1}^k w_i \sum_{s=0}^t \binom{2t}{2s} \langle \mu_i, v\rangle^{2s} \norm{v}^{2t-2s} (2t-2s-1)!!\\
&= (2t-1)!! \norm{v}^{2t} + \sum_{i=1}^k w_i \sum_{s=1}^t \binom{2t}{2s} \langle \mu_i, v\rangle^{2s} \norm{v}^{2t-2s} (2t-2s-1)!!\\
&\leq (2t-1)!! \norm{v}^{2t} + \sum_{s=1}^t \binom{2t}{2s} w_{\min}^{-s} \delta^s \norm{v}^{2t} (2t-2s-1)!!\\
&\leq (2t-1)!! \norm{v}^{2t} + (4t)^t w_{\min}^{-t} \delta \norm{v}^{2t}
\end{align*}
and the lower bound
\begin{align*}
\cC \proves{v}{O(\log w_{\min}^{-1})} \mathbb{E}\langle \bm{x}, v\rangle^{2t}
&= \sum_{i=1}^k w_i \sum_{s=0}^t \binom{2t}{2s} \langle \mu_i, v\rangle^{2s} (v^\top \Sigma v)^{t-s} (2t-2s-1)!!\\
&\geq (2t-1)!! (v^\top \Sigma v)^t\\
&\geq (2t-1)!! (1-\delta)^{t} \norm{v}^{2t}\\
&\geq (2t-1)!! \norm{v}^{2t} - (2t)^{t+1} \delta \norm{v}^{2t}\,.
\end{align*}
Then
\[\cC \proves{v}{O(\log w_{\min}^{-1})} \Paren{\mathbb{E}\langle \bm{x}, v\rangle^{2t} - (2t-1)!! \norm{v}^{2t}}^2 \leq (4t)^{2t} w_{\min}^{-2t} \delta^2 \norm{v}^{4t} \leq w_{\min}^{-4t} \delta^2 \norm{v}^{4t}\,.\]

\end{proof}

\subsubsection{Identifying Eigenspaces with Large Eigenvalues Using Approximate Moments}
\label{sec:pp-2}

In this section we show how to construct a system of polynomial constraints $\widehat{\cA}$ with properties similar to $\cA$ when only having access to approximate moments. Before doing so, we state a result on the closeness of the empirical moments to the population moments.

\begin{theorem}[Theorem 1.3 and Lemma 5.4 in~\cite{kothari2017outlier}]
\label{thm:pp-sos-mom-est}
Given $\poly(d^t, w_{\min}^{-t}, \epsilon^{-1})$ samples from the mixture with an $\epsilon$-fraction of corruptions, where $\epsilon \leq \Omega(w_{\min}^2/t^2)$, there exists an algorithm that runs in time $\poly(d^t, w_{\min}^{-t}, \epsilon^{-1})$ and outputs symmetric tensor moment estimates $\widehat{M}_2 \in \R^{d^2}, \ldots, \widehat{M}_{2t} \in \R^{d^{2t}}$ such that with high probability, for all vectors $v \in \R^d$ and all $1 \leq q \leq t$,
\[\langle M_{2q} - \widehat{M}_{2q}, v^{\otimes 2q}\rangle^2 \leq O(t^q w_{\min}^{-q} \sqrt{\epsilon}) \cdot \Iprod{M_2, v^{\otimes 2}}^q\,.\]
Furthermore, there exist degree-$O(t)$ sum-of-squares proofs in $v$ of these inequalities.
\end{theorem}

We obtain the following simple corollary for mixtures in $\gamma$-approximate isotropic position:

\begin{corollary}
\label{cor:pp-sos-mom-est-iso}
Suppose the mixture is in $\gamma$-approximate isotropic position with $\gamma \leq \Omega(1/q)$. Then the same result as in~\Cref{thm:pp-sos-mom-est} holds with upper bounds
\[\langle M_{2q} - \widehat{M}_{2q}, v^{\otimes 2q}\rangle^2 \leq O(t^q w_{\min}^{-q} \sqrt{\epsilon}) \norm{v}^{2q}\,.\]
\end{corollary}
\begin{proof}
See the proof of~\Cref{cor:sos-mom-est-iso}.
\end{proof}

In~\Cref{lem:pp-zero-mean-feasibility} we prove that we can construct a system of constraints $\widehat\cA(v)$ that is roughly equivalent to $\cA(v)$.

\begin{lemma}[Approximate moment feasibility]
\label{lem:pp-zero-mean-feasibility}
Suppose the mixture is in $\gamma$-approximate isotropic position, with $\gamma \leq \Omega(1/q)$.
Given $\poly(d^{\log w_{\min}^{-1}}, (w_{\min}^{-1})^{\log w_{\min}^{-1}}, \epsilon^{-1})$ samples from the mixture with an $\epsilon$-fraction of errors, where $\epsilon \leq \Omega(w_{\min}^{20 \log w_{\min}^{-1}})$, there exists an algorithm that runs in time \linebreak $\poly(d^{\log w_{\min}^{-1}}, (w_{\min}^{-1})^{\log w_{\min}^{-1}}, \epsilon^{-1})$ and computes a system of polynomial inequalities $\widehat\cA(v, \epsilon)$ of size $\poly(d^{\log w_{\min}^{-1}})$ in indeterminate $v$ such that with high probability
\[\widehat\cA(v, \epsilon) \proves{v}{O(\log w_{\min}^{-1})} \cA\Paren{v, O\Paren{(w_{\min}^{-1})^{10 \log w_{\min}^{-1}} \sqrt{\epsilon}}}\]
and 
\[\cA(v, \epsilon) \proves{v}{O(\log w_{\min}^{-1})} \widehat\cA\Paren{v, O\Paren{(w_{\min}^{-1})^{10 \log w_{\min}^{-1}} \sqrt{\epsilon}}}\,.\]
Furthermore, each monomial $r(v)^2 \prod_{i \in S} p_i(v)$ in the sum-of-squares proofs has $|S| \leq O(1)$.
\end{lemma}
\begin{proof}
We compute by~\Cref{cor:pp-sos-mom-est-iso} in time $\poly(d^{\log w_{\min}^{-1}}, (w_{\min}^{-1})^{\log w_{\min}^{-1}}, \epsilon^{-1})$ some $\widehat{M}_2 \in \R^{d^2}$, $\ldots$, $\widehat{M}_{2t} \in \R^{d^{2t}}$ with $t = 10 \log w_{\min}^{-1}$ such that for all $1 \leq q \leq t$
\[\proves{v}{O(\log w_{\min}^{-1})} \langle M_{2q} - \widehat{M}_{2q}, v^{\otimes 2q}\rangle^2 \leq O(t^q w_{\min}^{-q} \sqrt{\epsilon}) \norm{v}^{4q}\,.\]
Let $\epsilon' = O(t^t w_{\min}^{-t} \sqrt{\epsilon})$.
Then we construct the following system of polynomial inequalities $\widehat\cA(v, \epsilon)$ in indeterminate $v \in \R^d$: 
\begin{enumerate}
    \item $\norm{v}^2 \leq 1$\,,
    \item $\Paren{\langle \widehat{M}_{2t}, v^{\otimes 2t}\rangle - (2t-1)!!\|v\|^{2t}}^2 \leq \epsilon' \|v\|^{4t}$, for all $t=1,\ldots,10 \log w_{\min}^{-1}$\,.
\end{enumerate}
We start by proving that $\widehat\cA(v, \epsilon) \proves{v}{O(\log w_{\min}^{-1})} \cA(v, O(\epsilon'))$. We have trivially that ${\widehat\cA(v, \epsilon) \proves{v}{2} \norm{v}^2 \leq 1}$. For the second constraint, we have by~\Cref{fact:almost_triangle}
\begin{align*}
\widehat\cA(v, \epsilon) \proves{v}{O(\log w_{\min}^{-1})}
&\Paren{\langle M_{2t}, v^{\otimes 2t}\rangle - \|v\|^{2t}(2t-1)!!}^2\\
&\leq 2 \Paren{\langle \widehat{M}_{2t}, v^{\otimes 2t}\rangle - \|v\|^{2t}(2t-1)!!}^2 + 2 \langle M_{2t} - \widehat{M}_{2t}, v^{\otimes 2t}\rangle^2\\
&\leq O(\epsilon') \norm{v}^{4t}\,.
\end{align*}
Therefore $\widehat\cA(v, \epsilon) \proves{v}{O(\log w_{\min}^{-1})} \cA(v, O(\epsilon'))$.

A similar set of inequalities gives the other direction.
\end{proof}

Finally, we show that even if the mixture is only in $\gamma$-approximate isotropic position, $\cA(v)$ is still roughly equivalent to the system of constraints that we would have if the mixture were in exact isotropic position.

\begin{lemma}[Approximate isotropic position]
\label{lem:pp-approx-isotropic}
Suppose the mixture is in $\gamma$-approximate isotropic position, with $\gamma \leq \Omega(w_{\min}^{20 \log w_{\min}^{-1}})$.
Let $\tilde\mu = \sum_{i=1}^k w_i \mu_i$ and ${\tilde\Sigma = \sum_{i=1}^k w_i \mu_i \mu_i^\top + \Sigma}$, and let $\tilde{\cA}(v, \epsilon)$ be the system of polynomial inequalities in indeterminate $v \in \R^d$ in~\Cref{def:pp-system} for the mixture $\tilde\Sigma^{-1/2} (\bm{x} - \tilde\mu)$ where $\bm{x}$ is distributed according to the ground truth mixture, that is, for the ground truth mixture put into exact isotropic position.
Then
\[\cA(v, \epsilon) \proves{O(\log w_{\min}^{-1})}{v} \tilde{\cA}\Paren{v, O\Paren{\epsilon + \Paren{w_{\min}^{-1}}^{40 \log w_{\min}^{-1}} \gamma^{2}}}\]
and 
\[\tilde\cA(v, \epsilon) \proves{O(\log w_{\min}^{-1})}{v} \cA\Paren{v, O\Paren{\epsilon + \Paren{w_{\min}^{-1}}^{40 \log w_{\min}^{-1}} \gamma^{2}}}\,.\]
Furthermore, each monomial $r(v)^2 \prod_{i \in S} p_i(v)$ in the sum-of-squares proofs has $|S| \leq O(1)$.
\end{lemma}
\begin{proof}
We start by proving that $\cA(v, \epsilon) \proves{O(\log w_{\min}^{-1})}{v} \tilde{\cA}\Paren{v, O\Paren{\epsilon + 2^{O(t)} w_{\min}^{-2t} \gamma^{2}}}$.
We have trivially that $\cA(v, \epsilon) \proves{v}{2} \norm{v}^2 \leq 1$.

For the second constraint, we have for all $t=1,\ldots, 10\log w_{\min}^{-1}$ by~\Cref{fact:almost_triangle} and~\Cref{fact:cauchy-schwarz}
\begin{align*}
\cA(v, \epsilon)
&\proves{v}{O(\log w_{\min}^{-1})} \Paren{\E \langle \tilde\Sigma^{-1/2} (\bm x - \tilde\mu), v\rangle^{2t} - (2t-1)!!\norm{v}^{2t}}^2\\
&\leq 2 \Paren{\E \langle \bm x, v\rangle^{2t} - (2t-1)!!\norm{v}^{2t}}^2 + 2 \Paren{\E \langle \tilde\Sigma^{-1/2} (\bm x - \tilde\mu), v\rangle^{2t} - \E \langle \bm x, v\rangle^{2t}}^2\\
&\leq 2\epsilon \norm{v}^{4t} + 2 \Paren{\E \langle \tilde\Sigma^{-1/2} (\bm x - \tilde\mu), v\rangle^{2t} - \E \langle \bm x, v\rangle^{2t}}^2\\
&= 2\epsilon \norm{v}^{4t} + 2 \Paren{\E (\langle \bm x, v\rangle + \langle (\tilde\Sigma^{-1/2}-I_d)\bm x - \tilde\Sigma^{-1/2} \tilde\mu, v\rangle)^{2t} - \E \langle \bm x, v\rangle^{2t}}^2\\
&\leq 2\epsilon \norm{v}^{4t} + 2^{O(t)} \cdot O\Paren{\sum_{\ell=0}^{2t-1} \Paren{\E \langle \bm{x}, v\rangle^{\ell} \langle (\tilde\Sigma^{-1/2}-I_d)\bm x - \tilde\Sigma^{-1/2} \tilde\mu, v\rangle^{2t-\ell}}^2}\\
&\leq 2\epsilon \norm{v}^{4t} + 2^{O(t)} \cdot O\Paren{\sum_{\ell=0}^{2t-1} \E \langle \bm{x}, v\rangle^{2\ell} \E \langle (\tilde\Sigma^{-1/2}-I_d)\bm x - \tilde\Sigma^{-1/2} \tilde\mu, v\rangle^{4t-2\ell}}\,,
\end{align*}
where in the last line we used~\Cref{fact:expect-cauchy-schwarz}.
Because of the $\gamma$-approximate isotropic position we have $\norm{\mu_i} \leq w_{\min}^{-1/2}$ for all $i \in [k]$ and $\norm{\Sigma} \leq O(1)$, and then we can bound for any $\ell$ using~\Cref{fact:spectral_norm_bound}
\[\proves{v}{O(\log w_{\min}^{-1})} \E \langle \bm x, v\rangle^{2\ell} \leq 2^{O(\ell)} \sum_{i=1}^k w_i \langle \mu_i, v\rangle^{2\ell} + \ell^{O(\ell)} (v^\top \Sigma v)^{\ell} \leq \ell^{O(\ell)} w_{\min}^{-\ell} \norm{v}^{2\ell}\]
and similarly, using that $\norm{\tilde\Sigma^{-1/2}-I_d} \leq O(\gamma)$ and $\norm{\tilde\mu} \leq O(\gamma)$,
\begin{align*}
\proves{v}{O(\log w_{\min}^{-1})} \E\langle (\tilde\Sigma^{-1/2}-I_d)\bm x - \tilde\Sigma^{-1/2} \tilde\mu, v\rangle^{2\ell}
&\leq 2^{O(\ell)} \E\langle (\tilde\Sigma^{-1/2}-I_d)\bm x, v\rangle^{2\ell} + 2^{O(\ell)} \gamma^{2\ell}\\
&\leq \ell^{O(\ell)} w_{\min}^{-\ell} \gamma^2 \norm{v}^{2\ell}\,.
\end{align*}

Therefore we get
\begin{align*}
\cA(v, \epsilon)
&\proves{v}{O(\log w_{\min}^{-1})} \Paren{\E \langle \tilde\Sigma^{-1/2} (\bm x - \tilde\mu), v\rangle^{2t} - \norm{v}^{2t} (2t-1)!!}^2\\
&\leq 2\epsilon \norm{v}^{4t} + t^{O(t)} w_{\min}^{-2t} \gamma^{2} \norm{v}^{4t}\\
&\leq 2\epsilon \norm{v}^{4t} + w_{\min}^{-4t} \gamma^{2} \norm{v}^{4t}\,.
\end{align*}

A similar set of inequalities gives the other direction.
\end{proof}

\subsubsection{Final Constraint Construction}
\label{sec:pp-3}
In this section, we show that we can construct $\widehat{\cA}$ that satisfies the conditions of the rounding algorithm analyzed in~\Cref{thm:sos-rounding-2}.

\begin{lemma}[Identical-covariance rounding conditions]
\label{lem:pp-complement}
Suppose the mixture is in $\gamma$-approximate isotropic position with $\gamma \leq \Omega(1/q)$.
Let $\cA(v, \epsilon)$ be the system of polynomial inequalities for the mixture.
Let $\tilde\mu = \sum_{i=1}^k w_i \mu_i$ and ${\tilde\Sigma = \sum_{i=1}^k w_i \mu_i \mu_i^\top + \Sigma}$, and let $\tilde{\cA}(v, \epsilon)$ be the system of polynomial inequalities in indeterminate $v \in \R^d$ in~\Cref{def:pp-system} for the mixture $\tilde\Sigma^{-1/2} (\bm{x} - \tilde\mu)$ where $\bm{x}$ is distributed according to the ground truth mixture, that is, for the ground truth mixture put into exact isotropic position.

    Let $P$ be the orthogonal projection to the subspace of eigenvectors of $\tilde\Sigma^{-1/2} \Sigma \tilde\Sigma^{-1/2}$ whose eigenvalues are at most $1-\delta^{1/16}$.
    Also let $R$ be the orthogonal projection to the subspace of eigenvectors of $\tilde\Sigma^{-1/2} \Sigma \tilde\Sigma^{-1/2}$ whose eigenvalues are at most $1-\delta$.

    Suppose $\epsilon \leq \Omega(w_{\min}^{160 \log w_{\min}^{-1}})$, $\gamma \leq \Omega(w_{\min}^{40 \log w_{\min}^{-1}} \epsilon)$, and $\delta \sim \Omega(w_{\min}^3 \epsilon^{1/(10 \log w_{\min}^{-1})})$ small enough.
    Given $\poly(d^{\log w_{\min}^{-1}}, (w_{\min}^{-1})^{\log w_{\min}^{-1}}, \epsilon^{-1})$ samples from the mixture with an $\epsilon$-fraction of corruptions, there exists an algorithm that runs in time $\poly(d^{\log w_{\min}^{-1}}, (w_{\min}^{-1})^{\log w_{\min}^{-1}}, \epsilon^{-1})$ and computes a set of polynomial inequalities $\widehat\cA$ of size $\poly(d^{\log w_{\min}^{-1}})$ in indeterminate $v \in \R^d$, including $\norm{v}^2 \leq 1$, such that with high probability:
    \begin{enumerate}
        \item $\widehat\cA \proves{v}{O(\log w_{\min}^{-1})} \Set{\Norm{Pv}^2 \leq O(w_{\min}^{-1} \epsilon^{1/(80\log w_{\min}^{-1})})}$,
        \item $\Set{\norm{v}^2 \leq 1, Rv = 0} \proves{v}{O(\log w_{\min}^{-1})} \widehat\cA$.
    \end{enumerate}
    Furthermore, each monomial $r(v)^2 \prod_{i \in S} p_i(v)$ in the sum-of-squares proofs has $|S| \leq O(1)$.
\end{lemma}
\begin{proof}
    By~\Cref{lem:pp-zero-mean-feasibility} and~\Cref{lem:pp-approx-isotropic} we can construct a system of polynomial inequalities $\widehat{\cA}(v, \epsilon)$ in time $\poly(d^{\log w_{\min}^{-1}}, (w_{\min}^{-1})^{\log w_{\min}^{-1}}, \epsilon^{-1})$ such that 
    \begin{align*}
    \widehat{\cA}(v, \epsilon)
    &\proves{v}{O(\log w_{\min}^{-1})} \cA\Paren{v, O\Paren{(w_{\min}^{-1})^{10 \log w_{\min}^{-1}} \sqrt{\epsilon}}}\\
    &\proves{v}{O(\log w_{\min}^{-1})} \tilde{\cA}\Paren{v, O\Paren{(w_{\min}^{-1})^{10 \log w_{\min}^{-1}} \sqrt{\epsilon} + \Paren{w_{\min}^{-1}}^{40 \log w_{\min}^{-1}} \gamma^{2}}}\,.
    \end{align*}
    Let $\epsilon' = O\Paren{(w_{\min}^{-1})^{10 \log w_{\min}^{-1}} \sqrt{\epsilon} + \Paren{w_{\min}^{-1}}^{40 \log w_{\min}^{-1}} \gamma^{2}}$.
    Then by~\Cref{lem:pp-subspace}
    \[\widehat{\cA}(v, \epsilon) \proves{v}{O(\log w_{\min}^{-1})} \Norm{Pv}^2 \leq O((\log w_{\min}^{-1}) (\epsilon')^{1/(20\log w_{\min}^{-1})} \delta^{-1/16})\,,\]
    where using that $\delta^{1/16} \geq (\epsilon')^{1/(40\log w_{\min}^{-1})}$ we have the bound
    \[\widehat{\cA}(v, \epsilon) \proves{v}{O(\log w_{\min}^{-1})} \Norm{Pv}^2 \leq O((\log w_{\min}^{-1}) (\epsilon')^{1/(40\log w_{\min}^{-1})})\,,\]
    where using the bounds on $\epsilon$ and $\gamma$ we get the bound
    \[\widehat{\cA}(v, \epsilon) \proves{v}{O(\log w_{\min}^{-1})} \Norm{Pv}^2 \leq O(w_{\min}^{-1} \epsilon^{1/(80\log w_{\min}^{-1})})\,.\]

    For the second claim, we start by noting that, by~\Cref{lem:pp-feasibility}, there exists a sum-of-squares proof in indeterminate $v \in \R^d$ that, if $\norm{v}^2 \leq 1$ and $Rv=0$, then $v$ is feasible for $\tilde\cA(v, O((w_{\min}^{-1})^{40 \log w_{\min}^{-1}} \delta^{20 \log w_{\min}^{-1}}))$.
    Then by~\Cref{lem:pp-zero-mean-feasibility} and~\Cref{lem:pp-approx-isotropic}
    \begin{align*}
    &\tilde\cA\Paren{v, O\Paren{(w_{\min}^{-1})^{40 \log w_{\min}^{-1}} \delta^{20 \log w_{\min}^{-1}}}}\\
    &\quad \proves{v}{O(\log w_{\min}^{-1})} \cA\Paren{v, O\Paren{(w_{\min}^{-1})^{40 \log w_{\min}^{-1}} \delta^{20 \log w_{\min}^{-1}} + \Paren{w_{\min}^{-1}}^{40 \log w_{\min}^{-1}} \gamma^{2}}}\\
    &\quad \proves{v}{O(\log w_{\min}^{-1})} \hat\cA\Paren{v, O\Paren{\Paren{w_{\min}^{-1}}^{10 \log w_{\min}^{-1}} \Paren{(w_{\min}^{-1})^{40 \log w_{\min}^{-1}} \delta^{20 \log w_{\min}^{-1}} + \Paren{w_{\min}^{-1}}^{40 \log w_{\min}^{-1}} \gamma^{2}}^{1/2}}}\,.
    \end{align*}
    When $\delta \leq \Omega(w_{\min}^{3} \epsilon^{1/(10 \log w_{\min}^{-1})})$
    and $\gamma \leq \Omega(w_{\min}^{40 \log w_{\min}^{-1}} \epsilon)$, this implies $\widehat{\cA}(v,\epsilon)$.

    The claim about the monomials in the sum-of-squares proofs follows from composition of sum-of-squares proofs.
\end{proof}

\subsubsection{Proof of~\Cref{thm:same-cov-subspace-rounding}}

Let $\tilde\mu = \sum_{i=1}^k w_i \mu_i$ and $\tilde\Sigma = \sum_{i=1}^k w_i \mu_i \mu_i^\top + \Sigma$.
Let $P$ be the orthogonal projection to the subspace of eigenvectors of $\tilde\Sigma^{-1/2} \Sigma \tilde\Sigma^{-1/2}$ whose eigenvalues are at most $1-\delta^{1/16}$.
Also let $R$ be the orthogonal projection to the subspace of eigenvectors of $\tilde\Sigma^{-1/2} \Sigma \tilde\Sigma^{-1/2}$ whose eigenvalues are at most $1-\delta$, where $\delta \sim w_{\min}^3 \epsilon^{1/(10 \log w_{\min}^{-1})}$.

By~\Cref{lem:pp-complement}, we can compute in time $\poly(d^{\log w_{\min}^{-1}}, (w_{\min}^{-1})^{\log w_{\min}^{-1}}, \epsilon^{-1})$ a system of
polynomial inequalities $\widehat\cA$ of size $\poly(d^{\log w_{\min}^{-1}})$ in indeterminate $v \in \R^d$, including $\norm{v}^2 \leq 1$, such that
\begin{itemize}
    \item $\widehat\cA \proves{v}{O(\log w_{\min}^{-1})} \Set{\Norm{Pv}^2 \leq O(w_{\min}^{-1} \epsilon^{1/(80\log w_{\min}^{-1})})}$,
    \item $\Set{\norm{v}^2 \leq 1, Rv = 0} \proves{v}{O(\log w_{\min}^{-1})} \widehat\cA$,
\end{itemize}
and such that each monomial $r(v)^2 \prod_{i \in S} p_i(v)$ in the sum-of-squares proofs has $|S| \leq O(1)$.

Because of the isotropic position, we have trivially that $\rank(P) \leq k$. 
Then, by applying~\Cref{thm:sos-rounding-2} to $\widehat\cA$ with orthogonal projection matrices $P$ and $R$, we can compute in time $f(w_{\min}^{-1}, \epsilon^{-1}) \cdot \poly(d^{\log^2 w_{\min}^{-1}})$ a $D$-dimensional subspace such that, for every unit vector $v \in \R^d$ such that $\norm{Pv}^2 \geq 1 - \alpha$, there exists a unit vector in the subspace that is $\beta$-close to it, where 
\[D = O(w_{\min}^{1/8} \epsilon^{-1/(640\log w_{\min}^{-1})})^k\,,\]
\[\alpha = \Omega(w_{\min}^{-1} \epsilon^{1/(80\log w_{\min}^{-1})})\,,\]
\[\beta = O(w_{\min}^{-1/8} \epsilon^{1/(640\log w_{\min}^{-1})})\,.\]

We argue now that the unit vectors $v \in \R^d$ with $\norm{Pv}^2 \geq 1 - \alpha$ are the ones we are interested in.
First, we have that all unit vectors $v$ that lie in the subspace of eigenvectors of $\tilde\Sigma^{-1/2}\Sigma\tilde\Sigma^{-1/2}$ with eigenvalues at most $1-\delta^{1/16}$ satisfy $\norm{P v}^2 = 1 \geq 1-\alpha$.

Second, we want a result about $\Sigma$, not $\tilde\Sigma^{-1/2} \Sigma \tilde\Sigma^{-1/2}$.
Note that 
\begin{align*}
|v^\top \Sigma v - v^\top \tilde\Sigma^{-1/2} \Sigma \tilde\Sigma^{-1/2} v|
&= |v^\top (\Sigma - \tilde\Sigma^{-1/2} \Sigma \tilde\Sigma^{-1/2}) v|
\leq \norm{v}^2 \norm{\Sigma - \tilde\Sigma^{-1/2} \Sigma \tilde\Sigma^{-1/2}}\\
&= \norm{\Sigma - (I_d + E) \Sigma (I_d + E)}
\leq 2\norm{E} \norm{\Sigma} + \norm{E}^2 \norm{\Sigma}\\
&\leq O(\gamma)\,,
\end{align*}
where we used that $E = \tilde{\Sigma}^{-1/2} - I_d$ satisfies $\norm{E} \leq O(\gamma)$ and $\norm{\Sigma} \leq 1 + \gamma$.
Let $P'$ be the orthogonal projection to the subspace of eigenvectors of $\Sigma$ with eigenvalues at most $1/2$.
Note that all unit vectors $w \in \R^d$ that lie in the subspace associated with $I_d-P$ satisfy $w^\top \tilde\Sigma^{-1/2} \Sigma \tilde\Sigma^{-1/2} w > 1-\delta^{1/16}$, so also $w^\top \Sigma w > 1-\delta^{1/16}-O(\gamma) \geq 1-O(\delta^{1/16})$.
Then because of the $\gamma$-approximate isotropic position we get for all such unit vectors $w$ that
\[1-O(\delta^{1/16}) \leq w^\top \Sigma w \leq (1/2)\norm{P'w}^2 + (1+\gamma) (1-\norm{P'w}^2)\,,\]
so $\norm{P'w}^2 \leq O(\delta^{1/16})+O(\gamma) \leq O(\delta^{1/16})$.
This implies that each unit vector $v$ with $\norm{P'v}^2=1$ has squared inner product at most $O(\delta^{1/16})$ with any unit vector in $I_d-P$, so it has squared inner product at least $1-O(\delta^{1/16})$ with some unit vector in $P$, so it has distance at most $O(\delta^{1/32})$ from some unit vector in $P$.
Therefore $v$ is $(\beta+O(\delta^{1/32}))$-close to a unit vector in the returned subspace, which is dominated by $O(\beta)$.

$\qed$

\subsection{Algorithm for Clustering Non-Spherical Mixtures with Identical Unknown Covariance}
\label{sec:pp-full-alg}
We now present the clustering algorithm for mixtures whose components have an identical unknown covariance. 

\begin{algorithm}[H]
    \SetAlgoLined
    \SetKwInOut{Input}{input}
    \SetKwInOut{Output}{output}

    \Input{a collection of samples $\mathcal{T}$, the number of components $k$, the minimum mixing weight $w_{\min}$, the separation parameter $\Delta$}
    \Output{a clustering $\mathcal{S}$ of $\mathcal{T}$}

    Partition $\mathcal{T} = \mathcal{T}_1 \cup \mathcal{T}_2$ by putting each sample in $\mathcal{T}_1$ independently with probability $1/2$\footnote{This increases the fraction of outliers in $\mathcal{T}_1$ and $\mathcal{T}_2$ with high probability by at most a constant factor.}\;
    Let $\mathcal{C} = \Set{\mathcal{S}}$ where $\mathcal{S}$ is the trivial partial clustering with one subset $\mathcal{T}_1$\;
    \While{$\exists\, \mathcal{S} \in \mathcal{C}$ such that $\vert \mathcal{S} \vert < k$}{
        Remove $\mathcal{S}$ from $\mathcal{C}$\;
        Let $L$ be the output of the partial clustering refinement algorithm (\Cref{alg:pp-refinement}) on $\mathcal{S}$\;
        Add all $\mathcal{S}' \in L$ to $\mathcal{C}$\;
    }
    Run the clustering selection algorithm (\Cref{claim:list-reduction}) on $\mathcal{C}$ with upper bound $\Delta^{-1/(\log w_{\min}^{-1})^{2k}}$ on the fraction of corruptions and with new samples $\mathcal{T}_2$, and let $\mathcal{S}$ be the returned clustering\;
    \Return{$\mathcal{S}$}\;
\caption{Mixtures of identical-covariance Gaussians learning algorithm}
\label{alg:pp-cluster}
\end{algorithm}

\begin{algorithm}[H]
    \SetAlgoLined
    \SetKwInOut{Input}{input}
    \SetKwInOut{Output}{output}

    \Input{a candidate partial clustering $\mathcal{S}$, the number of components $k$, the minimum mixing weight $w_{\min}$, the separation parameter $\Delta$}
    \Output{a list of candidate refinements of $\mathcal{S}$}
    Let $L = \varnothing$\;
    \For{each $S \in \mathcal{S}$ and each guess of $\epsilon$ with $2^{-\Delta}$-bit complexity that satisfies $\epsilon \geq \Delta^{-1}$} {
        Put the samples of $S$ in isotropic position with the algorithm corresponding to~\Cref{thm:robust-isotropic-position}\;
        Let $Q$ be the subspace returned by the same-covariance subspace recovery algorithm (\Cref{thm:same-cov-subspace-rounding}) run on $S$ with fraction of outliers set to $\epsilon$ (but abort and set $Q=\emptyset$ if the algorithm performs more than $f(\Delta) \cdot \poly(d^{\log^2 w_{\min}^{-1}})$ steps for some function $f$)\;
        If $\dim(Q) > f(\Delta)$ for some function $f$, set $Q=\emptyset$\;
        \For{each unit vector $v$ in a $\Delta^{-1}$-net of $Q$}{
            \For{each $1/(100k)$-resolution choice of $\tau$ in $[-2w_{\min}^{-1/2}, 2w_{\min}^{-1/2}]$}{
                Partition the samples $S = S_1 \cup S_2$ based on whether their projection on $v$ is in the interval $(-\infty, \tau]$ or outside it\;
                Add $(\mathcal{S}\setminus S) \cup \{S_1, S_2\}$ to $L$\;
            }
        }
    }
    \Return{$L$}\;
\caption{Mixtures of identical-covariance Gaussians partial clustering refinement algorithm}
\label{alg:pp-refinement}
\end{algorithm}

\subsubsection{Analysis of Partial Clustering Refinement (\Cref{alg:pp-refinement})}

We now show that the refinement subroutine has the desired properties.

\begin{lemma}[Mean clustering result]
\label{lem:pp-mean-sep}
Let $\mathcal{S}$ be a $(1-\epsilon)$-good partial clustering.
Suppose there exists some $S \in \mathcal{S}$ with $|S| \geq \poly(d, \Delta)$ that contains at least two components, and suppose that the mixture corresponding to components in $\mathsf{comp}(S)$ is in approximate isotropic position: $\Norm{\sum_{i=1}^k w'_i \mu_i} \leq \epsilon^{2}$ and $(1-\epsilon^{2})I_d \preceq \sum_{i \in \mathsf{comp}(S)} w'_i \mu_i\mu_i^\top + \Sigma \preceq (1+\epsilon^{2})I_d$, where $w'_i = \frac{w_i}{\sum_{i \in \mathsf{comp}(S)} w_i}$.
Then for $\Delta^{-1} \leq \epsilon \leq w_{\min}^{O(\log w_{\min}^{-1})}$ the output of~\Cref{alg:pp-refinement} contains with high probability some $\mathcal{S}'$ that is a $(1-\epsilon')$-good partial clustering for $\epsilon' = O(w_{\min}^{-4} \epsilon^{1/(640\log w_{\min}^{-1})})$.
\end{lemma}
\begin{proof}
Consider the iteration of the loop in which we select a cluster $S$ and a value $\epsilon$ that correspond to the assumptions.
We show that with high probability in this case the inner loop will produce a $(1-\epsilon')$-good refinement of $\mathcal{S}$.

Let $\mu_1, \mu_2$ be the means of any two components in $\mathsf{comp}(S)$, and let $v$ be a direction in which they are separated, i.e.,
\[\langle \mu_1 - \mu_2, v \rangle^2 \geq \Delta^2 v^\top \Sigma v\,.\]
By the approximate isotropic position, we have that $\langle \mu_1, v\rangle^2, \langle \mu_2, v\rangle^2 \leq 2w_{\min}^{-1}$, and therefore $v^\top \Sigma v \leq O(w_{\min}^{-1} \Delta^{-2})$.

Because $v^\top \Sigma v \leq O(w_{\min}^{-1} \Delta^{-2})$, we have that $v$ is $O(w_{\min}^{-1/2}\Delta^{-1})$-close to the subspace of eigenvectors of $\Sigma$ with eigenvalues at most $1/2$.
Then when we run the subspace recovery algorithm in~\Cref{thm:same-cov-subspace-rounding}, we obtain with high probability a subspace that contains a unit vector that is $O(w_{\min}^{-1/2}\Delta^{-1} + w_{\min}^{-1/8} \epsilon^{1/(640\log w_{\min}^{-1})})$-close to $v$, which is dominated by $O(w_{\min}^{-1/8} \epsilon^{1/(640\log w_{\min}^{-1})})$.
Then there also exists a unit vector $v'$ in the $\Delta^{-1}$-net that is $O(w_{\min}^{-1/8} \epsilon^{1/(640\log w_{\min}^{-1})})$-close to $v$.
For this $v'$, we have by~\Cref{lem:quadratic-form-closeness} that $(v')^\top \Sigma v' \leq O(w_{\min}^{-1/8} \epsilon^{1/(640\log w_{\min}^{-1})})$.

We argue now that in direction $v'$ there exist two means that are $\Omega(1)$-far.
We have because of the approximate isotropic position that
\[\sum_{i \in \mathsf{comp}(S)} w_i' \langle \mu_i, v'\rangle^2 + (v')^\top \Sigma v' \geq 1-\epsilon^{2}\,,\]
so there exists some $i \in \mathsf{comp}(S)$ such that $|\langle \mu_i, v'\rangle| \geq 0.9$. On the other hand, $\sum_{i \in \mathsf{comp}(S)} w_i' \langle \mu_i, v'\rangle \leq \epsilon^{2} \leq w_{\min}^{100}$, so there also exists some $j \in \mathsf{comp}(S)$ such that $\langle \mu_i, v'\rangle \langle \mu_j, v'\rangle \leq 0$, because otherwise the contribution of $\langle \mu_i, v'\rangle$ to the sum would be too large. Overall, this implies that $|\langle \mu_i - \mu_j, v'\rangle| \geq 0.9$.

Then, if we sort all components in increasing order by their means in direction $v'$, there exist two consecutive indices $i$ and $i+1$ such that $|\langle \mu_{i+1} - \mu_i, v'\rangle| \geq 1/(2k)$.
Also recall that $(v')^\top \Sigma (v') \leq O(w_{\min}^{-1/8} \epsilon^{1/(640\log w_{\min}^{-1})})$.
Consider partitioning the samples projected along direction $v'$ with the threshold $\tau = \frac{\langle \mu_{i}, v'\rangle + \langle \mu_{i+1}, v'\rangle}{2}$. Then we have by Chebyshev's inequality that an independent sample from the mixture is partitioned erroneously with independent probability at most $O(w_{\min}^{-3} \epsilon^{1/(640\log w_{\min}^{-1})})$.

Then, by Hoeffding's inequality, we have that with probability $\exp\Paren{-|S| \cdot O(w_{\min}^{-6} \epsilon^{1/(320\log w_{\min}^{-1})})}$ at most an $O(w_{\min}^{-3} \epsilon^{1/(640\log w_{\min}^{-1})})$-fraction of the samples are partitioned erroneously.
In the worst case in which one part contains only one component with weight $w_{\min}$, the fraction of outliers in it is at most $O(w_{\min}^{-4} \epsilon^{1/(640\log w_{\min}^{-1})})$.

A problem we omitted is that the bound on the probability that a sample is partitioned erroneously holds for a fixed direction independent of the samples, but $v'$ can depend on the samples.
We can show however by a union bound that the result holds for all $v'$.
Consider a $\poly(d^{-1}, \Delta^{-1})$-net of $\R^d$, which has size $\poly(d, \Delta)^d$.
The probability that a sample is partitioned erroneously in a direction that is $\poly(d^{-1}, \Delta^{-1})$-close to $v'$ is the same up to constants to the probability for direction $v'$, so the bound by Hoeffding's inequality is also essentially the same.
Then to union bound over all the directions in the net it suffices to have $|S| \geq \poly(d, \Delta)$.
    
Furthermore, we note that we obtain essentially the same probability of partitioning a sample erroneously even if the threshold is within an additive distance of $1/(100k)$ from the threshold $\tau$ we analyzed.
Also note that the threshold is trivially upper bounded by $2w_{\min}^{-1/2}$, because $|\langle \mu_i, v\rangle| \leq 2w_{\min}^{-1/2}$ for all $i \in [k]$.

Thus, when we loop over every unit vector in the net, we find a direction $v'$ for which we produce a $(1-\epsilon')$-good partial clustering for $\epsilon' = O(w_{\min}^{-4} \epsilon^{1/(640\log w_{\min}^{-1})})$.

Finally, we note that we do not need the exact value of $\epsilon$ and an upper bound within $2^{-\Delta}$ of the truth suffices.
\end{proof}

We also show that the refinement procedure produces a list of dimension-independent size.

\begin{lemma}[Partial clustering refinement produces a small list]
\label{lem:refinement-small-pp}
    For all partitions $\mathcal{S}$ of the sample universe such that $\vert \mathcal{S} \vert < k$, the size of the list returned by~\Cref{alg:pp-refinement} has size at most $f(w_{\min}^{-1}, \Delta)$ for some function $f$.
\end{lemma}
\begin{proof}
    We first consider the number of candidate refinements added to the list in a single iteration of the loop.
    We run the identical-covariance mixture subspace recovery algorithm and for each $v$ in a $\Delta^{-1}$-net of the outputted subspace we add at most $\poly(w_{\min})$ possible refinements to the output list.
    The dimension of the subspace outputted by the identical-covariance mixture subspace recovery algorithm in~\Cref{thm:same-cov-subspace-rounding} is at most some function of $\Delta$.
    Then, the size of the $\Delta^{-1}$-net of the subspace is also bounded by some function of $\Delta$. 
    Finally, the number of iterations of the loop is bounded by a function exponential in $\Delta$.
\end{proof}

\begin{lemma}[Correctness of partial clustering refinement]
\label{lem:pp-mean-refinement-correctness}
    Given a partial clustering $\mathcal{S}$ with $\vert \mathcal{S} \vert < k$, \Cref{alg:pp-refinement} outputs a list of partial clusterings $L$  with the following properties:
    \begin{enumerate}
        \item Every $\mathcal{S}' \in L$ is a refinement of $\mathcal{S}$,
        \item The size of $L$ is bounded by $f(w_{\min}^{-1}, \Delta)$ for some function $f$,
        \item If $\mathcal{S}$ is a $(1-\epsilon)$-good partial clustering for $\Delta^{-1} \leq \epsilon \leq w_{\min}^{O(\log w_{\min}^{-1})}$ with $|S| \geq \poly(d,\Delta)$ for all $S \in \mathcal{S}$, then the output of the algorithm contains with high probability at least one refinement $\mathcal{S}'$ that is a $(1-\epsilon')$-good clustering with $\epsilon' = O(w_{\min}^{-4} \epsilon^{1/(6000 \log w_{\min}^{-1})})$.
    \end{enumerate}
\end{lemma}
\begin{proof}
    The first property is immediate from the definition of the algorithm. The second property follows by~\Cref{lem:refinement-small-pp}.

    For the rest of the proof, suppose that $\mathcal{S}$ is a $(1-\epsilon)$-good partial clustering with $|S| \geq \poly(d,\Delta)$ for all $S \in \mathcal{S}$.

    The robust isotropic position algorithm corresponding to~\Cref{thm:robust-isotropic-position} succeeds with high probability and we can assume the mixture is in approximate isotropic position, with new mixture mean and covariance ${\Norm{\mu'} \leq \epsilon^{1/4}}$, ${(1-\epsilon^{1/4}) I_d \preceq \Sigma' \preceq (1+\epsilon^{1/4}) I_d}$, and ${\Norm{\Sigma' - I_d}_F \leq \epsilon^{1/4}}$.

    Because the size of $\mathcal{S}$ is less than $k$, for some cluster $S \in \mathcal{S}$ we have $|\mathsf{comp}(S)| > 1$.
    Recall that all the components are $\Delta$-separated, so for such a cluster with more than one component each $i \neq j \in \mathsf{comp}(S)$ have means that are $\Delta$-separated.

    Then we apply~\Cref{lem:pp-mean-sep} with the assumption that we have a $(1-\epsilon^{1/8})$-good partial clustering (so that the approximate isotropic position requirement of~\Cref{lem:pp-mean-sep} is satisfied), and conclude that with high probability the outputted list contains a ${(1-O(w_{\min}^{-4} \epsilon^{1/(6000 \log w_{\min}^{-1})}))}$-good clustering.
\end{proof}

\subsubsection{Proof of~\Cref{thm:same-cov-clustering}}

As in the proof of~\Cref{thm:zero-mean-main}, all partial clusterings that are contained in $\mathcal{C}$ in~\Cref{alg:pp-cluster} at any point are associated with a tree where the children of a partial clustering $\mathcal{S}$ are the partial clustering $\mathcal{S}'$ produced when we run the refinement algorithm on $\mathcal{S}$.
The number of clusters increases by one at each level and is bounded by $k$ and thus the depth is at most $k$.

For correctness, note that by~\Cref{lem:pp-mean-refinement-correctness} with high probability this tree contains at least one path from the root to a leaf such at each level if the node is a $(1-\epsilon)$-good partial clustering with $\epsilon \geq \Delta^{-1}$ then the next node in the path is an $(1-\epsilon')$-good partial clustering with $\epsilon' = O(w_{\min}^{-4} \epsilon^{1/(6000 \log w_{\min}^{-1})})$. Conditioning on the existence of a path from the root to level $i$, we can extend the path to level $i+1$ with high probability by~\Cref{lem:pp-mean-refinement-correctness}.
By taking a union bound over the $k-1$ steps of~\Cref{alg:pp-refinement} needed to extend the path from the root to the leaves, we have that this path will exist with high probability. 
We can assume we start with $\epsilon = \max\Set{O(\epsilon^*), \Delta^{-1}}$, such that with high probability the samples in $\mathcal{T}_1$ have at most an $\epsilon$-fraction of outliers.
Then, we have that at level $i$ the partial clustering is a $\left(1-O\Paren{w_{\min}^{-5} \epsilon^{1/(6000\log w_{\min}^{-1})^i}}\right)$-good clustering when $\epsilon^{-1} \geq f(w_{\min}^{-1})$ is large enough, which we lower bound by $1-\epsilon^{1/(\log w_{\min}^{-1})^{2k}}$ for $i \leq k$. 

Thus for $(\epsilon^*)^{-1} \geq f(w_{\min}^{-1})$ large enough there exists some leaf that corresponds to a candidate clustering in the final list $\mathcal{C}$ with few outliers, as required by~\Cref{claim:list-reduction}.
Then by~\Cref{claim:list-reduction} we have that we output a $(1-O(kw_{\min}^{-1}\max\{\epsilon, \Delta^{-\Omega(1/k)}\}))$-good clustering of the input samples.

The argument for the time complexity is similar to that of the proof of~\Cref{thm:zero-mean-main}, with the difference that now the subspace finding algorithm takes time $g(w_{\min}^{-1}, \Delta) \cdot \poly(d^{\log^2 w_{\min}^{-1}})$ for some function $g$.

$\qed$

\section*{Acknowledgments}
PA received funding from the Frank Quick Fellowship and MIT's Vice Chancellor's Inclusive Excellence Fellowship.
RB and DS received funding from the European Research Council (ERC) under the European Union’s Horizon 2020 research and innovation programme (grant agreement No 815464).
PK was supported by NSF CAREER Award \#2047933, NSF \#2211971, an Alfred P. Sloan Fellowship, and a Google Research Scholar Award.

\bibliographystyle{alpha}
\bibliography{references}

\appendix
\section{Clustering Selection}
We describe now an algorithm that takes input a list of clusterings such that one of them is ``good" and outputs a good clustering.

\begin{algorithm}[H]
    \SetAlgoLined
    \SetKwInOut{Input}{input}
    \SetKwInOut{Output}{output}

    \Input{a list of candidate clusterings $\mathcal{C}$ containing one good clustering, the number of components $k$, the minimum mixing weight $w_{\min}$, an upper bound $\epsilon$ on the fraction of corruptions in the best clustering, and a set of new samples from the original distribution such that the samples are independent of the list of candidate clusterings}
    \Output{a clustering $\mathcal{S}$ of all samples in the input}

    \For{each candidate clustering $\mathcal{S} \in \mathcal{C}$}{
        \For{each $S_i \in \mathcal{S}$}{
            Run the Gaussian estimation algorithm in~\Cref{fact:gaussian-robust-learning} on $S_i$ with fraction of outliers $\epsilon$ to produce estimates $\hat{\mu}_i$ and $\hat{\Sigma}_i$\;
            Estimate $\hat{w}_i$ as $|S_i|/\sum_{S \in \mathcal{S}} |S|$\;
        }
        Let $M_{\mathcal{S}} = \sum_{S_i \in \mathcal{S}} \hat{w}_i N(\hat{\mu}_i, \hat{\Sigma}_i)$\;
    }
    Run the tournament algorithm in~\Cref{fact:robust-tournament} on the distributions $M_{\mathcal{S}}$ for all clusterings $\mathcal{S} \in \mathcal{C}$, given the set of new samples\;
    Let $\hat{M}$ be the winning distribution in the tournament\;
    Cluster all samples in the input into a clustering $\mathcal{S}$ using the parameters of $\hat{M}$, as in the proof of Theorem 9.3 of~\cite{MR4490078-Ivkov22}\;
    \Return{$\mathcal{S}$}
\caption{Clustering selection algorithm}
\label{alg:list-filtering}
\end{algorithm}

\begin{lemma}[Clustering selection]
\label{claim:list-reduction}
Let $\mathcal{D}$ be a mixture $\sum_{i=1}^k w_i N(\mu_i, \Sigma_i)$ with $w_{\min} = \min_i w_i$.
Assume all components $i\neq j$ are $\Delta$-separated with $\Delta \geq f(w_{\min}^{-1})$ for some function $f$.
Also let $\mathcal{D}'$ be a distribution satisfying $d_{\TV}(\mathcal{D}', \mathcal{D}) \leq \epsilon^*$.
Let $\cC$ be a list of $m$ candidate clusterings of $n$ i.i.d. samples from $\mathcal{D}'$.
Suppose that all clusterings in $\cC$ have size $k$ and that there exists a clustering $\mathcal{S} \in \cC$ that is $(1-\epsilon)$-good.
Suppose you are also given an $\poly(m, \epsilon^{-1})$ new i.i.d. samples from $\mathcal{D}'$.
Suppose $\epsilon^*, \epsilon \leq g(w_{\min})$ for some function $g$.
Then, for $n \geq \poly(d, w_{\min}^{-1}, \epsilon^{-1})$, \Cref{alg:list-filtering} outputs with high probability a $(1-O\Paren{k w_{\min}^{-1} \max\{\epsilon^*, \Delta^{-\Omega(1/k)}\})}$-good clustering of all $n + \poly(m, \epsilon^{-1})$ samples.
\end{lemma}
\begin{proof}
Consider the clustering $\mathcal{S} \in \cC$ that is $(1-\epsilon)$-good.
Let $\mathcal{S} = \Set{S_1, \ldots, S_k}$.
Because $\mathcal{S}$ is $(1-\epsilon)$-good, for each $S_i$ we have that $|\mathsf{comp}(S_i)|=1$. Denote by $N(\mu_i, \Sigma_i)$ the Gaussian distribution corresponding to the component in $S_i$.
Then the points in $S_i$ can be considered an $O(\epsilon)$-corrupted set of i.i.d. samples from $N(\mu_i, \Sigma_i)$.
Then~\Cref{fact:gaussian-robust-learning} guarantees that for each $S_i$  the algorithm recovers some $\hat{\mu}_i$ and $\hat{\Sigma}_i$ such that
\[d_{\TV}(N(\mu_i, \Sigma_i), N(\hat{\mu}_i, \hat{\Sigma}_i)) \leq O(\epsilon \log^{3/2}(1/\epsilon))\,.\]
We also have by standard concentration bounds that, as long as the number of samples is a large enough polynomial in $w_{\min}^{-1}$, with high probability $|w_i - \hat{w}_i| \leq O(\epsilon)$.
Then the total variation distance between $M_{\mathcal{S}}$ and the ground truth mixture is at most $O(\epsilon \log^{3/2}(1/\epsilon))$.

Applying the hypothesis selection algorithm in~\Cref{fact:robust-tournament} on the hypotheses $M_{\mathcal{S}}$ computed for all clusterings $\mathcal{S} \in \cC$, we obtain with high probability some distribution $\hat{M}$ such that the total variation distance to the ground truth mixture is at most $O(\epsilon \log^{3/2}(1/\epsilon))$.

The parameter separation of the components of the ground truth also implies that the total variation distance between these components is, say, at least constant --- otherwise it would be \emph{information-theoretically} impossible to cluster the mixture, but we know that clustering is possible, e.g., because of the algorithm of~\cite{bakshi2020outlierrobust}.
Then by~\Cref{fact:param-identifiability}, for each component $\hat{G}_i$ of $\hat{M}$ there exists a component $G_i$ of the ground truth such that $d_{\TV}(G_i, \hat{G}_i) \leq \poly_k(\epsilon)$. 
This implies that $\hat{G}_i$ and $G_i$ also have parameter separation upper bounded by, say, an absolute constant (otherwise they could not be close in total variation distance).
Then clustering can be achieved by standard robust clustering techniques: see the proof of Theorem 9.3 in~\cite{MR4490078-Ivkov22}, assuming that a good set of $k$ parameters already exists. Note that their analysis implicitly assumes that, roughly, $\epsilon^* \geq \Delta^{-\Omega(1/k)}$, and we account for this in our guarantee.
\end{proof}

\end{document}